\documentclass{amsart}
\usepackage[centering]{geometry}
\usepackage[english]{babel}
\usepackage[latin1]{inputenc}
\usepackage[T1]{fontenc}
\usepackage{indentfirst,amsmath,amssymb,amscd,mathrsfs,bbold}
\usepackage[all]{xy}

\usepackage{calligra}
  \DeclareMathAlphabet{\mathcalligra}{T1}{calligra}{m}{n}  
  \DeclareFontShape{T1}{calligra}{m}{n}{<->s*[2.3]callig15}{}

\newcommand{\Av}{\mathrm{adv}}

\newcommand{\Bo}{\mathscr{B}}
\newcommand{\C}{\mathscr{C}}
\newcommand{\CC}{\mathbb{C}}
\newcommand{\D}{\mathscr{D}}
\newcommand{\dd}{\partial}

\newcommand{\E}{\mathscr{E}}

\newcommand{\Fun}{\mathscr{F}}

\newcommand{\JJ}{\mathscr{J}}
\newcommand{\id}{\mathbb{1}}
\newcommand{\Ko}{\mathscr{K}}

\newcommand{\Li}{\mathscr{L}}
\newcommand{\M}{\mathscr{M}}
\newcommand{\N}{\mathscr{N}}
\newcommand{\NN}{\mathbb{N}}
\newcommand{\OO}{\mathscr{O}}
\newcommand{\Po}{\mathscr{P}}

\newcommand{\Rt}{\mathrm{ret}}
\newcommand{\RR}{\mathbb{R}}
\newcommand{\Riem}{\mathrm{Riem}}

\newcommand{\UU}{\mathscr{U}}
\newcommand{\Un}{\mathfrak{U}}
\newcommand{\ud}{\mathrm{d}}
\newcommand{\VV}{\mathscr{V}}
\newcommand{\W}{\mathscr{W}}
\newcommand{\WF}{\mathrm{WF}}

\newcommand{\ZZ}{\mathbb{Z}}

\renewcommand{\Re}{\mathrm{Re}}
\renewcommand{\Im}{\mathrm{Im}}

\newcommand{\loc}{\mathrm{loc}}

\newcommand{\supp}{\mathrm{supp}\ \!}

\newcommand{\pr}{\mathrm{pr}}

\newcommand{\To}{\rightarrow}

\newcommand{\If}{\Leftarrow}
\newcommand{\Iff}{\Leftrightarrow}
\newcommand{\Then}{\Rightarrow}
\newcommand{\sm}{\smallsetminus}

\newcommand{\gotoas}[3]{\stackrel{{#2}\rightarrow{#3}}{\longrightarrow}{#1}}

\newcommand{\dlim}[1]{\underrightarrow{\lim}_{#1}}
\newcommand{\ilim}[1]{\underleftarrow{\lim}_{#1}}

\newcommand{\spr}[1]{\langle{#1}\rangle}
\newcommand{\Spr}[1]{\left\langle{#1}\right\rangle}

\newcommand{\In}[1]{\mathring{#1}}

\newcommand{\ol}[1]{\overline{#1}}

\newcommand{\restr}[1]{\vert_{#1}}
\newcommand{\Restr}[1]{\left.\vphantom{\frac{}{}}\!\right\vert_{#1}}

\newcommand{\Ret}{\mathsf{R}}

\newcommand{\Adv}{\mathsf{A}}

\let\oldtocsection=\tocsection

\let\oldtocsubsection=\tocsubsection

\let\oldtocsubsubsection=\tocsubsubsection

\renewcommand{\tocsection}[2]{\hspace{0em}\oldtocsection{#1}{#2}}
\renewcommand{\tocsubsection}[2]{\hspace{1em}\oldtocsubsection{#1}{#2}}
\renewcommand{\tocsubsubsection}[2]{\hspace{2em}\oldtocsubsubsection{#1}{#2}}

\begin{document}

%

\swapnumbers
\newtheorem{theorem}{\mdseries\scshape Theorem}[subsection]
\newtheorem{lemma}[theorem]{\mdseries\scshape Lemma}
\newtheorem{proposition}[theorem]{\mdseries\scshape Proposition}
\newtheorem{corollary}[theorem]{\mdseries\scshape Corollary}
\newtheorem{scholium}[theorem]{\mdseries\scshape Scholium}

\theoremstyle{definition}
\newtheorem{definition}[theorem]{\mdseries\scshape Definition}

\theoremstyle{remark}
\newtheorem{remark}[theorem]{\mdseries\scshape Remark}
\newtheorem{conjecture}[theorem]{\mdseries\scshape Conjecture}

%

\title[Algebraic Structure of Classical Field Theory]{Algebraic Structure 
  of Classical Field Theory: \\ Kinematics and Linearized Dynamics \\
  for Real Scalar Fields}

\author{Romeo Brunetti}
\address{Dipartimento di Matematica, Università di Trento -- 
  Via Sommarive 14, I-38053 Povo (TN), Italy}
\email{romeo.brunetti@unitn.it}

\author{Klaus Fredenhagen}
\address{II. Institut für theoretische Physik, Universität Hamburg -- 
  Luruper Chaussee 149, D-22761 Hamburg (HH), Germany}
\email{klaus.fredenhagen@desy.de}

\author{Pedro Lauridsen Ribeiro}
\address{Centro de Matemática, Computação e Cognição, Universidade Federal do ABC (UFABC) -- 
  Avenida dos Estados 5001 -- 09210-580 Santo André (SP), Brazil}
\email{pedro.ribeiro@ufabc.edu.br}

\date{\today}

\begin{abstract}
We describe the elements of a novel structural approach to classical field 
theory, inspired by recent developments in perturbative algebraic quantum 
field theory. This approach is local and focuses mainly on the observables 
over field configurations, given by certain spaces of functionals which are 
studied here in depth. The analysis of such functionals is characterized by a 
combination of geometric, analytic and algebraic elements which (1) make our 
approach closer to quantum field theory, (2) allow for a rigorous analytic 
refinement of many computational formulae from the functional formulation of 
classical field theory and (3) provide a new pathway towards understanding 
dynamics. Particular attention will be paid to aspects related to nonlinear 
hyperbolic partial differential equations and their linearizations.
\end{abstract}

\maketitle




\tableofcontents

\section{\label{s1-intro}Introduction}

The longstanding problem of finding a coherent and systematic mathematical 
structure for classical field theories has been addressed in various ways. 
Among them, we quote two main lines of investigation: one based on 
(multi)\-sym\-plec\-tic geometry \cite{carci,forgerr,gotay,kijowski}, seeking 
a covariant generalization of Hamiltonian mechanics and that goes back to de 
Donder \cite{dedonder} and Weyl \cite{weyl}; and the other based on the 
so-called formal theory of systems of partial differential equations 
\cite{anderson,kralv,seiler,vino1,vino2}, seeking a higher-order generalization 
of S. Lie's and É. Cartan's geometric approach to the analysis of integrability 
and symmetries of such systems. Both approaches have several points of contact 
and lead to a highly developed framework for the calculus of variations. As far 
as \emph{relativistic} field theories are concerned, however, the solution 
spaces of the dynamics generated by the variational principle are essentially 
taken for granted and their properties are seldom studied in depth, a 
noteworthy exception being the approach of Christodoulou \cite{christo}. 

Physicists, on the other hand, are keen on formal functional methods 
\cite{crnwit,dewitt}, tailored to the needs of (path-integral-based) 
quantum field theory, which are essentially a heuristic infinite-dimensional 
generalization of Lagrangian mechanics. To a certain extent, it is possible
to make these latter methods rigorous (see for instance \cite{abraham,binz}).
However, in these approaches the field configuration spaces are usually modeled 
on Banach spaces, which provide a simple differential calculus but entail some 
physically undesirable restrictions on the allowed space-times and on the 
regularity of the allowed field configurations. Moreover, these approaches 
also tend to deemphasize aspects related to covariance and locality, which 
are central in any relativistic field theory since then Euler-Lagrange equations
of motion are \emph{differential} (expressing locality of the underlying 
variational principle) and \emph{hyperbolic} (expressing finiteness of the 
propagation speed of dynamical effects). 

Even more importantly, a pivotal aspect that none of the above approaches
has addressed in a satisfactory manner is the characterization of \emph{local 
observables}, as opposed to spaces of field configurations. This remark is 
the starting point of our present investigation. Namely, we contend that if 
one wants to study the structure of local observables in a \emph{model-independent} 
fashion, one is inevitably led to an \emph{algebraic} viewpoint. This is a 
deep lesson learned from quantum field theory \cite{haag}, which however does 
not seem to have echoed back to classical field theory until quite recently, 
the only exception to our knowledge being \cite{leyrob}. This state of things 
has started to change due to the recent developments in perturbative algebraic 
quantum field theory \cite{brdut,brudf,brufre1,brufre2,bfk,dutfre1,dutfre2}. 
This is a research program aiming at a mathematically precise understanding of 
perturbative quantum field theory and renormalization from an algebraic
viewpoint -- to wit, renormalized perturbative quantum field theory can be seen
as a formal deformation of classical field theory, in a rather precise sense
\cite{brdut,brudf,dutfre1}.  

The key upshot of this program, which motivated the present work, is that it 
singles out the relevant class of observables for classical field theory from 
a few, physically reasonable requirements which, at the quantum level, are 
needed to restrict the class of allowed counterterms in renormalization. This 
serves as a starting point for a new, algebraic framework for classical 
(relativistic) field theory in its own right, which emphasizes from the very 
beginning the role of local observables and how they are affected by the 
dynamics. Presenting this framework in full detail is the objective of 
this paper. Let us now give an overview of its results.

As we shall see, local observables are represented by certain classes of 
functionals over the space of smooth field configurations. More precisely, 
the kinematical requirements on functionals in order to qualify as local 
observables lead, among other things, to a surprisingly simple structure
 for the local algebras they generate -- for instance, these algebras, 
when suitably topologized, turn out to be \emph{nuclear}, opening the way 
to a seamless composition of classical subsystems by means of tensor 
products \cite{brufre2}. 

A cornerstone of our approach concerns the treatment of \emph{dynamics}. We
do not impose any equations of motion directly on field configurations -- that 
is, we adopt an \emph{off-shell} viewpoint. We show that, on an infinitesimal 
level, the dynamics is implemented algebraically on local observables by means 
of a Poisson structure associated to certain Lagrangians, given in covariant 
form by the Peierls bracket \cite{dewitt,forgerr,marolf,peierls}. This bracket 
is a covariant generalization of the canonical Poisson bracket \cite{binz,wald1}, 
and has an unambiguous off-shell extension which however becomes degenerate. 
This degeneracy can be removed by taking the quotient of our local Poisson 
algebras of functionals modulo the ideal generated by the equations of motion, 
which turns out to be a \emph{Poisson} ideal. As a consequence, the quotient 
algebra is a Poisson algebra as well when endowed with the bracket induced on
the quotient by the Peierls bracket. The quotient amounts to imposing the 
equations of motion on field configurations pretty much in the spirit of algebraic 
geometry, and allows for a unified analysis of quantum anomalies as violations 
of identities following from the classical equations of motion due to
perturbative quantization and renormalization \cite{brdut,dutfre1}.

We conclude this introduction with a summary of the contents of the paper.
In Section \ref{s2-kin}, we discuss the bare minimum of kinematical concepts 
underlying our approach. For simplicity, we will consider only \emph{real scalar} 
fields, since the case when the fields live in a general fiber bundle poses 
a different set of questions, which demand a separate treatment (we shall 
have more to say about this in the final Section \ref{s5-ciao}). In Subsection 
\ref{s2-kin-pre}, we present a fair amount of background on Lorentzian geometry, 
vector bundles and jets, which is also used in Subsection \ref{s2-kin-geom}
to give an overview of the geometric and topological properties of the space of 
smooth field configurations. In Subsection \ref{s2-kin-obs} we introduce 
suitable classes of functionals over this space and discuss their support and 
localization properties, so as to be able to proceed to a detailed analysis of 
infinitesimal (i.e. linearized) dynamics in Section \ref{s3-dyn}; the full 
nonlinear dynamics is to be analyzed in a forthcoming paper. Euler-Lagrange 
equations are obtained from a class of local functionals parametrized by smooth, 
compactly supported functions $f$ specifying the localization of these functionals 
in space-time. Such functionals are called \emph{generalized Lagrangians}, examples 
of which are provided by integrals of Lagrangian densities multiplied by $f$ 
over the space-time manifold (Subsection \ref{s3-dyn-lag}). We are mainly 
interested in those generalized Lagrangians which lead to (normally) hyperbolic 
Euler-Lagrange operators, which are discussed in Subsection \ref{s3-dyn-hyp}. 
Therein we also define the Peierls bracket associated with such operators, and 
study its properties in depth. This bracket is shown to yield a Lie bracket in 
the space of so-called \emph{microcausal} functionals, which are distinguished 
by the singularity structure of their functional derivatives. A particular 
highlight of this development is perhaps the first fully fledged and rigorous 
proof of the Jacobi identity for the Peierls bracket in the literature (Corollary 
\ref{s3c3}), parts of which having previously appeared or been sketched in 
\cite{brdut,dutfre1,jakobs}. A thorough discussion of the topological and 
algebraic aspects of the *-algebras of microcausal functionals is carried 
out in Section \ref{s4-gen}, using the previous Sections as motivation. 
We show in Subsection \ref{s4-gen-top} that the Lie bracket provided by the 
Peierls bracket is in fact a Poisson bracket; another noteworthy result, shown
in Subsection \ref{s4-gen-cr}, is that the (Poisson) *-algebras of microcausal 
functionals also bear a $\C^\infty$\emph{-ring} structure \cite{moerr}, that is, 
they admit a sort of smooth functional calculus (Theorem \ref{s4t2}), which 
leads to a number of interesting consequences. For example, one recovers some 
basic facts from commutative C*-algebra theory: the *-algebra of microcausal 
functionals over a domain of field configurations completely encodes the topology 
of this domain (Proposition \ref{s4p1} (i)) and one may even reconstruct the domain 
itself as the space of *-characters of the *-algebra (Proposition \ref{s4p1} 
(iii--iv)). Moreover, any open cover of the domain admits locally finite 
partitions of unity whose members belong to this *-algebra (Proposition \ref{s4p1} 
(ii)). Finally, in Subsection \ref{s4-gen-em} we show that the ideal generated
by a hyperbolic Euler-Lagrange equation is a Poisson *-ideal (Proposition \ref{s4p2})
and therefore the quotient of the Poisson *-algebra of microcausal functionals
modulo this ideal is again a Poisson *-algebra. Section \ref{s5-ciao} concludes 
our work by presenting some future prospects and challenges. Appendix \ref{a1-calc} 
recalls basic concepts of differential calculus on locally convex topological vector 
spaces. 

\section{\label{s2-kin}Kinematics}

\subsection{\label{s2-kin-pre}Preliminaries}

Given nonvoid sets $A,A_1,\ldots,A_m$, we denote by $\id=\id_A:A\To A$ the \emph{identity map}
$\id_A(a)=a$, and by $\pr_{j_1,\ldots,j_k}:A_1\times\cdots\times A_m\To A_{j_1}\times\cdots\times
A_{j_k}$ the \emph{canonical projection} $\pr_{j_1,\ldots,j_k}(a_1,\ldots,a_m)=(a_{j_1},\ldots,a_{j_k})$,
$1\leq j_1<j_2<\cdots<j_k\leq m$. If $k=1$, we say that $\pr_j$ is the \emph{canonical projection 
  onto the $k$-th factor}.

First of all, a small refresher on Lorentzian geometry to fix our notation and terminology
(we basically follow \cite{hawkellis,wald1}). Let $(\M,g)$ be a \emph{space-time}, that is, 
an oriented $d$-dimensional Lorentzian manifold. The underlying manifold $\M$ (called the 
\emph{space-time manifold}) is assumed to be smooth, Hausdorff, paracompact and second 
countable (in particular, $\M$ has at most a countable number of connected components). By 
a \emph{region} of $\M$ (or of $(\M,g)$) we mean any subset of $\M$ with nonvoid interior. 
The Lorentzian metric $g$ on $T\!\M$ endows $\M$ with the volume element $\ud\mu_g=
\sqrt{|\det g|}\ud x$, the Levi-Civita connection $\nabla$, the lowering (resp. raising) 
musical isomorphisms $g^\flat:T\!\M\To T^*\!\!\M$ (resp. $g^\sharp:T\!\M\To T^*\!\!\M$) given 
by $g^\flat(X)\doteq g(X,\cdot)$ (resp. $g^\sharp(\xi)\doteq(g^\flat)^{-1}(\xi)$), and the 
inverse Lorentzian metric $g^{-1}$ on $T^*\!\!\M$ given by $g^{-1}(\xi_1,\xi_2)\doteq\xi_1
(g^\sharp(\xi_2))$. We occasionally write $g(T)$ (resp. $g^{-1}(\omega)$) with a single 
argument $T$ (resp. $\omega$), which is understood to be a contravariant (resp. covariant) 
tensor of rank two. We will use the chosen orientation to identify smooth densities with 
smooth $d$-forms. We adopt for $g$ the signature convention that, for all $p\in\M$, the 
subspace of $T_p\M$ consisting of eigenvectors of $g(p)$ with \emph{negative} eigenvalues 
is one-dimensional and therefore consists of \emph{timelike} vectors. Recall that $X\in
\ T_p\M$ is timelike (resp. null, causal, spacelike) if $g(X,X)<0$ (resp. $=0$, $\leq 0$, 
$>0$) -- hence, the subspace of $T_p\M$ consisting of (spacelike) eigenvectors of $g(p)$ 
with positive eigenvalues is $(d-1)$-dimensional. We always assume that $\M$ is 
\emph{time-oriented}, that is, there is a global timelike vector field $T$ on $\M$ -- we 
then say that a causal $X\in T_p\M$ is future (resp. past) directed if $g(X,T)<0$ (resp. 
$>0$). 

Recall as well that, given an interval $I\subset\RR$ with nonvoid interior, a 
(piecewise) smooth curve $\gamma:I\ni\lambda\To\gamma(\lambda)\in\M$ is said to be 
timelike (resp. null, causal, spacelike) if $g(\dot{\gamma}(\lambda),\dot{\gamma}
(\lambda))<0$ (resp. $=0$, $\leq 0$, $>0$) for all $\lambda\in I$ (such that $\gamma$ 
is smooth at $\lambda$), and that a causal curve is said to be future (resp. past) 
directed if $g(\dot{\gamma}(\lambda),T)<0$ (resp. $>0$) for any $\lambda\in I$ as 
above and any future directed timelike $T\in T_{\gamma(\lambda)}\M$. This allows us to 
define the chronological (resp. causal) future / past $I^{+/-}(U,g)$ (resp. $J^{+/-}(U,g)$) 
of $U\subset\M$ as
\[
\begin{split}
  I^{+/-}(U,g) \doteq\{ &p\in\M:\exists\gamma:[0,1]\To\M\text{ piecewise smooth, future / 
    past directed} \\ & \text{timelike such that }\gamma(0)\in U,\gamma(1)=p\}\ ,
\end{split}
\]
\[
\begin{split}
  J^{+/-}(U,g) \doteq\{ &p\in\M:\exists\gamma:[0,1]\To\M\text{ piecewise smooth, future / 
    past directed} \\ & \text{causal such that }\gamma(0)\in U,\gamma(1)=p\}\ .
\end{split}
\]
We also set $I^{+/-}(\{p\},g)\doteq I^{+/-}(p,g)$ (resp. $J^{+/-}(\{p\},g)\doteq J^{+/-}(p,g)$)
for any $p\in\M$, and, given $U,V\subset\M$, we write $U\gg_g / \ll_g V$ (resp. 
$U\geq_g / \leq_g V$) whenever $U\subset I^{+/-}(V,g)$ (resp. $U\subset J^{+/-}(V,g)$).
If $U=\{p\}$ (resp. $V=\{q\}$) for some $p,q\in\M$, we replace $U$ (resp. $V$) by
$p$ (resp. $q$) in the above notation. Finally, we always assume that $g$ is 
\emph{globally hyperbolic}, that is, $g$ is \emph{causal} (which means that there is 
no causal $\gamma:[0,1]\To\M$ such that $\gamma(0)=\gamma(1)$) and given $p\leq_g q\in\M$,
the set $J^+(p,g)\cap J^-(q,g)$ is \emph{compact}. An useful, equivalent description 
of global hyperbolicity can be given as follows \cite{bernsan1,bernsan2,bernsan3,bernsan4}:
there is a smooth, surjective function $\tau:\M\To\RR$ such that $g^\sharp(\ud\tau)$ is 
a future directed timelike vector field and $\Sigma^\tau_t\doteq\tau^{-1}(t)$ is a
\emph{Cauchy hypersurface} for $\M$ at each $t\in\RR$, that is, $\Sigma^\tau_t$
is a codimension-one, smooth and boundary-less submanifold of $\M$ such that any 
inextendible causal curve\footnote{A causal (resp. timelike, null) curve $\gamma:I\To\M$ 
  is said to be \emph{inextendible} if there is no causal (resp. timelike, null) curve
  $\tilde{\gamma}:\tilde{I}\To\M$ such that $\tilde{I}\supsetneqq I$ and $\tilde{\gamma}
  \restr{I}=\gamma$.} intersects $\Sigma^\tau_t$ exactly once. Such a $\tau$ is 
called a \emph{Cauchy time function} with respect to $(\M,g)$. Moreover, if $(\M,g)$ 
has a Cauchy hypersurface $\Sigma$, one can build a Cauchy time function $\tau$ such 
that $\tau^{-1}(0)=\Sigma$ \cite{bernsan4} -- in particular, $\M$ must then be 
diffeomorphic to $\RR\times\Sigma\cong\RR\times\Sigma^\tau_t$ for any $t\in\RR$.

Occasionally, we will need to work with smooth sections of vector bundles over
the space-time manifold $\M$ or over Cartesian powers thereof. Recall, for the 
sake of fixing nomenclature, that a (real) \emph{vector bundle} of rank $D$ 
over $\M$ is given by a smooth surjective submersion $\pi:\E\To\M$ from the 
\emph{total space} $\E$ to the \emph{base} $\M$, called the \emph{projection map}, 
such that there is an open covering $\{U_j\}_{j\in J}$ of $\M$ and for each 
$j\in J$ a smooth diffeomorphism $\psi_j:\pi^{-1}(U_j)\To U_j\times\RR^D$ 
(called a \emph{local trivialization} over $U_j$) such that $\psi_k\circ
\psi_j^{-1}(x,\zeta)=(x,t_{kj}(x,\zeta))=(x,T_{kj}(x)\zeta)$ for all 
$x\in U_j\cap U_k$, $j,k\in J$, where the \emph{transition functions} 
$T_{kj}:U_j\cap U_k\To GL(D,\RR)$ are smooth. The collection of pairs 
$\{(U_j,\psi_j)\}_{j\in J}$ is called a \emph{vector bundle atlas} for $\pi$. 
We usually identify a vector bundle with its projection map. Given $U\subset\M$ 
open, we say that a local trivialization $\psi$ over $U$ is said to be 
$\pi$\emph{-compatible} if for every $j\in J$ such that $U\cap U_j\neq\varnothing$ 
we have that $\psi\circ\psi_j^{-1}(x,\zeta)=(x,t_j(x,\zeta))=(x,T_j(x)\zeta)$ 
where $T_j:U_j\cap U\To GL(D,\RR)$ is smooth. A map $\vec{\varphi}:\M\To\E$ is 
said to be a \emph{section} of $\pi$ if $\pi\circ\vec{\varphi}=\id_\M$. Notice 
that if, in the above discussion, we replace $\RR^D$ by a manifold $Q$, and just 
demand that the smooth maps $t_{kj}$ are diffeomorphisms of $Q$ for each fixed 
$x\in U_j\cap U_k$ and the smooth maps $t_j$ are diffeomorphisms of $Q$ for 
each fixed $x\in U\cap U_j$ and $\pi$-compatible local trivialization $\psi$, 
$j,k\in J$, we get instead a (general) fiber bundle with typical fiber $Q$ and 
bundle atlas $\{(U_j,\psi_j)\}$.

Using a vector bundle atlas one can define (fiberwise) linear combinations 
$\alpha\vec{\varphi}_1+\beta\vec{\varphi}_2$ of any two sections $\vec{\varphi}_1,
\vec{\varphi}_2$ ($\alpha,\beta\in\RR$) by setting $\psi_j\circ(\alpha\vec{\varphi}_1
+\beta\vec{\varphi}_2)(p)=\alpha\psi_j\circ\vec{\varphi}_1(p)+\beta\psi_j\circ
\vec{\varphi}_2(p)$, $p\in U_j$, $j\in J$. This definition is readily seen to be 
independent of the choice of vector bundle atlas with $\pi$-compatible local 
trivializations. In particular, every vector bundle $\pi$ over $\M$ has a canonical 
section $0$ (called the \emph{zero section} of $\pi$), defined on every local 
trivialization $\psi$ compatible with $\pi$ by $\psi\circ 0(p)=(p,0)$, and
with respect to which we can define the \emph{support} of a section $\vec{\varphi}$
as $\supp\vec{\varphi}=\ol{\{p\in\M:\vec{\varphi}(p)\neq 0(p)\}}\subset\M$. 
It follows from the inverse function theorem that $\M$ is diffeomorphic 
to the range of the zero section in $\E$, which we also denote by $0$.
We denote by
\[
\Gamma^\infty(\pi)=\Gamma^\infty(\E\To\M)=\{\vec{\varphi}:\M\To\E\text{ smooth }\ |\ 
\pi\circ\vec{\varphi}=\id_\M\}
\]
the vector space of \emph{smooth sections} of $\pi$, and by
\[
\Gamma^\infty_c(\pi)=\Gamma^\infty_c(\E\To\M)=\{\vec{\varphi}\in\Gamma^\infty(\pi)\ |\ 
\supp\vec{\varphi}\text{ compact}\}
\]
the vector space of smooth sections of $\pi$ with compact support. Likewise, 
we denote by 
\[
\D'(\pi)=\D'(\E\To\M)=\Gamma^\infty_c(\E'\otimes\wedge^d T^*\!\!\M\To\M)'
\]
the space of $\E$-valued distributions, where $\pi':\E'\To\M$ is the dual bundle of $\pi$.
The fiberwise scalar multiplication turns $\Gamma^\infty(\pi)$, $\Gamma^\infty_c(\pi)$ and
$\D'(\pi)$ into $\C^\infty(\M)$-modules, so that multiplication of sections by 
$f\in\C^\infty_c(\M)$ is even a $\C^\infty(\M)$-linear map from $\Gamma^\infty(\pi)$ into 
$\Gamma^\infty_c(\pi)$, for $\supp(f\vec{\varphi})\subset\supp f$ for all $f\in\C^\infty(\M)$, 
$\vec{\varphi}\in\Gamma^\infty(\pi)$.

We also briefly recall the notion of \emph{jets} of smooth maps between manifolds
$\M,\M'$ of respective dimensions $d,D$, referring to \cite{kms} for a thorough 
exposition. Let $r\in\NN$; we say that two smooth maps $\psi_1,\psi_2:\M\To\M'$ 
\emph{have the same} $r$\emph{-th order jet} at $p\in\M$ if for some (hence, any) 
coordinate charts $x:U\supset p\To\RR^d$, $y:V\supset\psi_1(p),\psi_2(p)\To\RR^D$, 
the $r$-th order Taylor polynomials of $y\circ\psi_1\circ x^{-1}$ and $y\circ
\psi_2\circ x^{-1}$ at $x(p)$ coincide. Having the same $r$-th order jet at $p\in\M$ 
is clearly an equivalence relation in the space $\C^\infty(\M,\M')$ of all smooth 
maps from $\M$ into $\M'$, and the equivalence class of $\psi\in\C^\infty(\M,\M')$ 
is called the $r$\emph{-th order jet} of $\psi$ at $p$, denoted by $j^r\psi(p)$. 
The $r$\emph{-th order jet bundle} of $\C^\infty(\M,\M')$, given by
\[
\pi^r_0:J^r(\M,\M')\ni j^r\psi(p)\mapsto\pi^r_0(j^r\psi(p))=(p,\psi(p))\in
\M\times\M'\ ,\quad\psi\in\C^\infty(\M,\M')\ ,
\]
is an affine bundle over $\M\times\M'$, whose typical fiber is the space 
of $r$-th order, $\RR^D$-valued polynomials vanishing at $0\in\RR^d$. 
Given $\psi\in\C^\infty(\M,\M')$, the corresponding section $j^r\psi:p\mapsto 
j^r\psi(p)$ of $\pi^r_0$ is called the $r$\emph{-th order jet prolongation} 
of $\psi$. Truncation of $r$-th order Taylor polynomials to order $1\leq s<r$ 
induces surjective submersions $\pi^r_s:J^r(\M,\M')\To J^s(\M,\M')$ which 
satisfy $\pi^r_r=\id$ and $\pi^s_t\circ\pi^r_s=\pi^r_t$ for all $0\leq t
\leq s\leq r$, which allow one to define the projective limit $\pi^\infty_0:
J^\infty(\M,\M')\To\M\times\M'$, called the \emph{infinite-order jet bundle}
of $\C^\infty(\M,\M')$. One can then identify the sequence $(j^r\psi)_{r\geq 0}$ 
of jet prolongations with a section $j^\infty\psi$ of $J^\infty(\M,\M')$, called 
simply the \emph{infinite-order jet prolongation} of $\psi$. $J^\infty(\M,\M')$, 
being a countable projective limit of second-countable, finite-dimensional 
manifolds, can be made into a second-countable, metrizable Fréchet manifold 
\cite{km}. If $\pi:\E\To\M$ is a fiber bundle over $\M$, we can define the subspace 
$J^r(\pi)\subset J^r(\M,\E)$ of $r$-jets $X=j^r\psi(p)$ of smooth \emph{sections} 
$\psi$ of $\pi$ (i.e. smooth maps from $\M$ to $\E$ satisfying $\pi\circ\psi=\id_\M$), 
$1\leq r\leq\infty$. Then we can identify $\pi^r_0\restr{J^r(\pi)}$ with $\pr_2\circ
\pi^r_0$, and we call the affine bundle $\pi^r_0:J^r(\pi)\To\E$ the $r$\emph{-th 
  order jet bundle} of $\pi$. 

\subsection{\label{s2-kin-geom}Topology and geometry of the space of field configurations}

Let $(\M,g)$ be a globally hyperbolic space-time, and $\C^\infty(\M)\doteq\C^\infty(\M,\RR)$ 
be the space of real-valued smooth functions on $\M$. We call $\C^\infty(\M)$ a(n 
\emph{off-shell}) \emph{space of (real scalar) field configurations}\footnote{Some 
  physics texts, such as \cite{dewitt}, call $\C^\infty(\M)$ the \emph{space of field histories} 
  on $\M$.}. It can be topologized in two different ways by means of the infinite-order 
jet prolongation of its elements, as follows. Let $\C(\M,J^\infty(\M,\RR))$ be the 
space of continuous functions from $\M$ into $J^\infty(\M,\RR)$. The \emph{compact-open} 
topology on $\C(\M,J^\infty(\M,\RR))$ is generated by the sub-basis
\[
\UU_{K,V}=\{X\in\C(\M,J^\infty(\M,\RR))\ |\ X(K)\subset V\}\ ,
\]
for all $K\subset\M$ compact, $V\subset J^\infty(\M,\RR)$ open. The initial topology
on $\C^\infty(\M)\ni\varphi$ induced by the compact-open topology on $\C(\M,J^\infty(\M,\RR))$
through the map $\varphi\mapsto j^\infty\varphi$ is also called the \emph{compact-open} topology
on $\C^\infty(\M)$. The \emph{graph} (or \emph{Whitney}) topology on $\C(\M,J^\infty(\M,\RR))$, 
on its turn, is given by taking
\[
\UU_W=\{X\in\C(\M,J^\infty(\M,\RR))\ |\ (p,X(p))\in W\text{ for all }p\in\M\}\ ,
\]
for all $W\subset\M\times J^\infty(\M,\RR)$ open in the product topology, as a basis 
of open sets. Obviously, to have $\UU_W\neq\varnothing$ one needs $W$ to satisfy 
$\pr_1(W)=\M$. As $J^\infty(\M,\RR)$ is metrizable and $\M$ is paracompact, another 
basis for this topology is given around any $Y\in\C(\M,J^\infty(\M,\RR))$ by 
$\{X\in\C(\M,J^\infty(\M,\RR))\ |\ d(X(p),Y(p))<\epsilon(p)\}$, for all positive 
$\epsilon\in\C(\M,\RR)$. The initial topology on $\C^\infty(\M)\ni\varphi$ induced by 
the graph topology on $\C(\M,J^\infty(\M,\RR))$ through the single map $\varphi\mapsto 
j^\infty\varphi$ is called the \emph{Whitney} topology on $\C^\infty(\M)$. It is in 
general finer than the compact-open topology, and coincides with the latter if and 
only if $\M$ is compact, which is \emph{not} our case. On the other hand, notice that 
since $\M$ is locally compact (for $\M$ is finite dimensional) and second countable, 
we have that $\M$ is $\sigma$\emph{-compact}, that is, $\M$ admits a so-called 
\emph{exhaustion} by a sequence $K_n\subset\In{K}_{n+1}$ of compact regions
$K_n\subset\M$, which means that $\cup^\infty_{n=1}K_n=\M$. We can then use any exhaustion 
$(K_n)_{n\geq 1}$ of $\M$ to show that any set $\UU_W$ as above must be a $G_\delta$ 
\emph{set} (i.e. a countable intersection of open sets) in the compact-open topology 
of $\C(\M,J^\infty(\M,\RR))$. Indeed, we have that
\[
\UU_W=\bigcap^{\infty}_{n=1}\UU_{K_n,\pr_2(W)}\ ,
\]
where $\pr_2$ is an open mapping. Therefore, the Whitney topology on $\C^\infty(\M)$ 
admits a basis made of $G_\delta$ subsets of $\C^\infty(\M)$ in the compact-open topology. 

The compact-open topology on $\C^\infty(\M)\ni\varphi$ can be understood as the topology
of uniform convergence of derivatives of all orders $k\geq 0$ on compact regions $K\subset\M$, 
as induced by the seminorms
\begin{equation}\label{s2e1}
  \begin{split}
    \|\varphi\|_{\infty,k,K} &\doteq\sup_{p\in K}\sqrt{\sum^k_{j=0}|\nabla^j\varphi(p)|^2_e}\ ,\\
    |\nabla^j\varphi|^2_e &\doteq \otimes^j e^{-1}(\nabla^j\varphi,\nabla^j\varphi)\ ,\\
  \end{split}
\end{equation}
where $\otimes^j e^{-1}$ is the Riemannian metric induced on the bundle $\otimes^j 
T^*\!\!\M$ of covariant tensors of rank $j$ on $\M$ by a Riemannian metric $e$ on 
$T\!\M$, and $\nabla^j\varphi$ is the iterated covariant derivative of order $j$ of 
$\varphi$ with respect to a torsion-free connection $\nabla$ on $T\!\M$, given 
recursively by 
\begin{equation}\label{s2e2}
  \begin{split}
    \nabla^1\varphi=\nabla\varphi=\ud\varphi\ ,\\\nabla^j\varphi(X_1,\ldots,X_j)
    &=\nabla_{X_1}\nabla^{j-1}\varphi(X_2,\ldots,X_j)\\ &-\sum^j_{l=2}\nabla^{j-1}
    \varphi(X_2,\ldots,X_{l-1},\nabla_{X_1}X_l,X_{l+1},\ldots,X_j)\ .
  \end{split}
\end{equation}
A countable family of seminorms is obtained by exploiting the $\sigma$-compactness
of $\M$ and choosing an exhaustion $(K_n)_{n\geq 1}$ of $\M$ as above. The topology 
induced by the seminorms $\|\cdot\|_{\infty,k,K_n}$ is then independent of the choice 
of $e$, $\nabla$ and the exhaustion $(K_n)_{n\in\NN}$. It is clearly a vector space 
topology with respect to the standard vector space operations in a space of vector 
bundle sections, and gives rise to a Fréchet space structure on $\C^\infty(\M)$. An 
equivalent, separating family of seminorms generating this topology is given by 
\begin{equation}\label{s2e3}
  \|\varphi\|_{\infty,k,f}\doteq\sup_{p\in\M}\sqrt{\sum^k_{j=0}|f(p)\nabla^j\varphi(p)|^2_e}\ ,
\end{equation}
where $f$ runs over the space $\C^\infty_c(\M)\doteq\C^\infty_c(\M,\RR)$ of real-valued
smooth functions with compact support. To see the equivalence, let $(f_n)_{n\in\NN}$
be a sequence in $\C^\infty_c(\M)$ taking values in $[0,1]$ such that $f_n\equiv 1$ 
in $K_n$ and $\supp f_n\subset\In{K}_{n+1}$, where $(K_n)_{n\in\NN}$ is the exhaustion
of $\M$ defined above. Then one clearly has $\|\varphi\|_{\infty,k,K_n}\leq
\|\varphi\|_{\infty,k,f_n}\leq\|\varphi\|_{\infty,k,K_{n+1}}$ for all $\varphi\in\C^\infty(\M)$.
Finally, yet another equivalent, separating family of seminorms generating the 
compact-open topology which will play a major role in this work is given by the 
\emph{local ($L^2$) Sobolev seminorms}
\begin{align}
  \|\varphi\|_{2,k,K} &\doteq\sqrt{\sum^k_{j=0}\int_K|\nabla^j\varphi|^2_e\ud\mu_e}\ ,\quad
                        K\subset\M\text{ compact, }\In{K}\neq\varnothing\ ,\label{s2e4}\\
  \|\varphi\|_{2,k,f} &\doteq\sqrt{\sum^k_{j=0}\int_\M|f\nabla^j\varphi|^2_e\ud\mu_e}\ ,\quad
                        f\in\C^\infty_c(\M)\ ,\label{s2e5}
\end{align}
where $\ud\mu_e$ is the volume element associated to the Riemannian metric $e$ on $\M$.
The equivalence can be established by means of the sequence $(f_n)_{n\in\NN}$ defined
above together with the Sobolev inequalities, using a partition of unity subordinated 
to a finite covering of $\supp f_n$ by suitable domains of coordinate charts for each 
$n\in\NN$. 

\begin{remark}\label{s2r1}
  Formulae \eqref{s2e1}--\eqref{s2e5} can be extended to the space $\Gamma^\infty(\pi)$ 
  of smooth sections of a vector bundle $\pi:\E\To\M$ over $\M$: given a torsion-free 
  connection $\bar{\nabla}$ on $\E$ and a torsion-free connection $\nabla$ on $T\!\M$, 
  we can combine them into a torsion-free connection on $\otimes^kT^*\!\!\M\otimes\E$ 
  for all $k$ by using Leibniz's rule. We denote such a connection by $\nabla$ for all 
  $k\geq 0$, since there will be no danger of confusion. Once we write $\nabla^1
  \vec{\varphi}(X)=\nabla_X\vec{\varphi}$ for all $\vec{\varphi}\in\Gamma^\infty(\pi)$, 
  $X\in\Gamma^\infty(T\!\M\To\M)$, we can define $k$-th order iterated covariant 
  derivatives $\nabla^k\vec{\varphi}$ of $\vec{\varphi}\in\Gamma^\infty(\pi)$ for all 
  $k\geq 2$ by means of \eqref{s2e2}. We can now endow $\E$ with a Riemannian fiber 
  metric $\bar{e}$ and define
  \begin{equation}\label{s2e6}
    |\nabla^k\vec{\varphi}|^2_{\bar{e}}=((\otimes^k \bar{e}^{-1})\otimes\bar{e})
    (\nabla^k\vec{\varphi},\nabla^k\vec{\varphi})\ .
  \end{equation}
  Substituting \eqref{s2e6} into \eqref{s2e1} and \eqref{s2e3}--\eqref{s2e5} allows
  us to define the seminorms $\|\vec{\varphi}\|_{p,k,f}$, $\|\vec{\varphi}\|_{p,k,K}$ 
  of $\vec{\varphi}\in\Gamma^\infty(\pi)$ for all $f\in\C^\infty_c(\M)$, $\varnothing
  \neq\In{K}\subset K\subset\M$ compact, $p=2,\infty$. 
\end{remark}

The Whitney topology on $\C^\infty(\M)$, unlike the compact-open topology, is
\emph{not} a vector space topology in general. Since a sequence $(\varphi_n)_{n\in\NN}$ 
converges to $\psi\in\C^\infty(\M)$ in this topology if and only if there is 
a compact subset $K\subset\M$ such that $\psi_n(p)=\psi(p)$ for all $p\in\M\sm K$ 
and $\psi_n$ converges uniformly to $\psi$ on $K$ together with all its derivatives 
\cite{km}, we see that scalar multiplication is not Whitney-continuous at zero 
unless $\M$ is compact.

Nonetheless, the Whitney topology induces on $\C^\infty(\M)$ the structure of 
a \emph{flat affine manifold}, modelled over the subspace $\C^\infty_c(\M)$. To 
wit, for every $\varphi\in\C^\infty(\M)$ there is an open neighborhood basis on 
$\varphi$ of the form $\UU+\varphi=\{\varphi+\vec{\varphi}\ |\ \vec{\varphi}\in\UU\}$, 
where $\UU$ runs over a basis of open neighborhoods of zero in $\C^\infty_c(\M)$ 
in the latter's usual inductive limit topology. In particular, the connected 
component of $\varphi$ in the Whitney topology is exactly $\varphi+\C^\infty_c(\M)$. 
The coordinate chart associated to $\UU+\varphi$ is then given by 
$\kappa_\varphi(\varphi+\vec{\varphi})=\vec{\varphi}$, and the coordinate change 
map from $\UU_1$ to $\UU_2$ is given by $\kappa_{\varphi_2}\circ\kappa_{\varphi_1}^{-1}
(\vec{\varphi}_1)=\vec{\varphi}_1+(\varphi_1-\varphi_2)$, which is clearly affine. 
We remark that, due to the aforementioned connectedness property of the Whitney 
topology, the respective domains $\UU_1+\varphi_1$, $\UU_2+\varphi_2$ of 
$\kappa_{\varphi_1}$ and $\kappa_{\varphi_2}$ have nonvoid intersection if and 
only if $\varphi_1-\varphi_2$ has compact support, in which case we conclude 
from the argument in the previous paragraph that $\kappa_{\varphi_1}^{-1}\circ
\kappa_{\varphi_2}$ is even continuous with respect to the Whitney topology.

As argued in Appendix \ref{a1-calc}, the notion of smooth curves in the 
modelling space $\C^\infty_c(\M)$ allows one as well to use the atlas
\begin{equation}\label{s2e7}
  \Un=\{(\UU+\varphi,\kappa_\varphi)\ |\ \UU\subset\C^\infty_c(\M)\ni 0\text{ open, }
  \varphi\in\C^\infty(\M)\}\ .
\end{equation}
we have built in the previous paragraph to induce a \emph{smooth} manifold 
structure on $\C^\infty(\M)$. In particular, due to the affine structure of 
$\C^\infty(\M)$, the tangent and cotangent bundles of $\C^\infty(\M)$ are 
trivial, being respectively given by
\begin{align*}
  T\C^\infty(\M) &=\C^\infty(\M)\times\C^\infty_c(\M)\ ,\\
  T^*\C^\infty(\M) &=\C^\infty(\M)\times\D'(\wedge^d T^*\!\!\M\To\M)\ ,
\end{align*}
where $\D'(\wedge^d T^*\!\!\M\To\M)=\C^\infty_c(\M)'$ is the space of 
$d$-form-valued distributions on $\M$. We endow $T\C^\infty(\M)$ with 
a flat connection, to be defined as follows. The parallel transport 
operator $P_\gamma^{\lambda_1,\lambda_2}(\gamma(\lambda_1),\vec{t}\:)=
(\gamma(\lambda_2),\vec{t}\:)$ on $T\RR=\RR\times\RR$ along $\gamma\in
\C^\infty(\RR,\RR)$ associated to the standard flat connection on the 
target space $\RR$ of $\C^\infty(\M)$ can be pulled back to 
$T\C^\infty(\M)$ by setting
\[
P_{\alpha}^{\lambda_1,\lambda_2}(\alpha(\lambda_1,\cdot),\vec{\varphi})(p)\doteq
(\alpha(\lambda_2,p),\vec{\varphi}(p))=P_{\alpha(\cdot,p)}^{\lambda_1,\lambda_2}
(\alpha(\lambda_1,p),\vec{\varphi}(p))\ ,
\]
where $\alpha:\RR\times\M\To\RR$ defines a smooth curve in $\C^\infty(\M)$ with
respect to the Whitney topology (see Appendix \ref{a1-calc}). Given sections 
$X,Y$ of $T\C^\infty(\M)$ taking smooth curves in $\C^\infty(\M)$ with respect 
to the Whitney topology to smooth curves in $T\C^\infty(\M)$, we may define at 
each $\varphi\in\C^\infty(\M)$ 
\begin{equation}\label{s2e8}
  D_YX[\varphi]=\frac{\partial}{\partial\lambda}\restr{\lambda=0}\left(P^{\lambda,0}_{\alpha}
    X[\alpha(\lambda,\cdot)]\right)\ ,
\end{equation}
where $\alpha:\RR\times\M\To\RR$ is a smooth curve in $\C^\infty(\M)$ such that 
$\alpha(0,p)=\varphi(p)$ and $\frac{\dd}{\dd\lambda}\restr{\lambda=0}\alpha
(\lambda,p)=\pr_2(Y[\varphi])(p)$. An example of such a curve is 
\begin{equation}\label{s2e9}
  \alpha(\lambda,p)=\varphi(p)+\lambda\pr_2(Y[\varphi])(p)\ .
\end{equation}
We say that $D$ is the \emph{ultralocal lift} of the standard flat connection on 
the target space\footnote{In the context of field theory, such connections were 
  formally introduced in \cite{dewitt}. They allow one to extend to higher orders
  the notion of \emph{fiber derivative} employed in the calculus of variations \cite{binz}. 
  For a precise, general concept of ultralocal lifts of connections on target spaces, 
  see for instance Example 4.5.3, pp. 94 of \cite{hamilton}.} $\RR$, for $D_YX[\varphi](p)$ 
depends only on $\varphi(p)$. In what follows, we automatically extend $D$ to all 
covariant and contravariant tensor fields on $\C^\infty(\M)$ (see Appendix \ref{a1-calc} 
for a precise definition) in the standard fashion, i.e. by tensoring and taking adjoint 
inverses of the parallel transport operator.

It is clear from the above definition that $P_\alpha$ defined above is the parallel 
transport operator along $\alpha$ associated to $D$. It is a consequence of the 
ultralocality of $D$, however, that much more is true:

\begin{enumerate}
\item The geodesic $\alpha$ starting at $(\varphi,\vec{\varphi})\in T\C^\infty(\M)$ 
  is given by \eqref{s2e9}. As a consequence, the exponential map $\exp_D:T\C^\infty(\M)
  \To\C^\infty(\M)\times\C^\infty(\M)$ of $D$ is complete and given by
  \[
  \exp_D(\varphi,\vec{\varphi})=(\varphi,\varphi+\vec{\varphi})=
  (\varphi,\kappa_\varphi(\vec{\varphi}))\ .
  \]
  In other words, the chart $\kappa_\varphi$ is precisely the normal coordinate 
  chart around $\varphi$ associated to $D$.
\item The curvature tensor of $D$ is given by 
  \[
  \Riem_D(X,Y)[\varphi](p)=\Riem_{\varphi(p)}(X[\varphi](p),Y[\varphi](p))=0\ ,
  \]
  where $\Riem_t\equiv 0$ is the Riemann curvature of the standard flat connection
  on the target space $\RR$ at the point $t$. As a consequence, the $k$-th order
  iterated covariant derivative $D^kX(Y_1,\ldots,Y_k)$ of a tensor field $X$ along 
  vector fields $Y_1,\ldots,Y_k$ is \emph{symmetric} in $Y_1,\ldots,Y_k$ for 
  all $k$. 
\end{enumerate}

We close this Subsection with two technical Lemmata. The first is a simple 
but useful manifestation of the fact that the Whitney topology is finer than
the compact-open topology:

\begin{lemma}\label{s2l1}
  Let $\UU\subset\C^\infty(\M)$ be open with respect to the Whitney topology. 
  Then for every $\varphi_0\in\UU$, $V\subset\M$ open, there is a $\varphi\in\UU$
  such that $\supp(\varphi-\varphi_0)\neq\varnothing$ is compact and contained 
  in $V$, and $\lambda(\varphi-\varphi_0)\in\UU-\varphi_0$ for all $\lambda\in[-1,1]$. 
\end{lemma}
\begin{proof}
  By the reasoning in the paragraphs preceding this Lemma, there is an
  absolutely convex open neighborhood $\VV$ of zero in $\C^\infty_c(\M)$
  contained in $\UU-\varphi_0$. Given any $\varphi_1\in\C^\infty_c(\M)$, there
  is a $t_1>0$ such that $t\psi_1\in\VV$ for all $t\in\RR$ with $|t|\leq t_1$, 
  since $\VV$ is absorbent. Choose $\varphi_1$ with $\supp\varphi_1\subset V$,
  set $\varphi=\varphi_0+t_1\varphi_1$, and we are done.
\end{proof}

The second allows one to strengthen the conclusion of Lemma \ref{s2l1}: 

\begin{lemma}\label{s2l2}
  Let $\varphi_0\in\C^\infty(\M)$, $r\in\NN\cup\{0\}$, $p\in\M$. Then there is 
  $\varphi\in\C^\infty(\M)$ satisfying $j^r\varphi(p)=j^r\varphi_0(p)$, such that 
  $\supp(\varphi-\varphi_0)\neq\varnothing$ is contained in an arbitrarily small 
  open neighborhood $U$ of $p$ with compact closure $K$, and 
  $\|\varphi'-\varphi_0\|_{\infty,r,K}<\epsilon$ for $\epsilon>0$ arbitrarily small. 
\end{lemma}
\begin{proof}
  Since we are dealing with a local statement, we assume without loss of generality 
  that $\M=\RR^d$, $p=0$, $\varphi_0\equiv 0$, $e$ is the standard Euclidean metric 
  and $\nabla=\dd$ is the associated (flat) Levi-Civita connection. Let now $\varphi'
  \in\C^\infty(\M)$ be such that $j^r\varphi'(p)=j^r\varphi_0(p)$; it follows from 
  Taylor's formula with remainder that $\dd^\alpha\varphi'(x)=O(\|x\|^{r+1-|\alpha|})$ 
  as $\|x\|\To 0$, for all multi-indices $\alpha$ such that $0\leq|\alpha|\leq r$. 
  Let $f\in\C^\infty_c(\RR^d)$ such that $f(x)=1$ for $\|x\|\leq\frac{1}{2}$ and 
  $f(x)=0$ for $\|x\|\geq 1$. Given $R>0$, define $f_R(x)=f(R^{-1}x)$. It follows 
  from the chain rule that $\dd^\alpha f_R(x)=R^{-|\alpha|}(\dd^\alpha f)(R^{-1}x)$. 
  Define now $\varphi=f_R\varphi'$; Leibniz's rule gives us that 
  \[
  \|\varphi\|_{\infty,r,K}\leq C_{r,K}\|\varphi'\|_{\infty,r,\ol{B_R(0)}}
  \|f\|_{\infty,r,\ol{B_1(0)}}R
  \] 
  for all $K\subset\RR^d$ such that $\In{K}\supset\ol{B_R(0)}$, where 
  $B_\lambda(0)=\{x\in\RR^d\ |\ \|x\|<\lambda\}$. Taking $R$ sufficiently small
  yields the desired bound. 
\end{proof}

\subsection{\label{s2-kin-obs}Functionals as observables}

Our observable quantities will be maps $F:\UU\To\CC$ which we call \emph{functionals}, 
where $\UU\subset\C^\infty(\M)$ is usually some open set in the compact-open topology,
though we may occasionally consider more general subsets. The need to localize the 
domain of definition of functionals comes from the fact that, in the study of nonlinear
equations of motion, one is led to consider functionals which are not \emph{a priori} 
defined for \emph{all} field configurations due to the existence of solutions blowing 
up in finite time. We shall now introduce a concept which tells us in which sense 
functionals are localized in a certain region of space-time, following \cite{brudf}.

\begin{definition}[Space-time support]\label{s2d1} 
  Let $\UU\subset\C^\infty(\M)$. The \emph{space-time support} $\supp F$ of a functional 
  $F:\UU\To\CC$ is the (closed) subset composed by the points $p\in\M$ such that for any 
  neighborhood $U$ of $p$ we can find $\psi\in\UU,\varphi\in\UU-\psi$ with $\supp\varphi
  \subset U$ for which $F(\varphi+\psi)\neq F(\psi)$. The space of functionals over $\UU$ 
  with \emph{compact} space-time support in $\M$ will be denoted by $\Fun_{00}(\M,\UU)$. 
\end{definition}

In other words, a functional $F$ is insensitive to disturbances of its argument which 
are localized outside $\supp F$. As shown by Lemma \ref{s2l5} below, Definition \ref{s2d1}
gives a nonlinear generalization of the notion of support of a distribution. It is 
important, on the one hand, to emphasize that Definition \ref{s2d1} depends on the 
domain of definition $\UU$ of $F$. For instance, if we restrict $F$ to a smaller domain 
of definition $\VV\subset\UU$, then $\supp F$ will in general be a \emph{smaller} subset 
of $\M$ (see Remark \ref{s2r2} right below). On the other hand, the domain of definition 
of $F$ will always be clear from the context, so we refrain from referring to it in the 
notation.

\begin{remark}\label{s2r2}
  Let us give some simple examples of functionals. Given a compact region $K\subset\M$ of
  the space-time manifold $\M$ and $0\leq f\in\C^\infty_c(\M)$ satisfying $\int_\M f\ud\mu_g=1$, 
  we define the functionals $F,G,H:\C^\infty(\M)\To\CC$ as 
  \begin{equation}\label{s2e10}
    F(\varphi)=\|\varphi\|_{\infty,0,K}=\sup_K|\varphi|\ ,\quad G(\varphi)=\int_\M f\varphi\ud\mu_g\ ,
    \quad H(\varphi)=
    \begin{cases} 
      \frac{1}{1+\sup_\M|\varphi|} & \varphi\text{ bounded}\ ,\\ 0 & \text{otherwise}\ . 
    \end{cases}
  \end{equation}
  One clearly sees that $F$ and $G$ have compact space-time support (indeed, we have that 
  $\supp F=K$ and $\supp G=\supp f$), whereas $H$ does not. Other examples are the local 
  Sobolev seminorms $\|\varphi\|_{2,k,K}$ and $\|\varphi\|_{2,k,f}$ respectively defined in 
  \eqref{s2e4} and \eqref{s2e5}. We shall now give a slightly more complicated example
  which explicitly displays the dependence of the space-time support of a functional on the 
  latter's domain. Let $\chi:\RR\To[0,1]$ be a smooth function such that $\chi^{-1}(1)=[-1,1]$
  and $\chi^{-1}(0)=\RR\sm(-2,2)$. Setting $\chi_R(t)=\chi(R^{-1}t)$ for $R>0$, 
  define
  \begin{equation}\label{s2e11}
    G_R(\varphi)=\exp\left(1-\chi_R\circ G(\varphi)\right)\ ,
  \end{equation}
  with $G$ as defined in \eqref{s2e10}. Let now $\UU_{R'}=\{\varphi\in\C^\infty(\M)\ |\ 
  \|\varphi\|_{\infty,0,\supp f}<R'\}$ for $R'>0$; we have that
  \[
  \supp G_R\restr{\UU_{R'}}=
  \begin{cases}
    \varnothing & R'\leq R\ ,\\ \supp f & R'>R\ .
  \end{cases}
  \]
  Indeed, in the first case, we have that $G_R\restr{\UU_{R'}}\equiv 1$.
\end{remark}

We endow each $\Fun_{00}(\M,\UU)$, for $\UU\ni\varphi$ running over the compact-open 
topology of $\C^\infty(\M)$, with the following pointwise algebraic operations:

\begin{itemize}
\item Sum $F,G\mapsto (F+G)(\varphi)\doteq F(\varphi)+G(\varphi)$;
\item Product $F,G\mapsto (F\cdot G)(\varphi)\doteq F(\varphi)G(\varphi)$;
\item Involution $F\mapsto F^*(\varphi)\doteq\ol{F(\varphi)}$;
\item Multiplication by scalars $z\in\CC$, $F\mapsto(z\cdot F)(\varphi)\doteq zF(\varphi)$;
\item Unit $\id:\varphi\mapsto 1$.
\end{itemize}

Now we show that the algebraic operations of $\Fun_{00}(\M,\UU)$ preserve space-time 
supports. As a direct consequence, these operations turn $\Fun_{00}(\M,\UU)$ into 
a commutative unital *-algebra. Firstly, it is trivial to check that the scalar 
multiplication by any $0\neq\lambda\in\CC$ and the involution leave the support 
unchanged, whereas any scalar multiple of $\id$ has empty space-time support. The 
full assertion is then a consequence of the following

\begin{lemma}\label{s2l3}
  Let $\UU\subset\C^\infty(\M)$, $F,G$ functionals over $\UU$. Then:
  \begin{itemize}
  \item The sum $F+G$ satisfies
    \[
    \supp(F+G)\subset\supp F \cup \supp G\ ;
    \]
  \item The product $F\cdot G$ satisfies
    \[
    \supp(F\cdot G)\subset\supp F\cup\supp G\ .
    \]
    In particular, $\Fun_{00}(\M,\UU)$ is a commutative unital *-algebra.
  \end{itemize}
\end{lemma}
\begin{proof}
  Let us assume that $p\notin\supp F\cup\supp G$, that is,
  \[
  p\in\complement (\supp F \cup\supp G)=\complement\supp F\cap\complement\supp G\ .
  \] 
  By the definition of space-time support, there is an open neighborhood $V$ of 
  $p$ such that for all $\varphi_0\in\UU$, $\varphi\in\UU-\varphi_0$ satisfying
  $\supp\varphi\subset V$, we have that $F(\varphi_0+\varphi)=F(\varphi_0)$ and
  $G(\varphi_0+\varphi)=G(\varphi_0)$, hence $(F+G)(\varphi_0+\varphi)=
  (F+G)(\varphi_0)$ and $(F\cdot G)(\varphi_0+\varphi)=(F\cdot G)(\varphi_0)$ 
  for all such $\varphi_0,\varphi$. This entails that $p\notin\supp(F+G)$ and
  $p\notin\supp(F\cdot G)$, as desired.
\end{proof}

We emphasize that, unlike for supports of functions on $\M$, the stronger 
property $\supp(F\cdot G)\subset\supp F\cap\supp G$ \emph{does not} hold for 
space-time supports of functionals. A counter-example is given by $F=G_{R_1}$ and 
$G=G_{R_2}$ with $R_1<R_2$, where $G_R$ is defined for all $R>0$ in \eqref{s2e11}. 
The reason is that the notion of space-time support is a \emph{relative} one; it 
is not necessarily true that $F(\varphi)$ vanishes if $\varphi$ is supported outside 
$\supp F$. For instance, given $f\in\C^\infty(\M)$ with $\int_\M f\ud\mu_g=1$, we have 
that the functional $F(\varphi)=\int_\M f\exp(\varphi)\ud\mu_g$ satisfy $F(\varphi)=1$ 
for all $\varphi\in\C^\infty(\M)$ such that $\supp\varphi\cap\supp f=\varnothing$. In 
the above counter-example, we also have that $G_R$ is nowhere vanishing for all $R>0$. 

The \emph{raison d'être} of Definition \ref{s2d1} becomes evident if one assumes the 
following property:

\begin{definition}[Additivity]\label{s2d2} 
  Let $\UU\subset\C^\infty(\M)$. A functional $F\in\Fun_{00}(\M,\UU)$ is said to be 
  \emph{additive} if for all $\varphi_1\in\UU,\varphi_2,\varphi_3\in\UU-\varphi_1$ 
  such that $\varphi_2+\varphi_3\in\UU-\varphi_1$ and $\supp\varphi_2\cap\supp\varphi_3
  =\varnothing$ we have 
  \begin{equation}\label{s2e12}
    F(\varphi_1+\varphi_2+\varphi_3)=F(\varphi_1+\varphi_2)-F(\varphi_1)+F(\varphi_1+\varphi_3)
  \end{equation}
  or, more concisely,
  \begin{equation}\label{s2e13}
    F_{\varphi_1}(\varphi_2+\varphi_3)=F_{\varphi_1}(\varphi_2)+F_{\varphi_1}(\varphi_3) \,
  \end{equation}
  where $F_{\varphi}(\psi)\doteq F(\varphi+\psi)-F(\varphi)$.
\end{definition}

As it will be seen shortly, this notion essentially captures what it means for 
$F\in\Fun_{00}(\M,\UU)$ to be \emph{local} with respect to the space-time manifold 
$\M$. For instance, in the case that $\UU=\C^\infty(\M)$ we have\footnote{The restriction 
  on $\UU$ can be weakened in a certain sense. See Lemma \ref{s3l3}.} the following 
nonlinear analog of a partition of unity, introduced in Lemma 3.2 of \cite{brudf}.
Its simple proof is included here for the convenience of the reader.

\begin{lemma}\label{s2l4}
  Any additive functional $F\in\Fun_{00}(\M,\C^\infty(\M))$ can be decomposed as 
  a finite sum of additive functionals with arbitrarily small space-time support.
\end{lemma}
\begin{proof}
  First of all, let us endow $\M$ with a complete auxiliary Riemannian metric $h$,
  whose associated distance function is given by $d_h:\M\times\M\To[0,+\infty)$. 
  In what follows, by ``distance'' between $p,q\in\M$ we mean $d_h(p,q)$, and
  a ``ball of radius $R$'', an open set $\{q\in\M:d_h(p,q)<R\}$ for some $p\in\M$. 
  
  Let $\epsilon>0$ be arbitrary, and $(B_i)_{i=1,\ldots,n}$ a finite covering of 
  $\supp F$ by balls of radius $\epsilon/4$. Associate to this covering a subordinate 
  partition of unity $(\chi_i)_{i=1,\ldots,n}$. By a repeated use of the additivity of 
  $F$ we arrive at a decomposition of the form
  \begin{equation}\label{s2e14}
    F=\sum_I s_I F_I\ ,
  \end{equation} 
  with $s_I\in\{\pm 1\}$, $F_I(\varphi)=F(\varphi\sum_{i\in I}\chi_i)$, where $\varphi
  \in\C^\infty(\M)$ and $I$ runs over all subsets of $\{1,\ldots,n\}$ such that 
  $B_i\cap B_j\neq\varnothing$ for all $i,j\in I$. It is obvious that $F_I$ is again 
  an additive functional, and from the definition of space-time support we immediately 
  find that $\supp F_I\subset\bigcup_{i\in I}B_i\doteq B_I$. Since any two points in 
  $B_I$ have distance less than $\epsilon$, then each $B_I$ is contained in a ball of 
  radius $\epsilon$. 
\end{proof}

\begin{remark}\label{s2r3}
  The concept of an additive functional, although not exactly mainstream, is by no 
  means new in the mathematical literature (consider \cite{rao} as a starting point). 
  In the present case, it was motivated by the study of the set of possible counterterms 
  generated by all choices of renormalization prescription in perturbative algebraic 
  quantum field theory \cite{dutfre2,brudf}. We have already seen examples of additive
  functionals, such as the square of the local Sobolev seminorm \eqref{s2e5}. Counter-examples
  include the functional $F$ defined in \eqref{s2e10} and the functional $G_R$ defined
  in \eqref{s2e11}. 
\end{remark}

A further property we will demand from our functionals concerns their differentiability.
We will just spell the complete definition we need for convenience, which builds on the 
discussion in Appendix \ref{a1-calc}.

\begin{definition}\label{s2d3}
  Let $\UU\subset\C^\infty(\M)$ be open in the compact-open topology. We say 
  that a functional $F\in\Fun_{00}(\M,\UU)$ is \emph{differentiable of order $m$} if for 
  all $k=1,\ldots,m$ the $k$-th order directional (Gâteaux) derivatives (henceforth called 
  \emph{functional derivatives}) 
  \begin{equation}\label{s2e15}
    F^{(k)}[\varphi](\vec{\varphi}_1,\ldots,\vec{\varphi}_k)=\Spr{F^{(k)}[\varphi],\vec{\varphi}_1
      \otimes\cdots\otimes\vec{\varphi}_k}\doteq\frac{\dd^k}{\dd\lambda_1\cdots\dd\lambda_k}
    \Restr{\lambda_1=\cdots=\lambda_k=0}F\left(\varphi+\sum^k_{j=1}\lambda_j\vec{\varphi}_j\right)
  \end{equation}
  exist as jointly continuous maps from $\UU\times\C^\infty(\M)^k$ to $\RR$, where 
  $\spr{\cdot,\cdot}$ denotes dual pairing. In particular, for each $\varphi$ fixed, 
  $F^{(k)}[\varphi]$ is a distribution \emph{density} of compact support on $\M^k$. If $F$ 
  is differentiable of order $m$ for all $m\in\NN$, we say that $F$ is \emph{smooth}.
\end{definition}

In certain cases, we can extend Definition \ref{s2d3} to the case when $\UU$ is 
no longer open (see Appendix \ref{a1-calc}).

The relation of the notion of space-time support of a functional to the notion of 
support of a distribution can be made more transparent for \emph{differentiable}
elements of $\Fun_{00}(\M,\UU)$:

\begin{lemma}\label{s2l5} 
  Let $\UU\subset\C^\infty(\M)$ be open in the compact-open topology and convex. If 
  $F\in\Fun_{00}(\M,\UU)$ is a differentiable functional of order one, then 
  \begin{equation}\label{s2e16}
    \supp F=\ol{\bigcup_{\varphi\in\UU}\supp F^{(1)}[\varphi]}\ .
  \end{equation}
\end{lemma}
\begin{proof}
  If $p\in\supp F$, then by definition there are $\varphi\in\UU$, $\vec{\varphi}
  \in\UU-\varphi$ with $\vec{\varphi}$ supported in a neighborhood of $p$ such that 
  $F(\varphi+\vec{\varphi})\neq F(\varphi)$. By the fundamental theorem of Calculus 
  \eqref{a1e2}, there is a $\lambda_0\in(0,1)$ such that $F^{(1)}[\varphi+\lambda_0
  \vec{\varphi}](\vec{\varphi})\neq 0$ (this is the only place where convexity of 
  $\UU$ is used). This implies the inclusion $\supp F\subset\ol{\bigcup_{\varphi\in\UU}
    \supp F^{(1)}[\varphi]}$.
  
  For the opposite one we argue as follows. Let us suppose that $p\in\supp 
  F^{(1)}[\varphi]$, then this means that there is a $\vec{\varphi}\in\C^\infty(\M)$ 
  supported in a neighborhood of $p$ such that $F^{(1)}[\varphi](\vec{\varphi})\neq 0$,
  whence it follows that $F(\varphi+\lambda\vec{\varphi})\neq F(\varphi)$ for all
  $\lambda$ chosen sufficiently small (depending on $\vec{\varphi}$) so that 
  $\varphi+\lambda\vec{\varphi}\in\UU$. We then conclude that $p\in\supp F$, i.e.
  \[
  \supp F^{(1)}[\varphi]\subset\supp F\ .
  \] 
  Taking the union of the left-hand side with respect to all $\varphi\in\UU$ and 
  closing implies the thesis.
\end{proof}

We remark that the same argument used in Lemma \ref{s2l5} to prove the inclusion
$\supp F^{(1)}[\varphi]\subset\supp F$ can be used to show that if $F$ is 
differentiable of order $m\geq 1$, then $\supp F^{(k)}[\varphi]\subset(\supp F)^k$ 
for all $1\leq k\leq m$.

We shall now display formula \eqref{s2e16} in action using a specific example. Let
$G_R=G_R\restr{\UU_{R'}}\in\Fun_{00}(\M,\UU_{R'})$ be the functional defined as in 
\eqref{s2e11}. By Faà di Bruno's formula \eqref{a1e7}, one sees that $G_R$ is smooth 
for all $R,R'>0$. In particular, by the chain rule \eqref{a1e3},
\[
G^{(1)}_R[\varphi](\vec{\varphi})=-\frac{1}{R}G_R(\varphi)(\chi')_R\circ G(\varphi)
G(\vec{\varphi})\ .
\]
When $R'\leq R$, we have that $(\chi')_R\circ G(\varphi)=0$ for all $\varphi\in\UU_{R'}$,
whence $G^{(1)}_R[\varphi]=0$ for all such $\varphi$. If $R'>R$, then we have that
$\supp G^{(1)}_R[\varphi]=\varnothing$ if $\|\varphi\|_{\infty,0,\supp f}\leq R$, and
$\supp G^{(1)}_R[\varphi]=\supp f$ if $R<\|\varphi\|_{\infty,0,\supp f}<R'$.

\begin{remark}\label{s2r4}
  A natural question that arises at this point, whose answer is in general evaded 
  in the literature, is how Definition \ref{s2d3} fits into the manifold structure 
  of $\C^\infty(\M)$ induced by the Whitney topology. This question is answered by means
  of the following fact: given any compact region $K\subset\M$ and any nonvoid subset 
  $\UU\subset\C^\infty(\M)$, one can uniquely extend any $F\in\Fun_{00}(\M,\UU)$ with 
  $\supp F\subset\In{K}$ to the subset $i_\chi^{-1}(\UU)$, where $i_\chi:\C^\infty(\M)\To
  \varphi_0+\C^\infty_c(\M)$ is defined by $i_\chi(\varphi)=\varphi_0+\chi(\varphi-\varphi_0)$, 
  $\varphi_0\in\UU$ is fixed and $\chi\in\C^\infty_c(\M)$ satisfies $\chi(p)=1$ for all
  $p\in K$. It is clear that $i_\chi$ is a continuous (in fact, even smooth) map 
  from $\C^\infty(\M)$ into itself, if the domain is endowed with the \emph{compact-open} 
  topology and the codomain is endowed with the \emph{Whitney} topology. In particular, 
  if $\UU$ is a connected, Whitney-open neighborhood of $\varphi_0$, then $i_\chi^{-1}(\UU)$
  is open in the compact-open topology, where $F^{(k)}$ becomes uniquely defined whenever
  it exists, for all $k\geq 1$. Moreover, since $\supp F^{(k)}[\varphi]\subset(\supp F)^k$, 
  one also concludes that
  \begin{equation}\label{s2e17}
    F^{(k)}[\varphi](\vec{\varphi}_1,\ldots,\vec{\varphi}_k)=F^{(k)}[\varphi](\chi\vec{\varphi}_1,
    \ldots,\chi\vec{\varphi}_k)=D^kF[\varphi](\chi\vec{\varphi}_1,\ldots,\chi\vec{\varphi}_k)\ ,
  \end{equation}
  where $\chi\vec{\varphi}_j$ is understood as the covariantly constant vector
  field $\varphi\mapsto(\varphi,\chi\vec{\varphi})$, $1\leq j\leq k$. Since the 
  left hand side of the above formula is independent of $\chi$, we just write 
  \begin{equation}\label{s2e18}
    F^{(k)}[\varphi](\vec{\varphi}_1,\ldots,\vec{\varphi}_k)=D^kF[\varphi](\vec{\varphi}_1,
    \ldots,\vec{\varphi}_k)\ .
  \end{equation}
  In particular, $D^kF$ defines a smooth tensor field on $\C^\infty(\M)$ when 
  the latter is endowed with the smooth manifold structure induced from 
  $\C^\infty_c(\M)$ (see Remark \ref{a1r1} and the discussion preceding it).
\end{remark}

Throughout the paper, our functionals $F$ of interest will always be smooth
functionals with compact space-time support. Thanks to Remark \ref{s2r4},
if $F$ is only defined in an open subset $\UU\subset\C^\infty(\M)$ in the 
Whitney topology, we can uniquely extend such a $F$ to an open subset 
$\C^\infty(\M)\supset\tilde{\UU}\supset\UU$ in the compact-open topology
and unambiguously define $F^{(k)}$ therein for all $k\geq 1$. Three very 
important spaces of such functionals are the following: 

\begin{definition}\label{s2d4}
  Let $\UU\subset\C^\infty(\M)$ be open with respect to the compact-open topology. 
  The vector subspaces of $\Fun_{00}(\M,\UU)$ given by
  \begin{align}
    \Fun_0(\M,\UU) &=\{F\in\Fun_{00}(\M,\UU)\text{ smooth }\ |\ \ F^{(k)}[\varphi]\in
                     \Gamma^\infty_c(\wedge^{kd}T^*\!\!\M^k\To\M^k),\,\forall\varphi
                     \in\UU,\,k\geq 1\}\ ,\label{s2e19}\\
    \Fun_{\loc}(\M,\UU) &=\{F\in\Fun_{00}(\M,\UU)\text{ smooth }\ |\ \ \supp F^{(2)}[\varphi]
                          \subset\Delta_2(\M),\,\forall\varphi\in\UU\}\ \text{and}\label{s2e20}\\
    \Fun_{\mu\loc}(\M,\UU) &=\{F\in\Fun_{\loc}(\M,\UU)\ |\ F^{(1)}[\varphi]\in\Gamma^\infty_c
                             (\wedge^dT^*\!\!\M\To\M),\,\forall\varphi\in\UU\}\ ,\label{s2e21}
  \end{align}
  where $\Delta_k(\M)=\{(p,\ldots,p)\in\M^k:p\in\M\}$ is the small diagonal of $\M$ in $\M^k$, 
  are said to be respectively the spaces of \emph{regular}, \emph{local} and \emph{microlocal} 
  functionals over $\UU$. 
\end{definition}

Criterion \eqref{s2e20} for locality of a functional was put forward in \cite{dutfre2} 
in the case of functionals depending polynomially on the field configuration.
We stress that $\Fun_0(\M,\UU)$ is even a *-subalgebra of $\Fun_{00}(\M,\UU)$.

%

Microlocal functionals comprise many functionals of physical interest. For instance, 
let $\omega\in\Gamma^\infty(\wedge^dT^*\!J^r(\M,\RR)\To J^r(\M,\RR))$; given any 
$f\in\C^\infty_c(\M)$, the functional
\begin{equation}\label{s2e22}
  F(\varphi)=\int_\M f(j^r\varphi)^*\omega
\end{equation}
is clearly seen to be microlocal over any $\UU\subset\C^\infty(\M)$ open in the compact-open
topology. Conversely, it will be shown in Proposition \ref{s2p2} below that all microlocal
functionals are essentially of this form. The above example becomes somewhat trivial
if we take instead a closed $p$-form $\omega$ on $J^r(\M,\RR)$ with $p<d$, and define
\begin{equation}\label{s2e23}
  G(\varphi)=\int_\N(j^r\varphi)^*\omega\ ,
\end{equation}
where $\N\subset\M$ is a compact, $p$-dimensional submanifold without boundary. More
precisely, it can be shown \cite{wald2} that $DG[\varphi](\vec{\varphi})$ for $G$ as 
in \eqref{s2e23} is represented by the integral over $\N$ of an \emph{exact} $p$-form 
on $\M$, hence $DG[\varphi]=0$ for all $\varphi\in\UU$. In particular, $\supp G=\varnothing$, 
that is, $G$ is locally constant. If $\omega$ is not closed or $\N$ has a nonvoid 
boundary, then $G$ is still a local functional, but \emph{not} microlocal (see also
example \eqref{s2e25} below).

It is easy to display examples of smooth functionals with compact space-time
support which are \emph{not} local. If $F,G\in\Fun_{\mu\loc}(\M,\UU)$, Leibniz's rule 
\eqref{a1e4} applied twice to $(F\cdot G)^{(2)}[\varphi](\vec{\varphi}_1,\vec{\varphi}_2)$ 
gives rise to a term of the form $F^{(1)}[\varphi](\vec{\varphi}_1)G^{(1)}[\varphi]
(\vec{\varphi}_2)+G^{(1)}[\varphi](\vec{\varphi}_1)F^{(1)}[\varphi](\vec{\varphi}_2)$, 
whose kernel for fixed $\varphi$ is represented by a smooth, compactly supported density 
on $\M^2$ and hence \emph{not} supported on $\Delta_2(\M)$ unless it is identically zero, 
in which case either $F$ or $G$ must be constant. Hence, we conclude that $F\cdot G$ 
cannot be local if $\supp F,\supp G\neq\varnothing$. It follows from the same 
argument (using Faà di Bruno's formula \eqref{a1e7} instead of Leibniz's rule) that 
if, for instance, $\psi:\CC\To\CC$ is entire analytic and not affine, and $G$ is 
microlocal with $\supp G\neq\varnothing$, then $F=\psi\circ G$ cannot be local. A 
typical such example is
\begin{equation}\label{s2e24}
  F(\varphi)=\exp\left(\int_\M\varphi\omega\right)\ ,\quad\omega\in
  \Gamma^\infty_c(\wedge^dT^*\!\!\M\To\M)\ ,
\end{equation}
which even happens to be regular. In fact, one immediately sees that a regular 
functional is local if and only if it is affine, in which case it is also microlocal.

Finally, to display the difference between local and microlocal functionals, consider a
closed, smooth timelike submanifold without boundary $\N\subset\M$ and with codimension
$p>0$ (e.g. a smooth timelike curve parametrized over $\RR$). If $\iota:\N\hookrightarrow\M$ 
is the natural inclusion, $X$ is a normal unit $p$-vector field on $\N$ with respect to $g$
(suitably extended to an open neighborhood of $\N$ in $\M$) and $f\in\C^\infty_c(\M)$, set
\begin{equation}\label{s2e25}
  F(\varphi)=\int_\N\iota^*(\varphi fi_X\ud\mu_g)\ .
\end{equation}
$F$ is clearly local, for $F^{(2)}\equiv 0$. However, $F^{(1)}[\varphi]$ is $f$ times the
submanifold measure induced by $\ud\mu_g$ on $\N$, hence it is \emph{not} a smooth density 
on $\M$ and thus $F$ is not microlocal.

%

Returning to the general development of our framework, now we are in a position 
to make more precise the claim preceding Lemma \ref{s2l4}, sharpening Lemma 3.1
of \cite{brudf}.

\begin{proposition}\label{s2p1} 
  Let $\UU\subset\C^\infty(\M)$ be open with respect to the compact-open topology, 
  and $F\in\Fun_{00}(\M,\UU)$ be smooth. Then $F$ belongs to $\Fun_{\loc}(\M,\UU)$ 
  if and only if it is additive. Moreover, in this case we have that $\supp F^{(k)}
  [\varphi]\subset\Delta_k(\M)$ for all $k\geq 2$, $\varphi\in\UU$. 
\end{proposition}
\begin{proof}
  \noindent $(\If)$ for any $k\geq 2$, assume that in the support of $F^{(k)}[\varphi]$ 
  there are two points $x_i\neq x_j$. Then, there exist two smooth functions $\varphi_i,
  \varphi_j$ such that $x_i\in\supp\varphi_i$, $x_j\in\supp\varphi_j$ and $\supp\varphi_i
  \cap\supp\varphi_j=\varnothing$. By additivity we split the right-hand side of the 
  formula for $F^{(k)}$ in Definition \ref{s2d3} according to the supports of $\varphi_i$ 
  and $\varphi_j$, but the derivatives act always on all $\lambda_j$'s, hence we get zero. 
  In particular, the last assertion holds.
  
  \noindent $(\Then)$ Assume that $\varphi_1,\varphi_2\in\VV\subset\UU-\varphi$ are such that 
  $\supp\varphi_1\cap\supp\varphi_2=\varnothing$, where $\VV$ is an absolutely convex open 
  neighborhood of zero. Using the fundamental theorem of Calculus \eqref{a1e2}, we write 
  \[ 
  F_\varphi(\varphi_1+\varphi_2)=F_\varphi(\varphi_1)+\int_0^1\ud\lambda\ \frac{\ud}{\ud\lambda}
  \ F_\varphi(\varphi_1+\lambda\varphi_2)\ .
  \] 
  The integral in the right hand side can be rewritten as
  \[
  \int_0^1 \ud\lambda\ \frac{\ud}{\ud\lambda}\ F_\varphi(\varphi_1 +\lambda\varphi_2)=F_\varphi
  (\varphi_2)+\int_0^1 \ud\lambda\ \int_0^1 \ud\mu\ F^{(2)}[\mu\varphi_1 +\lambda\varphi_2+
  \varphi](\varphi_1,\varphi_2)\ .
  \]
  By locality, $F^{(2)}[\mu\varphi_1+\lambda\varphi_2+\varphi]$ is supported in $\Delta_2(\M)$, 
  but by our initial assumption $\supp\varphi_1\cap\supp\varphi_2=\varnothing$, hence 
  \eqref{s2e13} and the last assertion of the Proposition holds. 
\end{proof}

For \emph{microlocal} functionals there is a refinement of Proposition \ref{s2p1} which identifies 
this class of functionals with the kind of local functionals usually employed by physicists, such 
as \eqref{s2e22} and \eqref{s2e23}. We build over the argument sketched in the proof of Theorem 2, 
pp. 139 of \cite{brufre2}, with a few changes (see also Theorem I.2 of \cite{broudlr}). The only
missing ingredient is the following mild technical condition:

\begin{definition}\label{s2d5}
  Let $\Fun_1,\Fun_2$ be locally convex vector spaces, $\varnothing\neq\UU\subset\Fun_1$ open. 
  A map $T:\UU\To\Fun_2$ is said to be \emph{locally bornological} (into $\Fun_2$) if for all 
  $\varphi\in\UU$ there is $\VV\ni\varphi$, $\VV\subset\UU$ open such that $T\restr{\VV}$ 
  maps bounded subsets of $\VV$ into bounded subsets of $\Fun_2$.
\end{definition}

If $\Fun_1$ is normable, then locally bornological maps are just the same as locally bounded maps.
If $\Fun_1$ is semi-Montel (that is, any bounded subset of $\Fun_1$ is relatively compact), then
any continuous map $T:\UU\To\Fun_2$ is locally bornological: given any $\varphi\in\UU$, take an 
open neighborhood $\VV$ of $\varphi$ such that $\ol{\VV}$ is contained in $\UU$, so that $\ol{\W}$
is contained in $\UU$ and therefore $T$ is defined in $\ol{\W}$ for all bounded subsets $\W\subset
\VV$. By the semi-Montel property of $\Fun_1$ and the continuity of $T$, we have that $\ol{\W}$ and
therefore $T(\ol{\W})$ are compact, hence the latter is bounded and thus $T(\W)$ is bounded as 
well.

\begin{proposition}\label{s2p2}
  Let $\UU\subset\C^\infty(\M)$ be convex and open in the compact-open topology, $\varphi_0\in\UU$, 
  and $F\in\Fun_{00}(\M,\UU)$ be smooth. Assume in addition that $F^{(1)}$ is locally bornological
  into $\Gamma^\infty_c(\wedge^d T^*\!\!\M\To\M)$ (Definition \ref{s2d5}). Then $F$ is microlocal if 
  and only if there is a smooth $d$-form $\omega_{F,\varphi_0}$ on $J^\infty(\M,\RR)$ such that its 
  pullback $(j^\infty\varphi)^*\omega_{F,\varphi_0}$ by the infinite jet prolongation $j^\infty\varphi$ 
  of any $\varphi\in\UU$ is a smooth $d$-form of compact support on $\M$, and
  \begin{equation}\label{s2e26}
    F(\varphi)=F(\varphi_0)+\int_\M (j^\infty\varphi)^*\omega_{F,\varphi_0}\ .
  \end{equation}
  Moreover, $\omega_{F,\varphi_0}$ depends on infinite-order jets in the sense that for each $p\in\M$ 
  there is a $r\in\NN$ such that if $\varphi_1,\varphi_2\in\UU$ are such that $j^r\varphi_1(p)=
  j^r\varphi_2(p)$, then $((j^\infty\varphi_1)^*\omega_{F,\varphi_0})(p)=((j^\infty\varphi_2)^*
  \omega_{F,\varphi_0})(p)$.
\end{proposition}

In order to prove Proposition \ref{s2p2}, we need first a preparatory lemma which is of independent
interest.

\begin{lemma}\label{s2l6}
Let $\UU\subset\C^\infty(\M)$ be open in the compact-open topology, $F\in\Fun_{00}(\M,\UU)$ smooth. 
  Then $F^{(1)}$ is a (MB-)smooth map from $\UU$ into $\Gamma^\infty_c(\wedge^d T^*\!\!\M$\\$\To\M)$ 
  if and only if $F^{(1)}$ is locally bornological into $\Gamma^\infty_c(\wedge^d T^*\!\!\M\To\M)$.
\end{lemma}
\begin{proof}
  Let $*_{\!g}$ be the Hodge star operator associated to the metric $g$ (see formula \eqref{s3e27}
  below). It is clear from Lemma \ref{s2l5} that $T=*_{\!g} F^{(1)}$ takes values in $\D(K)$ for any 
  $K\subset\M$ compact such that $\supp F\subset\In{K}$.\footnote{\label{s2f1} For simplicity, here
    we allow ourselves a slight abuse of notation -- strictly speaking, the smooth density 
    supported in $\supp F$ representing $F^{(1)}[\varphi]$ for each $\varphi\in\UU$ is only defined
    up to an exact $d$-form, so when we write $*_{\!g}F^{(1)}$ we apply $*_{\!g}$ simultaneously to
    \emph{all} representatives of $F^{(1)}[\varphi]$ for each $\varphi\in\UU$. In other words, we 
    are dealing with all $d$-forms representing $F^{(1)}[\varphi]$ simultaneously. We shall be more
    precise with this from the proof of Proposition \ref{s2p2} onwards.} Moreover, it follows from
  the MB-smoothness of $F$ that $T$ is a MB-smooth map from $\UU$ into $\E'(K)$, as argued e.g. in 
  the discussion following Definition \ref{a1d2} below. Therefore, the thesis will follow if we can
  show that $T$ is MB-smooth into $\D(K)$ if and only if $T$ is locally bornological into $\D(K)$.
  
  Due to the discussion right after Definition \ref{s2d5}, necessity of local bornology into 
  $\D(K)$ follows from the fact that any MB-smooth map is continuous and $\C^\infty(\M)$ is nuclear 
  and complete, hence semi-Montel by Proposition 4.4.7, pp. 81--82 of \cite{pietsch} and Theorem 
  3.5.1, pp. 64 of \cite{jarchow}. To get sufficiency, consider a finite open cover $\{U_1,\ldots,
  U_q\}$ of $K$ by domains of coordinate charts $\psi_i:U_i\To\RR^d$ such that $\psi_i(U_i)$ is an 
  open neighborhood of the standard unit $d$-cube $Q=[0,1]^d$ for all $i=1,\ldots,q$ and such that
  $\cup^q_{i=1}\psi_i^{-1}(\In{Q})\supset K$. Given a partition of unity $\{f_1,\ldots,f_q\}$ 
  subordinate to the open covering $\{\psi_1^{-1}(\In{Q}),\ldots,\psi_q^{-1}(\In{Q})\}$ of $K$, 
  define for each $i=1,\ldots,q$ the map $T_i:\UU\To\D(Q)$ given by
  \[
  T_i(\varphi)=(\psi_i)_*(f_iT(\varphi))\ ,\quad\varphi\in\UU\ .
  \]
  It is clear that $T_i$ is a smooth map into $\E'(Q)$ which is locally bornological into $\D(Q)$
  and $\supp T_i(\varphi)\subset\In{Q}$ for all $i=1,\ldots,q$, $\varphi\in\UU$. Moreover, thanks 
  to the latter, we have that
  \[
  T(\varphi)=\sum^q_{i=1}(\psi_i^{-1})_*T_i(\varphi)\ ,\quad\varphi\in\UU\ .
  \]
  Finally, it suffices to prove that each $T_i$ maps smooth curves in $\UU$ to smooth curves in 
  $\D(Q)$ (i.e. $T_i$ is conveniently smooth from $\UU$ into $\D(Q)$), since the above formula
  then clearly implies that $T$ maps smooth curves in $\UU$ to smooth curves in $\D(K)$. The
  sufficiency claim will follow since $\C^\infty(\M)\ni\UU$ is metrizable and $\D(K)$ is complete 
  \cite{frolicher} (see Remark \ref{a1r2} below). Convenient smoothness of $T_i$ from $\UU$ into 
  $\D(Q)$ ensues from the following two facts:
  \begin{enumerate}
  \item[(i)] Given $u\in\E'(Q)$, $\alpha=(\alpha_1,\ldots,\alpha_d)\in\ZZ^d$, define the $\alpha$-th
    Fourier coefficient of $u$ as
    \[
    \hat{u}_\alpha=u(e^{2\pi i\spr{\alpha,\cdot}})\ .
    \]
    It immediately follows that there are $k\in\NN$, $C'>0$ such that for all $\alpha\in\ZZ^d$ we 
    have
    \[
    |\hat{u}_\alpha|\leq C'(1+|\alpha|)^k\ ,\text{ where }|\alpha|=\sum^d_{i=1}|\alpha_i|\ .
    \]
    Moreover, if it happens that $u\in\D(Q)$ then for all $k\in\NN$ there is a $C_k>0$ such that for
    all $\alpha\in\ZZ^d$ we have
    \[
    |\hat{u}_\alpha|\leq C_k(1+|\alpha|)^{-k}\ .
    \]
    Conversely, if the sequence $\hat{u}=(\hat{u}_\alpha)_{\alpha\in\ZZ^d}$ of Fourier coefficients of
    $u\in\E'(Q)$ satisfies the last family of estimates above, then we must have $u\in\D(Q)$. For
    a proof of this (well known) Fourier-analytic characterization of $\D(Q)$, see e.g. Corollary 
    3.2.10 and Proposition 3.2.12, pp. 181--182 of \cite{grafakos}.
  \item[(ii)] Let $\gamma:I=[a,b]\To\CC^{\ZZ^d}$, $\gamma(t)=(\gamma_\alpha(t))_{\alpha\in\ZZ^d}$ be
    a smooth curve (that is, $\gamma_\alpha:I\To\CC$ is smooth for all $\alpha\in\ZZ^d$) such that
    for all $k\in\NN$ there is a $C_k>0$ such that 
    \[
    \|\gamma_\alpha\|_{\infty,0,I}\leq C_k(1+|\alpha|)^{-k}\text{ for all }\alpha\in\ZZ^d
    \]
    and for all $j\in\NN$ there are $k'\in\NN$, $C'_j>0$ such that
    \[
    \|\gamma_\alpha^{(j)}\|_{\infty,0,I}\leq C'_j(1+|\alpha|)^{k'}\text{ for all }\alpha\in\ZZ^d\ .
    \]
    Then for all $j,k\in\NN$ there is a $C''_{j,k}>0$ such that
    \[
    \|\gamma_\alpha^{(j)}\|_{\infty,0,I}\leq C''_{j,k}(1+|\alpha|)^{-k}\text{ for all }\alpha\in\ZZ^d\ .
    \]
    This is a consequence of the following special case of the Gagliardo-Nirenberg interpolation 
    inequality (see e.g. Theorem 5.2, pp. 135--139 of \cite{adams}): if $f:I\To\CC$ is smooth, 
    then for all $0<j<m\in\NN$ we have a constant $C=C_{j,m,I}>0$ independent of $f$ such that
    \[
    \|f^{(j)}\|_{\infty,0,I}\leq C\|f^{(m)}\|^{\frac{j}{m}}_{\infty,0,I}\|f\|^{1-\frac{j}{m}}_{\infty,0,I}\ .
    \]
    Indeed, given $j,k\in\NN$, let $m,k'\in\NN$, $C'_j>0$ so that $j<m$ and $(1+|\alpha|)^{-k'}
    \|\gamma_\alpha^{(j)}\|_{\infty,0,I}\leq C'_j$ for all $\alpha\in\ZZ^d$. The Gagliardo-Nirenberg 
    interpolation inequality entails that for all $\alpha\in\ZZ^d$
    \[
    \begin{split}
      (1+|\alpha|)^k\|\gamma^{(j)}_\alpha\|_{\infty,0,I} &\leq \left((1+|\alpha|)^{-k'}
        \|\gamma^{(m)}_\alpha\|_{\infty,0,I}\right)^{\frac{j}{m}}\\
      &\phantom{\leq\Big(}\cdot\left((1+|\alpha|)^{\frac{mk+jk'}{m-j}}
        \|\gamma_\alpha\|_{\infty,0,I}\right)^{1-\frac{j}{m}}\\
      &\leq C'_j{}^{\frac{j}{m}}C_{\frac{mk+jk'}{m-j}}^{1-\frac{j}{m}}\doteq C''_{j,k}\ .
    \end{split}
    \]
    Since $j,k$ were arbitrary, the conclusion follows.
  \end{enumerate}
  If $\gamma:\RR\To\UU$ is a smooth curve, then $\widehat{T_i\circ\gamma}\restr{[a,b]}$ clearly 
  satisfies the assumptions of (ii) for all $a<b\in\RR$, $i=1,\ldots,q$, therefore by (i) 
  $T_j\circ\gamma:\RR\To\D(K)$ is smooth for all $i=1,\ldots,q$ as desired.\hspace*{\fill}\qed
\end{proof}

We note that, unlike Proposition \ref{s2p2}, the analogous Theorem I.2 of \cite{broudlr} assumes 
MB-smoothness of $F^{(1)}$ into $\Gamma^\infty_c(\wedge^d T^*\!\!\M\To\M)$. This condition has been 
considered before in similar contexts, see for instance Appendix A of \cite{bfrej}. Local bornology
into $\Gamma^\infty_c(\wedge^d T^*\!\!\M\To\M)$, on its turn, does not seem to follow from 
microlocality alone. As a rather indirect evidence of this (in view of the proof of Lemma 
\ref{s2l6}), let us display an example of a smooth curve $\gamma$ from $[0,1]$ into the space 
$s'$ of polynomially bounded sequences which takes values in the space $s$ of rapidly decaying
sequences but fails to be bounded therein. Consider the sequence $\gamma=(\gamma_n)_{n\in\NN}$ of 
smooth curves from $[0,1]$ into $\RR$ given by
\[
\gamma_n(t)=n^2t^n(1-t)\ .
\]
Since $\gamma_n(0)=\gamma_n(1)=0$ for all $n$ and $(n^k\gamma_n(t))_{n\in\NN}$ is bounded for all
$k\in\NN$, $t\in(0,1)$, we see that $(\gamma_n(t))_{n\in\NN}\in s$ for all $t\in[0,1]$. However,
it is \emph{not} true that $(n^k\|\gamma_n\|_{\infty,0,[0,1]})_{n\in\NN}$ is bounded for all $k\in\NN$:
to see this, notice that the maximum of $\gamma_n$ takes place at the unique positive zero $t_n=
1-\frac{1}{n+1}$ of $\gamma'_n(t)=n^3t^{n-1}(1-\frac{n+1}{n}t)$ and equals $\gamma_n(t_n)=
\frac{n^2}{n+1}(1-\frac{1}{n+1})^n$. From this formula one gets that asymptotically $\gamma_n(t_n)
\sim\frac{n}{e}$ for large $n$ and therefore $(n^k\|\gamma_n\|_{\infty,0,[0,1]})_{n\in\NN}$ is unbounded
for all $k\in\NN$, as claimed. A similar argument shows, on the other hand, that $(n^{-k-1}
\|\gamma^{(k)}_n\|_{\infty,0,[0,1]})_{n\in\NN}$ is bounded for all $k\in\NN\cup\{0\}$ and therefore
$(\gamma_n)_{n\in\NN}$ is a smooth curve into $s'$. 

\begin{proof}[Proof (of Proposition \ref{s2p2})]
  Smooth functionals $F$ with compact space-time support that satisfy the representation 
  formula \eqref{s2e26} with $\omega_{F,\varphi_0}$ as above are obviously microlocal. Moreover, 
  by Lemma \ref{s2l6} $F^{(1)}$ is locally bornological into $\Gamma^\infty_c(\wedge^d T^*\!\!\M
  \To\M)$ since it is smooth therein, so we are only left with proving the opposite implication. 
  By Lemma \ref{s2l6}, $F^{(1)}$ is a (MB-)smooth map from $\UU$ into $\Gamma^\infty_c(\wedge^d 
  T^*\!\!\M\To\M)$. Since $\UU$ is assumed convex, the fundamental theorem of Calculus 
  \eqref{a1e2} yields
  \[
  F(\varphi)=F(\varphi_0)+\int^1_0\ud\lambda F^{(1)}[\varphi_0+\lambda\varphi'](\varphi')=
  F(\varphi_0)+\int^1_0\ud\lambda\int_\M \varphi' E(F)[\varphi_0+\lambda\varphi']\ ,
  \]
  where $\varphi'=\varphi-\varphi_0$ and $E(F)[\psi]$ is the smooth density of compact 
  support that represents $F^{(1)}[\psi]$. Therefore, 
  \[
  p\mapsto F_p(\varphi)\doteq\int^1_0\ud\lambda(\varphi(p)-\varphi_0(p))E(F)[\varphi_0+
  \lambda(\varphi-\varphi_0)](p)
  \]
  is our candidate for the density $(j^\infty\varphi)^*\omega_{F,\varphi_0}$, which we will 
  identify with a smooth function by a choice of a volume element on a neighborhood of 
  $\supp F$, when needed. Take now $\varphi_1,\varphi_2\in\VV\subset\UU-\varphi$ such 
  that $\varphi_1+\varphi_2\in\VV$ and $\varphi_1-\varphi_2$ vanishes together with all 
  its partial derivatives in some (hence, any) coordinate chart at some $p\in\M$, where 
  $\VV$ is an absolutely convex open neighborhood of zero. The first condition can always 
  be achieved by multiplying $\varphi_1,\varphi_2\in\VV$ by a suitably small constant -- 
  this operation does not modify the second condition. Applying the fundamental theorem of 
  Calculus \eqref{a1e2} once more, together with the Fubini-Tonelli theorem, we get
  \[
  \begin{split}
    F_p(\varphi_0&+\varphi_2)-F_p(\varphi_0+\varphi_1)=\\ &=\varphi_1(p)\int^1_0\ud\lambda\left(E(F)
      [\varphi_0+\lambda\varphi_2](p)-E(F)[\varphi_0+\lambda\varphi_1](p)\right) \\ &=\varphi_1(p)
    \int^1_0\lambda\ud\lambda\int^1_0\ud\mu E(F)^{(1)}[\varphi_0+\lambda(\varphi_1
    +\mu(\varphi_2-\varphi_1))](\varphi_2-\varphi_1)(p)\ ,
  \end{split}
  \]
  where we have also made use of the fact that $\varphi_1(p)=\varphi_2(p)$. However, for 
  each $\psi_1,\psi_2\in\VV$ such that $\psi_1+\psi_2\in\VV$, the linear map 
  \[
  \C^\infty(\M)\ni\vec{\varphi}\mapsto\int^1_0\ud\lambda\int^1_0\ud\mu E(F)^{(1)}[\varphi_0+
  \lambda(\psi_1+\mu\psi_2)](\vec{\varphi})\in\C^\infty(\M)
  \]
  decreases supports, for the integrand in the right hand side coincides with $F^{(2)}
  [\varphi_0+\lambda(\psi_1+\mu\psi_2)](\vec{\varphi},\cdot)$ in the sense of distributions 
  and $F$ is local. By Peetre's theorem \cite{peetre}, the above linear map must be a linear 
  differential operator of order $r'$ with smooth coefficients supported in $\supp F$ for 
  some $r'\in\NN$. Due to the joint continuity of $F^{(2)}$, one may take the same $r'$ for
  all $\psi_1,\psi_2\in\VV$ (possibly after suitably shrinking $\VV$). Since we have assumed 
  that $\varphi_1$ and $\varphi_2$ coincide up to infinite order at $p$, it turns out that 
  \[
  \int^1_0\ud\lambda\varphi_1(p)E(F)[\varphi_0+\lambda\varphi_1](p)=\int^1_0\ud\lambda\varphi_2(p)
  E(F)[\varphi_0+\lambda\varphi_2](p)\ ,
  \]
  hence proving the first assertion. Moreover, since $F^{(2)}[\varphi_0+\lambda(\psi_1+\mu\psi_2)]
  (\varphi,\cdot)$ is a distribution supported in $\supp F$, it must be of finite order 
  $r\in\NN$ (say) for all $\psi_1,\psi_2\in\VV$, hence we may require that $\varphi_1$ and 
  $\varphi_2$ coincide only up to order $r$ at $p$, thus proving the second assertion.
  Finally, since infinite-order jet prolongations are conveniently smooth and the infinite jet 
  bundle is metrizable, it also follows from the same reasoning employed in the proof of
  Lemma \ref{s2l6} that $\omega_{F,\varphi_0}$ is MB-smooth.
\hspace*{\fill}\qed\end{proof}

\begin{remark}\label{s2r5} 
  A consequence of Proposition \ref{s2p2} is that a microlocal functional $F$ depends on 
  derivatives of its argument $\varphi$ at each $p\in\supp F$ only up to some \emph{finite} 
  order $r\geq 0$, which can be taken to be constant on some neighborhood of $\varphi$ but 
  otherwise depending on $\varphi$, thanks e.g. to Proposition 2, pp. 355 of \cite{zajtz}.
  A natural question at this point is whether the density determined by a microlocal 
  functional $F$ is of \emph{finite} order $r$, that is, $r$ is actually 
  $\varphi$\emph{-independent}, so that \eqref{s2e26} reduces to the form \eqref{s2e22}. 
  Obviously, this is equivalent to the same question posed for the smooth density $E(F)
  [\varphi]$ representing $F^{(1)}[\varphi]$. It follows from Lemma \ref{s2l4} and the 
  fundamental theorem of Calculus \eqref{a1e2} that a necessary condition for $E(F)[\varphi]$ 
  to be of globally finite order (say) $r\in\NN$ is that for every $R\geq 0$, $k\in\NN$ there 
  is a $C>0$ such that the Lipschitz estimates
  \begin{equation}\label{s2e27}
    \|\!*_{\!g}\!E(F)[\varphi_2]-*_{\!g}E(F)[\varphi_1]\|_{\infty,k,\supp F}\leq 
    C\|\varphi_2-\varphi_1\|_{\infty,k+r,\supp F}
  \end{equation}
  hold for every $\varphi_1,\varphi_2\in\UU$ such that $\|\varphi_1-\varphi_2\|_{\infty,k+r,\supp F}<R$,
  where $\!*_{\!g}$ is the Hodge star operator associated to the metric $g$ (see \eqref{s3e27} 
  below). On the other hand, \eqref{s2e27} implies that $F^{(1)}$ is locally bornological. Moreover, 
  it follows from Lemma \ref{s2l2} that if $\UU$ is such that for every $\varphi_0\in\UU$ there is 
  a $\delta>0$ such that $\{\varphi\in\C^\infty(\M)\ |\ \|\varphi-\varphi_0\|_{\infty,r,\supp F}<\delta\}
  \subset\UU$, then these estimates are also sufficient to yield finite order (see Proposition 5 
  and Theorem 1 in \cite{zajtz} for details). Slovák proposed in \cite{slovak} a different condition
  on the domain $\UU$, related to the applicability of Whitney's extension theorem, which allows 
  one to get finite order from microlocality and convenient smoothness of $F^{(1)}$ into 
  $\Gamma^\infty_c(\wedge^d T^*\!\!\M\To\M)$ without the need of assuming \eqref{s2e27}. This
  was shown in the particular case $\UU=\C^\infty(\M)$ in \cite{broudlr}. However, as argued
  in \cite{zajtz}, Slovák's criterion seems unnatural for domains $\UU$ coming e.g. from the
  study of differential equations and flows.
\end{remark}

We close this Section with a few comments on the algebraic structure of the spaces of local
and microlocal functionals. As we have seen, in spite of the nice structure of its elements, 
$\Fun_{\loc}(\M,\UU)$ and $\Fun_{\mu\loc}(\M,\UU)$ are not closed under pointwise products. 
However, the dynamical developments in the next Section will lead, for each $\UU\subset
\C^\infty(\M)$ open in the compact-open topology, to a space of functionals which includes 
both $\Fun_{\mu\loc}(\M,\UU)$ and $\Fun_0(\M,\UU)$ and is not only closed under products, but 
will also be shown later to possess good topological properties (see Section \ref{s4-gen}). 

\section{\label{s3-dyn}Off-shell linearized dynamics}

Unlike the standard approaches to classical field theory, we will not attempt to impose equations 
of motion directly on field configurations, but instead we do this algebraically by studying the 
effect of dynamics on observable quantities. More precisely, in this Section we want to describe 
how \emph{perturbing} a given dynamics affects observables. On an infinitesimal level, this 
corresponds to endowing a sufficiently large space of observables with a Poisson structure
associated to this dynamics, which will be introduced in Subsection \ref{s3-dyn-hyp}. 

\subsection{\label{s3-dyn-lag}Preliminaries. Generalized Lagrangians and the Euler-Lagrange derivative}

Our approach to dynamics is based on a local variational principle of Euler-Lagrange type.
In order to formulate it in our context, first we need to make the representation formula 
for microlocal functionals provided by Proposition \ref{s2p2} more flexible by allowing 
the support of the functional to be prescribed at will. This is accomplished by the 
following concept, introduced in a slightly different form by Definition 6.1 of \cite{brudf}
(see also the footnote preceding Lemma \ref{s3l2} below).

\begin{definition}\label{s3d1} Let $\UU\subset\C^\infty(\M)$. A 
  \emph{generalized Lagrangian} $\Li$ on $\UU$ is a map
  \[
  \Li:\C^\infty_c(\M)\rightarrow\Fun_{00}(\M,\UU)\ ,
  \]
  such that the following properties hold:
  \begin{enumerate}
  \item $\supp(\Li(f))\subset\supp f$;
  \item $\Li(f_1+f_2+f_3)=\Li(f_1+f_2)-\Li(f_2)+\Li(f_2+f_3)$, if 
    $\supp f_1\cap\supp f_3=\varnothing$.
  \end{enumerate}
  We call the argument $f$ of $\Li(f)$ its \emph{support function}. We say that 
  $\Li$ is \emph{smooth} if $\Li(f)$ is smooth for all $f\in\C^\infty_c(\M)$.
\end{definition}

In other words, a generalized Lagrangian is additive \emph{with respect to support functions}.
As with the case with additive functionals, one can work instead with \emph{relative}
generalized Lagrangians $\Li_{f_0}$ with respect to $f_0\in\C^\infty_c(\M)$, given by
\[
\Li_{f_0}(f)\doteq\Li(f_0+f)-\Li(f_0)\ ,
\]
in terms of which the additivity property with respect to support functions reads,
for all $f_1,f_2,f_3\in\C^\infty_c(\M)$ such that $\supp f_1\cap\supp f_3=\varnothing$,
\[
\Li_{f_2}(f_1+f_3)=\Li_{f_2}(f_1)+\Li_{f_2}(f_3)\ .
\]
Moreover, one has the following result, extracted from the proof of Proposition 6.2 of
\cite{brudf}.
\begin{lemma}\label{s3l1}
  Let $\Li$ be a generalized Lagrangian. Then $\supp\Li_{f_0}(f)\subset\supp f$, 
  for all $f,f_0\in\C^\infty(\M)$.
\end{lemma}
\begin{proof}
  Let $p\not\in\supp f$, and choose $f_0'\in\C^\infty_c(\M)$ such that $f_0'\equiv f_0$ 
  in a neighborhood of $p$ and $\supp f\cap\supp f_0'=\varnothing$. By additivity of 
  $\Li$ with respect to support functions, we have that $\Li_{f_0}(f)=\Li_{f_0-f_0'}(f)$,
  which implies that $\supp\Li_{f_0}(f)\subset\supp(f+f_0-f_0')\cup\supp(f_0-f_0')$. 
  Therefore, $p\not\in\supp\Li_{f_0}(f)$, as asserted.
\end{proof}

Additivity with respect to support functions is a weak substitute for linearity, but is 
strong enough to yield useful consequences. One of them is that the argument involving 
field configurations inherits this property\footnote{This property is assumed 
  \emph{a priori} in Definition 6.1 of \cite{brudf}.}: 

\begin{lemma}\label{s3l2}
  Let $\UU\subset\C^\infty(\M)$, and $\Li$ a generalized Lagrangian on $\UU$. 
  Then, for all $f\in\C^\infty_c(\M)$, $\Li(f)$ is additive. 
\end{lemma}
\begin{proof}
  Fix an arbitrary $f\in\C^\infty(\M)$, and let $\varphi_2\in\UU$, $\varphi_1,\varphi_3
  \in\UU-\varphi_2$ be such that $\supp\varphi_1\cap\supp\varphi_3=\varnothing$. Let 
  $\chi_1,\chi_3\in\C^\infty(\M)$ be such that $\chi_j\equiv 1$ in a neighborhood of 
  $\supp\varphi_j$, $j=1,3$, and $\supp\chi_1\cap\supp\chi_3=\varnothing$. Define 
  $f_1\doteq\chi_1f$, $f_3\doteq\chi_3f$, and $f_2\doteq f-f_1-f_3$. Then, by properties 
  (1) and (2) in Definition \ref{s3d1},
  \[
  \Li(f)(\varphi_1+\varphi_2+\varphi_3)=\Li(f_1+f_2)(\varphi_1+\varphi_2)-
  \Li(f_2)(\varphi_2)+\Li(f_2+f_3)(\varphi_2+\varphi_3)\ .
  \]
  However, we also have that
  \begin{align*}
    \Li(f)(\varphi_1+\varphi_2) &= \Li(f_1+f_2)(\varphi_1+\varphi_2)-\Li(f_2)(\varphi_2)
                                  +\Li(f_2+f_3)(\varphi_2)\ ,\\
    \Li(f)(\varphi_2) &= \Li(f_1+f_2)(\varphi_2)-\Li(f_2)(\varphi_2)
                        +\Li(f_2+f_3)(\varphi_2)\ ,\\
    \Li(f)(\varphi_2+\varphi_3) &= \Li(f_1+f_2)(\varphi_2)-\Li(f_2)(\varphi_2)
                                  +\Li(f_2+f_3)(\varphi_2+\varphi_3)\ ,
  \end{align*}
  whence it follows that 
  \begin{equation*}
    \begin{split}
      \Li(f)(\varphi_1+\varphi_2)&-\Li(f)(\varphi_2)+\Li(f)(\varphi_2+\varphi_3) \\
      &=\Li(f_1+f_2)(\varphi_1+\varphi_2)-\Li(f_2)(\varphi_2)
      +\Li(f_2+f_3)(\varphi_2+\varphi_3) \\
      &=\Li(f)(\varphi_1+\varphi_2+\varphi_3)\ , 
    \end{split}
  \end{equation*}
  which proves our assertion.
\end{proof}

\begin{corollary}\label{s3c1}
  Let $\UU\subset\C^\infty(\M)$ be open with respect to the compact-open topology, and
  $\Li$ be a \emph{smooth} generalized Lagrangian on $\UU$. Then $\Li(f)\in\Fun_{\loc}
  (\M,\UU)$ for all $f\in\C^\infty_c(\M)$.
\end{corollary}
\begin{proof}
  Apply Proposition \ref{s2p1} to the outcome of Lemma \ref{s3l2}.
\end{proof}

Another consequence is the following generalization of Lemma \ref{s2l4}
to \emph{any} open subset of $\C^\infty(\M)$:

\begin{lemma}\label{s3l3}
  Let $\UU\subset\C^\infty(\M)$, and $\Li$ be a generalized Lagrangian
  on $\UU$. Then, for any $f\in\C^\infty_c(\M)$ fixed, $\Li(f)$ can be written
  as a finite sum of additive functionals of arbitrarily small space-time support.
\end{lemma}
\begin{proof}
  Let $(\chi_i)_{i=1,\ldots,n}$ a the partition of unity subordinated to the finite
  open covering of $\supp f$ constructed in the proof of Lemma \ref{s2l4}. Then 
  \[
  \Li(f)=\Li\left(\sum^n_{i=1}\chi_if\right)\ .
  \]
  Applying additivity of $\Li$ with respect to support functions just as we did 
  in the proof of Lemma \ref{s2l4} yields the desired result.
\end{proof}


Motivated by Corollary \ref{s3c1}, we say that a generalized Lagrangian $\Li$ is 
\emph{microlocal} if $\Li(f)\in\Fun_{\mu\loc}(\M,\UU)$ for all $f\in\C^\infty_c(\M)$, 
and \emph{of (finite) order} $r\geq 0$ if, in addition, $\Li(f)$ is of finite order 
$r\in\NN$ for all such $f$. A simple but important example of microlocal generalized 
Lagrangians of order $r$ are the squares of the local Sobolev seminorms \eqref{s2e5} 
at order $k=r$
\begin{equation}\label{s3e1}
  \Li(f)(\varphi)=\|\varphi\|^2_{2,r,f}\ .
\end{equation}
With the concept of microlocal generalized Lagrangian at hand, we can write down the 
Euler-Lagrange variational principle in the form we will use. 

\begin{definition}\label{s3d2}
  Let $\UU\subset\C^\infty(\M)$ be open in the compact-open topology, $\Li$ a smooth generalized 
  Lagrangian, $k\geq 1$. The $k$-th order \emph{Euler-Lagrange derivative} of $\Li$ at 
  $\varphi\in\UU$ along $\vec{\varphi}_1,\ldots,\vec{\varphi}_k\in\C^\infty(\M)$ is given by
  \[
  D^k\Li(1)[\varphi](\vec{\varphi}_1,\ldots,\vec{\varphi}_k)=
  D^k\Li(f)[\varphi](\vec{\varphi}_1,\ldots,\vec{\varphi}_k)\ ,
  \]
  where $f\in\C^\infty_c(\M)$ satisfies $f\equiv 1$ on $\supp\vec{\varphi}_j$ for at least
  one $j=1,\ldots,k$ (due to Proposition \ref{s2p1} and Lemma \ref{s3l1}, the above definition 
  is independent of the choice of $f$). If $\Li$ is microlocal of finite order and $\supp
  \vec{\varphi}_1$ is compact, we have that for $k=1$,
  \[
  D\Li(1)[\varphi](\vec{\varphi_1})=\Spr{E(\Li)[\varphi],\vec{\varphi}_1}
  \]
  defines a partial differential operator $E(\Li):\UU\To\Gamma^\infty(\wedge^d T^*\!\!\M\To\M)$, 
  called the \emph{Euler-Lagrange operator} associated to $\Li$. The map $E(\Li)$ is
  clearly smooth, with derivatives of order $k\geq 1$ at $\varphi\in\UU$ along 
  $\vec{\varphi}_2,\ldots,\vec{\varphi}_{k+1}\in\C^\infty(\M)$ given by the identity
  \[
  \int_\M\vec{\varphi}_1D^kE(\Li)[\varphi](\vec{\varphi}_2,\ldots,\vec{\varphi}_{k+1})
  =D^{k+1}\Li(1)[\varphi](\vec{\varphi}_1,\vec{\varphi}_2,\ldots,\vec{\varphi}_{k+1})\ .
  \]
  For $\varphi\in\UU$ fixed, the maps $D^kE(\Li)[\varphi]:\otimes^k\C^\infty(\M)\To
  \Gamma^\infty(\wedge^d T^*\!\!\M\To\M)$ are (symmetric) $k$-linear $k$-differential
  operators (i.e. for each $j=2,\ldots,k+1$, $D^kE(\Li)[\varphi](\vec{\varphi}_2,\ldots,
  \vec{\varphi}_j,\ldots,\vec{\varphi}_{k+1})$ is a linear partial differential operator 
  acting on $\vec{\varphi}_j$ with all other arguments fixed). We call 
  \[
  E'(\Li)[\varphi]=DE(\Li)[\varphi]
  \]
  the \emph{linearized Euler-Lagrange operator} around $\varphi\in\UU$.
\end{definition}

For notational convenience, we occasionally write 
\[
\begin{split}
  E'(\Li)[\varphi](\vec{\varphi})&=E'(\Li)[\varphi]\vec{\varphi}\ ,\\
  D^kE(\Li)[\varphi](\vec{\varphi}_1,\ldots,\vec{\varphi}_k)&=
  D^kE(\Li)[\varphi](\vec{\varphi}_2,\ldots,\vec{\varphi}_k)\vec{\varphi}_1\ ,\quad k>1\ .
\end{split}
\]
Definition \ref{s3d2} prompts us to compare it with the standard formulation of the 
Euler-Lagrange variational principle in field theory \cite{kolar}. We sketch this 
comparison below. Our definition of Euler-Lagrange derivatives is tailored to get 
rid of boundary terms automatically; to make them appear, let $\Li(f)$ be a microlocal 
generalized Lagrangian which depends \emph{linearly} on the supporting function $f$. 
It follows from Peetre's theorem \cite{peetre} that $D\Li(f)[\varphi]$ is a linear 
partial differential operator acting on $f$ for each fixed $\varphi$, taking values 
on $\Gamma^\infty(\wedge^d T^*\!\!\M\To\M)$. Let now $f$ converge to the characteristic 
function $\chi_K$ of a compact region $K$ of $\M$ with smooth boundary $\dd K$ -- 
the part of $D\Li(f)[\varphi](\vec{\varphi})$ proportional to the term of zeroth 
order in $f$ yields
\[
\int_K \vec{\varphi}_1E(\Li)[\varphi]\ ,
\]
and the remaining terms become the integral over $\dd K$ of the Poincaré-Cartan 
$(d-1)$-form $\Theta[\varphi]$ associated to the action integral $\Li(\chi_K)$ over 
$K$. If $\Li$ is of order $r$, one can show \cite{kolar} that $E(\Li)$ has order at 
most $2r$. Therefore, Definition \ref{s3d2} does provide a generalization of the 
Euler-Lagrange variational principle. If $E(\Li)[\varphi]=0$, then one recovers
the usual formula for the on-shell variation of the action functional in terms of 
the integral of $\Theta[\varphi]$ over $\dd K$, which is of importance in the 
so-called \emph{covariant phase space formalism} for field theory (see e.g. 
formulae (94), pp. 398 of \cite{forgerr} and (6.24), pp. 114 of \cite{helein1}).

The role in our setup of Lagrangians which are total divergences (also called \emph{null 
  Lagrangians} in the literature, see e.g. Section 3.2 of \cite{christo}) is played by 
the following

\begin{definition}\label{s3d3}
  Let $\UU\subset\C^\infty(\M)$. A generalized 
  Lagrangian $\Li$ on $\UU$ is said to be \emph{trivial} if $\supp\Li(f)\subset
  \supp(\ud f)$ for all $f\in\C^\infty_c(\M)$. Two generalized Lagrangians $\Li_1,\Li_2$ 
  are said to be \emph{equivalent} if $(\Li_1-\Li_2)(f)\doteq\Li_1(f)-\Li_2(f)$ is
  trivial. This is clearly an equivalence relation in the space of all generalized
  Lagrangians. If $\UU$ is open in the compact-open topology and $\Li$ is a microlocal 
  generalized Lagrangian of order $r$, its equivalence class $S_\Li$ in the space of 
  all microlocal generalized Lagrangians of order $r$ is called an \emph{action 
    functional of order} $r$.
\end{definition}

Trivial generalized Lagrangians are thus called because they obviously have vanishing 
Euler-Lagrange derivatives of all orders whenever they are defined. Therefore, two equivalent 
generalized Lagrangians have the same Euler-Lagrange derivatives. In particular, the action 
functional $S_\Li$ associated to a microlocal generalized Lagrangian $\Li$ of finite order 
uniquely determines the Euler-Lagrange operator $E(\Li)$. As a typical class of examples of 
trivial generalized Lagrangians, we may take
\[
\Li(f)[\varphi]=\int_\M \ud f\wedge (j^r\varphi)^*\omega
\]
with $\omega\in\Gamma^\infty(\wedge^{d-1}T^*\!J^r(\M,\RR)\To J^r(\M,\RR))$. 

To briefly illustrate the relation of trivial generalized Lagrangians with the more standard 
notion of null lagrangians, consider once more a microlocal generalized Lagrangian $\Li(f)$ 
which depends linearly on the supporting function $f$. The reasoning preceding Definition 
\ref{s3d3} shows that $D\Li(f)[\varphi]$ can be written as
\[
D\Li(f)[\varphi]=\Theta[\varphi]\wedge\ud f+\ud\Xi(f)[\varphi]\ ,
\]
where $\Xi(f)[\varphi]$ is a smooth $(d-1)$-form supported in $\supp f$. If $\Li(f)$ is of
finite order (say, $r$) and \emph{trivial} (e.g. the example written in the previous paragraph),
making $f$ converge to $\chi_K$ as before shows that $D\Li(f)[\varphi](\vec{\varphi})$ converges 
to the integral of $\Theta[\varphi]$ over $\dd K$ alone.

\subsection{\label{s3-dyn-hyp}Normally hyperbolic Euler-Lagrange operators. Infinitesimal solvability and the Peierls bracket}

As discussed in the Introduction, we are mainly interested in relativistic classical
field theories. This means that the action functional determining the dynamics must
give rise to Euler-Lagrange equations of motion which are \emph{hyperbolic}. There
are several different concepts of hyperbolicity for partial differential operators
(see for instance \cite{christo}); the one we use is the notion of \emph{normal
  hyperbolicity}, as defined for instance in \cite{bgp} for linear partial differential
operators. For future convenience, the discussion in the linear case takes place in 
the wider context of smooth sections of vector bundles.

\begin{definition}\label{s3d4}
  Let $\pi:\E\To\M$ be a real vector bundle of rank $D$ over the space-time manifold 
  $\M$. A linear partial differential operator of second order $P:\Gamma^\infty(\pi)
  \To\Gamma^\infty(\pi)$ acting on $\Gamma^\infty(\pi)$ is said to be \emph{normally hyperbolic} 
  if its principal symbol $\hat{p}\in\Gamma^\infty(\vee^2 T\!\M\otimes\E'\otimes\E\To\M)$, 
  given by
  \[
  \frac{1}{2}P((f-f(x))^2\vec{\varphi})(x)\doteq\hat{p}(x,\ud f(x))\vec{\varphi}(x)\ ,
  \]
  ($f\in\C^\infty(\M),\vec{\varphi}\in\Gamma^\infty(\pi)$) is of the form 
  \[
  \hat{p}(x,\xi)=\hat{g}^{-1}(x)(\xi,\xi)\otimes\id_{\pi^{-1}(x)}\ ,\quad x\in\M\ ,\,\xi\in T^*_x\M\ ,
  \]
  where $\hat{g}$ is a Lorentzian metric on $\M$. 
\end{definition}

We remark that a linear partial differential operator $P$ is normally hyperbolic if
and only if $P$ is regularly hyperbolic in the sense of Christodoulou \cite{christo}
and has a scalar principal symbol. 

Any second-order linear partial differential operator $P:\Gamma^\infty(\pi)\To\Gamma^\infty(\pi)$
can be written in a coordinate-invariant fashion as follows. If we define iterated covariant
derivatives of smooth sections of $\pi$ with respect to some connection $\nabla$ (see Remark
\ref{s2r1}), $P$ assumes the form
\begin{equation}\label{s3e2}
  P\vec{\varphi}=\hat{p}\nabla^2\vec{\varphi}+A\nabla\vec{\varphi}+B\vec{\varphi}\ ,
\end{equation}
where $A\in\Gamma^\infty(T\!\M\otimes\E'\otimes\E\To\M)$, $B\in\Gamma^\infty(\E'\times\E\To\M)$
and $\hat{p}\in\Gamma^\infty(\vee^2T\!\M\otimes\E'\otimes\E\To\M)$ is the principal symbol. We
remark that, unlike $A$ and $B$, $\hat{p}$ is independent of the choice of $\nabla$. 

Before we continue, we introduce a strict partial order $<$ and a partial order 
$\lesssim$ in the space $\mathrm{Lor}^0(\M)$ of continuous Lorentzian metrics on 
$\M$. Let $g_1,g_2\in\mathrm{Lor}^0(\M)$; we say that 
\begin{equation}\label{s3e3}
  \begin{split}
    g_1<g_2 &\quad\text{if}\quad g_1(X,X)\leq 0\quad\text{implies}\quad g_2(X,X)<0\ ;\\
    g_1\lesssim g_2 &\quad\text{if}\quad g_1(X,X)< 0\quad\text{implies}\quad g_2(X,X)<0\ ,
  \end{split}
\end{equation}
for all $X\in T\!\M$. As usual, we write $g_1>g_2$ (resp. $g_1\gtrsim g_2$) if $g_2<g_1$
(resp. $g_2\lesssim g_1$). By continuity, $g_1\lesssim g_2$ implies that $g_2(X,X)\leq 0$ 
for all $X$ such that $g_1(X,X)\leq 0$ (the converse is not necessarily true). Both
partial orders clearly enjoy the property that if $g_1<g_2$ (resp. $g_1\lesssim g_2$),
then $\Omega_1 g_1<\Omega_2 g_2$ (resp. $\Omega_1 g_1\lesssim\Omega_2 g_2$) for all positive,
real-valued continuous functions $\Omega_1,\Omega_2$ on $\M$. In other words, $<$ and 
$\lesssim$ depend only on the conformal classes (hence, only on the causal structures) 
of $g_1$ and $g_2$. As shown by Lerner \cite{lerner}, the order topology on $\mathrm{Lor}^0(\M)$ 
associated to $<$ (i.e. the topology generated by the open intervals $\{g\ |\ g_1<g<g_2\}$
as $g_1,g_2$ run through $\mathrm{Lor}^0(\M)$), called the \emph{interval topology} on
$\mathrm{Lor}^0(\M)$, coincides with the latter's relative graph (Whitney) topology.
Moreover, Benavides Navarro and Minguzzi have shown \cite{benavs} (building on earlier
results by Geroch \cite{geroch}) that, given $g$ globally hyperbolic, there is $g_2>g$ 
such that $g_2$ is also globally hyperbolic. We shall use this fact to prove the 
following useful result:

\begin{lemma}\label{s3l4}
  The space of continuous, time-oriented and globally hyperbolic Lorentzian metrics on 
  $\M$ is an open subset of $\mathrm{Lor}^0(\M)$ in the interval topology (hence also 
  in the Whitney topology). Moreover, given any such metric $g_2$, all $g_1\in\mathrm{Lor}^0(\M)$ 
  such that $g_1\lesssim g_2$ are also globally hyperbolic and have the same time orientation 
  as $g_2$, and any Cauchy time function with respect to $g_2$ is also a Cauchy time 
  function with respect to $g_1$.
\end{lemma}
\begin{proof}
  Notice that if $g_1\lesssim g_2$ and $g_2$ is globally hyperbolic, then any Cauchy 
  hypersurface in $\M$ with respect to $g_2$ is also a Cauchy hypersurface with respect 
  to $g_1$, therefore $g_1$ is globally hyperbolic as well. The results of Lerner, 
  Benavides Navarro and Minguzzi quoted above then imply that any globally hyperbolic $g$ 
  is contained in the open interval $\{g'\ |\ g_1<g'<g_2\}$ for some pair $g_1,g_2\in
  \mathrm{Lor}^0(\M)$ such that $g_2$ is also globally hyperbolic. By the above reasoning, 
  any $g'$ in this set is globally hyperbolic as well. In particular, if $\tau$ is a 
  Cauchy time function on $\M$ with respect to $g_2$, then $\tau$ is also a Cauchy time 
  function with respect to any $g_1\lesssim g_2$ -- notice that \eqref{s3e3} implies that 
  if the tangent vector $X$ is spacelike with respect to $g_2$, then it is also spacelike 
  with respect to $g_1$; therefore $\ud\tau$ is a timelike covector field with respect to 
  $g_1$, since it is normal to the tangent bundle of all level sets of $\tau$, whose 
  elements must be all spacelike with respect to $g_1$. Finally, if $T_1=g_1^\sharp(\ud\tau)$ 
  and $T_2=g_2^\sharp(\ud\tau)$, where $g_1\lesssim g_2$ are time oriented and $\tau$ is a 
  Cauchy time function with respect to $g_2$, then $g_1(T_1,T_2)=\ud\tau(T_2)=g_2(T_2,T_2)<0$ 
  and $g_2(T_1,T_2)=\ud\tau(T_1)=g_1(T_1,T_1)<0$. In particular, if $T_1$ is future 
  directed with respect to $g_1$, then it is also future directed with respect to $g_2$. 
\end{proof}

Lemma \ref{s3l4} and its proof obviously extend to smooth metrics. Let now $P$ be 
a normally hyperbolic linear partial differential operator on $\Gamma^\infty(\pi)$. 
We assume the \emph{working hypothesis} ($\mathsf{NH}_g$) on $P$, given as follows:

\begin{enumerate}
\item[($\mathsf{NH}_g$)] The Lorentzian metric $\hat{g}$ on $\M$ associated to the 
  principal symbol $\hat{p}$ of $P$ satisfies $\hat{g}\lesssim g$.
\end{enumerate}

By the above discussion, all such $\hat{g}$'s are globally hyperbolic and have the
same time orientation as $g$. Moreover, by Lemma \ref{s3l4} these implications of 
($\mathsf{NH}_g$) are stable under perturbations of $\hat{g}$ in the interval topology, 
a fact that is also useful when dealing with nonlinear dynamics.

For $P$ normally hyperbolic and satisfying ($\mathsf{NH}_g$), one can prove the following 
fact, which is a restatement of results in \cite{bgp} (related partial results for the scalar 
case may be found e.g. in \cite{helein2}).

\begin{theorem}\label{s3t1}
  Let $(\M,g)$ be a globally hyperbolic space-time, and $\E\To\M$ be a real vector bundle 
  of rank $D$ over the space-time manifold $\M$, endowed with a connection $\nabla$. We 
  assume that $T\!\M$ is endowed with the Levi-Civita connection associated to the 
  space-time metric $g$. Let $P$ be a normally hyperbolic linear partial differential 
  operator on $\Gamma^\infty(\pi)$ satisftying \emph{($\mathsf{NH}_g$)}. Let $\Sigma$ be a 
  Cauchy hypersurface for $(\M,g)$, with future directed timelike normal $n\in\Gamma^\infty
  (T_\Sigma\M\To\M)$ (i.e. $g(n,n)=-1$ and $g(n,X)=0$ for all $X\in T\Sigma$), suitably 
  extended to an open neighborhood of $\Sigma$ in $\M$ (the exact form of the extension 
  is irrelevant for what follows). Given $\vec{\varphi}\in\Gamma^\infty(\pi)$, define 
  \begin{align}
    \rho_0^\Sigma(\vec{\varphi}) &=\vec{\varphi}\restr{\Sigma}\ ,\label{s3e4}\\
    \rho_1^\Sigma(\vec{\varphi}) &=(\nabla_n\vec{\varphi})\restr{\Sigma}\ .\label{s3e5}
  \end{align}
  Then for every $\vec{\varphi}_0,\vec{\varphi}_1\in\Gamma^\infty(\pi\restr{\Sigma})$, 
  $\psi\in\Gamma^\infty(\pi)$, there is a unique $\vec{\varphi}\in\Gamma^\infty(\pi)$ 
  such that 
  \begin{equation}\label{s3e6}
    \begin{split}
      P\vec{\varphi} &=\vec{\psi}\ ,\\
      \rho_j^\Sigma(\vec{\varphi}) &=\vec{\varphi}_j\ ,\quad j=0,1\ .
    \end{split}
  \end{equation}
  In other words, the map $\Phi:\Gamma^\infty(\pi)\To\Gamma^\infty(\pi)\oplus
  \Gamma^\infty(\pi\restr{\Sigma})\oplus\Gamma^\infty(\pi\restr{\Sigma})$ given by
  \begin{equation}\label{s3e7}
    \Phi(\vec{\varphi})=(P\vec{\varphi},\rho_0^\Sigma(\vec{\varphi}),\rho_1^\Sigma
    (\vec{\varphi}))
  \end{equation}
  is a linear isomorphism.\qed
\end{theorem}

We stress that $\Phi$ is even a \emph{topological} linear isomorphism with respect to the
standard Fréchet space topology on spaces of smooth sections of vector bundles. 

Let $\Psi$ be the inverse of $\Phi$. By the principle of superposition, one can write 
\begin{equation}\label{s3e8}
  \Psi(\vec{\psi},\vec{\varphi}_0,\vec{\varphi}_1)=K_P^{\Sigma,0}\vec{\varphi}_0+K_P^{\Sigma,1}
  \vec{\varphi}_1+\Delta_P^\Sigma\vec{\psi}\ ,
\end{equation}
where $K_P^{\Sigma,j}\vec{\varphi}_j$, $j=0,1$ is the unique solution of the 
initial value problem
\begin{equation}\label{s3e9}
  \begin{cases}
    P\vec{\varphi} &=0\ ,\\
    \rho_{1-j}^\Sigma(\vec{\varphi}) &=0\ ,\\
    \rho_j^\Sigma(\vec{\varphi}) &=\vec{\varphi}_j\ ,
  \end{cases}
\end{equation}
and $\Delta_P^\Sigma\vec{\psi}$ is the unique solution of the initial value problem
\begin{equation}\label{s3e10}
  \begin{cases}
    P\vec{\varphi} &=\vec{\psi}\ ,\\
    \rho_0^\Sigma(\vec{\varphi}) &=0\ ,\\
    \rho_1^\Sigma(\vec{\varphi}) &=0\ .
  \end{cases}
\end{equation}
In the scalar case, there is the following refinement of Theorem \ref{s3t1}, which is a 
restatement of Theorem 5.1.6 of \cite{duister} that, on its turn, tells us in great detail
how supports and singularities propagate under the dynamics associated to $P$.

\begin{theorem}\label{s3t2}
  Assume the hypotheses and definitions of Theorem \ref{s3t1}. Suppose that $\E=\M\times\RR$ and 
  $\pi(p,\lambda)=\pr_1(p,\lambda)=p$ for $p\in\M$, $\lambda\in\RR$, identifying $\Gamma^\infty(\pi)$
  with $\C^\infty(\M)$. Then $\Delta_P^\Sigma:\C^\infty(\M)\To\C^\infty(\M)$, $K_P^{\Sigma,0}:
  \C^\infty(\Sigma)\To\C^\infty(\M)$ and $K_P^{\Sigma,1}:\C^\infty(\Sigma)\To\C^\infty(\M)$ 
  satisfy the following properties:
  \begin{enumerate}
  \item[(a)] Continuity: $K_P^{\Sigma,j}$ is a (continuous) linear map which admits a continuous 
    linear extension to the space $\D'(\Sigma)$ of distributions on $\Sigma$ for $j=0,1$, and 
    $\Delta_P^\Sigma$ is a (continuous) linear map which admits a continuous\footnote{Here 
      ``continuous'' means \emph{sequentially} continuous with respect to the (weak) Hörmander 
      topology on the extended domain (see Subsection \ref{s4-gen-top} below), as shown e.g. by 
      Theorem 8.2.13, pp. 268--269 of \cite{horm1} and, more precisely, by Theorems 8.2.9. (iii) 
      and 8.2.10, pp. 515--520 of \cite{chazp}. One can see indirectly from the arguments in 
      \cite{broudh2} that one \emph{cannot} hope to upgrade this result to full continuity, unless 
      one uses instead the \emph{strong} Hörmander topology (see also Remark \ref{s4r3} below).} 
    linear extension to
    \begin{equation}\label{s3e11}
      \D'_\Sigma(\M)=\{v\in\D'(\M)\ |\ \WF(v)\cap N^*\Sigma=\varnothing\}\ ,
    \end{equation}
    where $\WF(v)$ denotes the wave front set of $v$ and $N^*\Sigma=\{\xi\in T^*_\Sigma\M\ |\ \xi(X)
    =0\text{ for all }X\in T\Sigma\}$ denotes the conormal bundle of $\Sigma$. We remark that 
    the continuous linear maps $\rho_j^\Sigma:\C^\infty(\M)\To\C^\infty(\Sigma)$ also admit a 
    continuous\footnotemark[\value{footnote}] linear extension to $\D'_\Sigma(\M)\ni v$, satisfying 
    for $j=0,1$ \cite{horm1}
    \begin{equation}\label{s3e12}
      \WF(\rho_j^\Sigma(v))=\{(x,\xi\restr{T\Sigma})\in T^*\Sigma\ |\ (x,\xi)\in\WF(v)\}\ .
    \end{equation}
  \item[(b)] Propagation of supports:
    \begin{equation}\label{s3e13}
      \supp (K_P^{\Sigma,j}u_j)\subset J^+(\supp u_j,\hat{g})\cup J^-(\supp u_j,\hat{g})
      \subset J^+(\supp u_j,g)\cup J^-(\supp u_j,g)
    \end{equation}
    and
    \begin{equation}\label{s3e14}
      \supp (\Delta_P^\Sigma v) \subset J^+(\supp v\cap J^+(\Sigma,\hat{g}),\hat{g})\cup 
      J^-(\supp v\cap J^-(\Sigma,\hat{g}),\hat{g})\ , 
    \end{equation}
    for all $u_j\in\D'(\Sigma)$, $v\in\D'_\Sigma(\M)$, $j=0,1$.
  \item[(c)] Propagation of singularities: given any $u_j\in\D'(\Sigma)$, $j=0,1$, we 
    have that $(x,\xi)\in\WF(K_P^{\Sigma,j}u_j)$ \emph{only if} there is $\lambda>0$ and
    a null geodesic segment $\gamma:[0,\Lambda]\To\M$ with respect to $\hat{g}$ 
    (i.e. $\hat{g}(\dot{\gamma}(\lambda),\dot{\gamma}(\lambda))=0$ for all 
    $\lambda\in[0,\Lambda]$) such that if
    \[
    E^{\hat{g}}_\gamma=\{(\gamma(0),\hat{g}^\flat(\dot{\gamma}(0))),(\gamma(\Lambda),
    \hat{g}^\flat(\dot{\gamma}(\Lambda)))\}\subset T^*\!\!\M
    \]
    is the set of endpoints of the bicharacterstic strip $\{(\gamma(\lambda),
    \hat{g}^\flat(\dot{\gamma}(\lambda)))\in T^*\!\!\M\ |\ \lambda\in[0,\Lambda]\}$, 
    then $(x',\xi'\restr{T\Sigma})\in\WF(u_j)$ for some $(x',\xi')\in E^{\hat{g}}_\gamma$ 
    and $(x,\xi)\in E^{\hat{g}}_\gamma$. Given any $v\in\D'_\Sigma(\M)$, we have that 
    $(x,\xi)\in\WF(\Delta_P^\Sigma v)$ \emph{only if} either $(x,\xi)\in\WF(v)$ or
    there is $\Lambda>0$ and a null geodesic segment $\gamma:[0,\Lambda]\To\M$ with 
    respect to $\hat{g}$ such that $\gamma((0,\Lambda))\cap\Sigma=\varnothing$ and
    $\WF(v)\cap E^{\hat{g}}_\gamma\neq\varnothing$, $(x,\xi)\in E^{\hat{g}}_\gamma$.
  \end{enumerate}
  In particular, given any $u_0,u_1\in\D'(\Sigma)$, $v\in\D'_\Sigma(\M)$, we 
  have that $K_P^{\Sigma,j}u_j$ and $\Delta_P^\Sigma v$ belong to $\D'_\Sigma(\M)$. 
  We have that $u=K_P^{\Sigma,j}u_j$, $j=0,1$ is the unique solution in $\D'_\Sigma(\M)$ 
  of the initial value problem
  \begin{equation}\label{s3e15}
    \begin{cases}
      Pu &=0\ ,\\
      \rho_{1-j}^\Sigma(u) &=0\ ,\\
      \rho_j^\Sigma(u) &=u_j\ ,
    \end{cases}
  \end{equation}
  and $u=\Delta_P^\Sigma v$ is the unique solution in $\D'_\Sigma(\M)$ of the initial 
  value problem
  \begin{equation}\label{s3e16}
    \begin{cases}
      Pu &=v\ ,\\
      \rho_0^\Sigma(u) &=0\ ,\\
      \rho_1^\Sigma(u) &=0\ .
    \end{cases}
  \end{equation}\qed
\end{theorem}

We note that part (b) of Theorem \ref{s3t2} is actually stronger than that provided
by Theorem 5.1.6 of \cite{duister} but it can be derived from energy estimates for $P$.
There are two particular cases of the initial value problem \eqref{s3e10} that deserve 
special attention:

\begin{enumerate}
\item[(R)] $\supp\vec{\psi}\subset I^+(\Sigma,g)$ -- The restriction of $\Delta_P^\Sigma$ to 
  the space of smooth sections of $\E$ with \emph{past compact} support with respect to $g$
  \begin{equation}\label{s3e17}
    \begin{split}
      \Gamma^\infty_{+}(\pi,g) &=\{\vec{\psi}\in\Gamma^\infty(\pi)\ |\ \forall p\in\M,\,\\
      &\phantom{=\{}J^-(p,g)\cap\supp\vec{\psi}\text{ is compact}\}\\
      &=\{\vec{\psi}\in\Gamma^\infty(\pi)\ |\ \forall K\subset\M\text{ compact, }\\
      &\phantom{=\{}J^-(K,g)\cap\supp\vec{\psi}\text{ is compact}\}
    \end{split}
  \end{equation}
  no longer depends on $\Sigma$, as long as condition (R) is satisfied. In this case 
  we write $\Delta_P^\Sigma=\Delta^\Rt_P$, calling it the \emph{retarded fundamental solution} 
  of $P$.
\item[(A)] $\supp\vec{\psi}\subset I^-(\Sigma,g)$ -- The restriction of $\Delta_P^\Sigma$ to 
  the space of smooth sections of $\E$ with \emph{future compact} support with respect to $g$
  \begin{equation}\label{s3e18}
    \begin{split}
      \Gamma^\infty_{-}(\pi,g) &=\{\vec{\psi}\in\Gamma^\infty(\pi)\ |\ \forall p\in\M,\,\\
      &\phantom{=\{}J^+(p,g)\cap\supp\vec{\psi}\text{ is compact}\}\\
      &=\{\vec{\psi}\in\Gamma^\infty(\pi)\ |\ \forall K\subset\M\text{ compact,}\\
      &\phantom{=\{}J^+(K,g)\cap\supp\vec{\psi}\text{ is compact}\}
    \end{split}
  \end{equation}
  no longer depends on $\Sigma$ either, as long as condition (A) is satisfied. In this case 
  we write $\Delta_P^\Sigma=\Delta^\Av_P$, calling it the \emph{advanced fundamental solution} 
  of $P$.
\end{enumerate}

The difference $\Delta_P=\Delta^\Rt_P-\Delta^\Av_P:\Gamma^\infty_{+}(\pi,g)\cap
\Gamma^\infty_{-}(\pi,g)\To\Gamma^\infty(\pi)$ is called the \emph{causal propagator} 
of $P$. We obviously have the identity $P\circ\Delta_P=\Delta_P\circ P=0$ wherever
it is defined. 

In the scalar case discussed in Theorem \ref{s3t2}, $\Delta^\Rt_P$ (resp. $\Delta^\Av_P$) 
is defined on the space of smooth functions on $\M$ with past (resp. future) compact 
support with respect to $g$
\begin{equation}\label{s3e19}
  \begin{split}
    \C^\infty_{+/-}(\M,g) &=\{\psi\in\C^\infty(\M)\ |\ \forall p\in\M,J^{-/+}(p,g)\cap\supp
    \psi\text{ is compact}\}\\
    &=\{\psi\in\C^\infty(\M)\ |\ \forall K\subset\M\text{ compact},J^{-/+}(K,g)\cap\supp
    \psi\text{ is compact}\}\ .
  \end{split}
\end{equation}
Specializing Theorem \ref{s3t2} to these two cases yields the

\begin{corollary}\label{s3c2}
  Let the hypotheses and notation of Theorem \ref{s3t2} be satisfied. Then $\Delta^\Rt_P$ 
  and $\Delta^\Av_P$ satisfy the following properties:
  \begin{enumerate}
  \item[(a)] Continuity: $\Delta^\Rt_P$ (resp. $\Delta^\Av_P$) admits a continuous extension 
    to the space of distributions on $\M$ with past (resp. future) compact support with 
    respect to $g$
    \begin{equation}\label{s3e20}
      \begin{split}
        \D'_{+/-}(\M,g) &=\{v\in\D'(\M)\ |\ \forall p\in\M\ ,\,J^{-/+}(p,g)\cap\supp v
        \text{ is compact}\}\\
        &=\{v\in\D'(\M)\ |\ \forall K\subset\M\text{ compact},\,J^{-/+}(K,g)\cap\supp v
        \text{ is compact}\}\ .
      \end{split}
    \end{equation}
  \item[(b)] Propagation of supports: 
    \begin{equation}\label{s3e21}
      \supp (\Delta^{\Rt/\Av}_Pv)\subset J^{+/-}(\supp v,\hat{g})\subset J^{+/-}(\supp v,g)
    \end{equation}
    for all $v\in\D'_{+/-}(\M,g)$.
  \item[(c)] Propagation of singularities: Given any $v\in\D'_{+/-}(\M,g)$, we have that \\
    $(x,\xi)\in\WF(\Delta^{\Rt/\Av}_Pv)$ \emph{only if} either $(x,\xi)\in\WF(v)$ or there 
    is $\Lambda>0$ and a null geodesic segment $\gamma:[0,\Lambda]\To\M$ with respect to 
    $\hat{g}$ such that $\WF(v)\cap E^{\hat{g}}_\gamma\neq\varnothing$, $(x,\xi)\in 
    E^{\hat{g}}_\gamma$. 
  \end{enumerate}
  We have that for all $v\in\D'_{+/-}(\M,g)$, $u=\Delta^{\Rt/\Av}_Pv$ is the unique solution of
  $Pu=v$ on $\M$ belonging to $\D'_{+/-}(\M,g)$.\qed
\end{corollary}

Corollary \ref{s3c2} implies that the causal propagator $\Delta_P$ propagates singularities
in the following fashion: since $\WF(u)\subset\WF(Pu)\cup\{(x,\xi)\in T^*\!\!\M\sm 0
\ |\ g^{-1}(x)(\xi,\xi)=0\}$ for all $u\in\D'(\M)$ (see for instance Proposition 5.1.1, 
page 113 of \cite{duister}), we conclude that, for all $v\in\D'_{+}(\M,g)\cap\D'_-(\M,g)$, 
$(x,\xi)\in\WF(\Delta_Pv)$ only if there is $\Lambda>0$ and a null geodesic segment 
$\gamma:[0,\Lambda]\To\M$ with respect to $\hat{g}$ such that $\WF(v)\cap E^{\hat{g}}_\gamma
\neq\varnothing$, $(x,\xi)\in E^{\hat{g}}_\gamma$, for we have that $P\Delta_Pv=0$. 

\begin{remark}\label{s3r1}
  It is easy to see that $\D'_\pm(\M,g)$ is the topological dual of the space 
  \begin{equation}\label{s3e22}
    \D_\mp(\wedge^d T^*\!\!\M\To\M)=\{\omega\in\Gamma^\infty(\wedge^d T^*\!\!\M\To\M)\ |\ \exists 
    K\subset\M\text{ compact: }\supp\omega\subset J^\mp(K,g)\}\ .
  \end{equation}
\end{remark}

The causal propagator $\Delta_P$ allows a covariant description of the space of 
solutions of $Pu=0$, which is a strengthening of Lemma A.3, page 227 of \cite{dimock}. 
We state and prove the result only for scalar fields, but it actually holds for 
arbitrary vector bundles \cite{bgp}:

\begin{lemma}\label{s3l5} 
  Let $u\in\D'(\M)$. Then $Pu=0$ if and only if $u=\Delta_Pv$ for some 
  $v\in\D'(\M)$ such that $\supp v$ is both past and future compact. If $\supp u\cap\Sigma$
  is compact for some (hence, any) Cauchy hypersurface, we can choose $v$ such 
  that $\supp v$ is compact. In both cases, we can choose $v$ such that $\supp v$ 
  is contained in a neighborhood of any prescribed Cauchy hypersurface 
  $\Sigma$ for $(\M,g)$. Moreover, $\Delta_Pv=0$ if and only if $v=Pw$ for 
  some $w\in\D'(\M)$ such that $\supp w$ is both past and future compact; if $\supp u\cap\Sigma$
  is compact for some (hence, any) Cauchy hypersurface, then $\supp w$ is compact.
\end{lemma}
\begin{proof}
  Let $\Sigma$ be any Cauchy hypersurface for $(\M,g)$. By the results in \cite{bernsan3},
  there is a Cauchy time function $\tau$ in $(\M,g)$ such that $\Sigma=\tau^{-1}(t_0)$ for some
  $t_0\in\RR$. We consider the following separate cases:
  \begin{enumerate}
  \item[(a)] $\supp u\cap\Sigma$ non-compact: $U_1,U_2\subset\M$ open such that
    $U_1=\tau^{-1}((-\infty,t_0+\epsilon))$ and $U_2=\tau^{-1}((t_0-\epsilon,+\infty))$
    for some $\epsilon>0$. Let $\{\chi_1,\chi_2\}$ be a partition of unity subordinated 
    to $\{U_1,U_2\}$. We have that $u=\chi_1u+\chi_2u$, and hence $P(\chi_1u)=-P(\chi_2u)=v$
    is supported inside $\tau^{-1}((t_0-\epsilon,t_0+\epsilon))$, whose closure is 
    past and future compact. Since $\chi_1u$ has past compact support and $\chi_2u$ has 
    future compact support, we have that $\chi_1u=\Delta^\Rt_P(P(\chi_1u)$ and 
    $\chi_2u=\Delta^\Av_P(P(\chi_2u))=-\Delta^\Av_P(P(\chi_1u))$, whence it follows 
    that $u=\Delta_P(P(\chi_1u))=-\Delta_P(P(\chi_2u))=\Delta_Pv$. 
  \item[(b)] $\supp u\cap\Sigma$ compact: $V_1,V_2,V_3\subset\M$ open such that
    $U_1=I^-(K\cap\Sigma,g)$, $U_2=I^+(K\cap\Sigma,g)$ and $V_3=\M\sm(J^+(\supp u\cap
    \Sigma,g)\cup J^-(\supp u\cap\Sigma,g))$, where $K\subset\Sigma$ is a compact subset 
    whose interior in $\Sigma$ contains $\supp u\cap\Sigma$, so that $\ol{U_1\cap U_2}$ 
    is compact. Let $\{\chi'_1,\chi'_2,\chi'_3\}$ be a partition of unity subordinated to 
    $\{V_1,V_2,V_3\}$. We have by Theorem \ref{s3t2} that $u=\chi'_1u+\chi'_2u$ and hence 
    $P(\chi'_1u)=-P(\chi'_2u)=v$ is supported in the compact subset $J^-(K,g)\cap J^+(K,g)$.  
    Since $\chi'_1u$ has past compact support and $\chi_2u$ has future compact support, 
    we have that $\chi'_1u=\Delta^\Rt_P(P(\chi'_1u)$ and $\chi'_2u=\Delta^\Av_P(P(\chi'_2u))
    =-\Delta^\Av_P(P(\chi'_1u))$, whence it follows that $u=\Delta_P(P(\chi'_1u))=-\Delta_P
    (P(\chi'_2u))=\Delta_Pv$.
  \end{enumerate}
  Finally, if $v$ has past and future compact support, and $\Delta_Pv=0$, we clearly 
  have that $\Delta^\Rt_Pv=\Delta^\Av_Pv=w$ has past and future compact support as well, 
  whence $v=Pw$ by Corollary \ref{s3c2}. If in addition $\supp v$ is compact, then $w$ has 
  compact support as well.
\end{proof}

We conclude with the following result:

\begin{proposition}\label{s3p1}
  Let $\lambda\mapsto P_\lambda$, $\lambda\in(a,b)$, $a<b\in\RR$ be a smooth curve of
  normally hyperbolic linear partial differential operators on $\Gamma^\infty(\pi)$
  satisfying \emph{($\mathsf{NH}_g$)}, in the sense that $P_\lambda$ is such an operator 
  for every $\lambda\in(a,b)$ and $\lambda\mapsto (P_\lambda u)(\omega)$ is smooth for all 
  $u\in\D'(\pi\omega\in\Gamma^\infty_c(\E'\otimes\wedge^d T^*\!\!\M\To\M)$. Then 
  $\lambda\mapsto K^{\Sigma,j}_{P_\lambda}$ ($j=0,1$), $\lambda\mapsto\Delta^\Sigma_{P_\lambda}$, 
  $\lambda\mapsto\Delta^\Rt_{P_\lambda}$ and $\lambda\mapsto\Delta^\Av_{P_\lambda}$ are 
  smooth in the sense that $\lambda\mapsto(K^{\Sigma,j}_{P_\lambda}u_j)(\omega)$, $\lambda
  \mapsto(\Delta^\Sigma_{P_\lambda}v)(\omega)$, $\lambda\mapsto(\Delta^\Rt_{P_\lambda}v^+)(\omega)$ 
  and $\lambda\mapsto(\Delta^\Av_{P_\lambda}v^-)(\omega)$ are smooth for all $u_j\in
  \D'(\Sigma)$, $j=0,1$, $v\in\D'_\Sigma(\M)$, $v^\pm\in\D'_\pm(\M,g)$. 
  Moreover, one has the following \emph{resolvent formulae}:
  \begin{align}
    \frac{\dd}{\dd\lambda}K^{\Sigma,j}_{P_\lambda} 
    &=-\Delta^\Sigma_{P_\lambda}\dot{P}_\lambda K^{\Sigma,j}_{P_\lambda}\ ,\label{s3e23}\\
    \frac{\dd}{\dd\lambda}\Delta^\Sigma_{P_\lambda} 
    &=-\Delta^\Sigma_{P_\lambda}\dot{P}_\lambda\Delta^\Sigma_{P_\lambda}\ ,\label{s3e24}\\
    \frac{\dd}{\dd\lambda}\Delta^\Rt_{P_\lambda} 
    &=-\Delta^\Rt_{P_\lambda}\dot{P}_\lambda\Delta^\Rt_{P_\lambda}\ ,\label{s3e25}\\
    \frac{\dd}{\dd\lambda}\Delta^\Av_{P_\lambda} 
    &=-\Delta^\Av_{P_\lambda}\dot{P}_\lambda\Delta^\Av_{P_\lambda}\ ,\label{s3e26}
  \end{align}
  where $\dot{P}_\lambda u=\frac{\dd}{\dd\lambda}(P_\lambda u)$ for all $u\in\D'(\M)$.
  In particular, for all $u_j\in\D'(\Sigma)$, $j=0,1$, $v\in\D'_\Sigma(\M)$, 
  $v^\pm\in\D'_\pm(\M,g)$, we have that $\WF(\frac{\dd}{\dd\lambda}K^{\Sigma,j}_{P_\lambda}
  u_j)\subset\WF(K^{\Sigma,j}_{P_\lambda}u_j)$, $\WF(\frac{\dd}{\dd\lambda}\Delta^\Sigma_{P_\lambda}
  v)\subset\WF(\Delta^\Sigma_{P_\lambda}v)$, $\WF(\frac{\dd}{\dd\lambda}\Delta^\Rt_{P_\lambda}v^+)
  \subset\WF(\Delta^\Rt_{P_\lambda}v^+)$ and $\WF(\frac{\dd}{\dd\lambda}\Delta^\Av_{P_\lambda}v^-)
  \subset\WF(\Delta^\Av_{P_\lambda}v^-)$.
\end{proposition}
\begin{proof}
  We shall restrict our discussion to $u_j,v,v^\pm$ smooth, $j=0,1$. The general case then
  follows from Theorem \ref{s3t2} and Corollary \ref{s3c2}. 
  
  It is straightforward to show that $P_\lambda$ is smooth in $\lambda$ in the above sense 
  if and only if the coefficients of $P_\lambda$ with respect to some (hence, any) choice of 
  connections on $\pi$ and $T\M$ are jointly smooth on $(a,b)\times\M$. Likewise, since 
  $\dot{P}_\lambda$ is a differential operator with smooth coefficients and hence preserves 
  wave front sets, the above statements on the latter also follow from Theorem \ref{s3t2} 
  and Corollary \ref{s3c2}.
  
  First we prove \eqref{s3e23}. Notice that for every $h\in\RR$ with $0<|h|<\min\{\lambda-a,
  b-\lambda\}$ we have that
  \[
  P_{\lambda+h}\left(\frac{1}{h}(K^{\Sigma,j}_{P_{\lambda+h}}u_j-K^{\Sigma,j}_{P_\lambda}u_j)\right)
  =-\frac{1}{h}(P_{\lambda+h}-P_\lambda)K^{\Sigma,j}_{P_\lambda}u_j
  \]
  and 
  \[
  \rho_0^\Sigma\left(\frac{1}{h}(K^{\Sigma,j}_{P_{\lambda+h}}u_j-K^{\Sigma,j}_{P_\lambda}u_j)\right)=
  \rho_1^\Sigma\left(\frac{1}{h}(K^{\Sigma,j}_{P_{\lambda+h}}u_j-K^{\Sigma,j}_{P_\lambda}u_j)\right)=0\ .
  \]
  This implies that $\lim_{h\To 0}\frac{1}{h}(K^{\Sigma,j}_{P_{\lambda+h}}u_j-K^{\Sigma,j}_{P_\lambda}u_j)
  \doteq u$ exists in the sense of distributions and solves the initial-value problem
  \[
  \begin{cases}
    Pu &=\dot{P}_\lambda K^{\Sigma,j}_{P_\lambda}u_j\ ,\\
    \rho_0^\Sigma(u) &=0\ ,\\
    \rho_1^\Sigma(u) &=0\ ,
  \end{cases}
  \]
  since
  \[
  \frac{1}{h}\Spr{(K^{\Sigma,j}_{P_{\lambda+h}}u_j-K^{\Sigma,j}_{P_\lambda}u_j),P_{\lambda+h}'\omega}
  =-\frac{1}{h}\Spr{(P_{\lambda+h}-P_\lambda)K^{\Sigma,j}_{P_\lambda}u_j,\omega}
  \]
  for all $\omega\in\Gamma^\infty_c(\E'\otimes\wedge^d T^*\!\!\M\To\M)$ and all $h\in\RR$ 
  with $0<|h|<\min\{\lambda-a,b-\lambda\}$, where $P_{\lambda+h}'$ is the formal adjoint 
  of $P_{\lambda+h}$. Hence, $u$ must be smooth and is given by the right-hand side of 
  \eqref{s3e23} applied to $u_j$. Finally, by Corollary 1.9, pp. 14 of \cite{km}, $u$ 
  must coincide with the left-hand side of \eqref{s3e23} applied to $u_j$.
  
  The reasoning for proving \eqref{s3e24} is similar, since for every $h\in\RR$ with 
  $0<|h|<\min\{\lambda-a,b-\lambda\}$ we have that
  \[
  \begin{split}
    P_{\lambda+h}\left(\frac{1}{h}(\Delta^\Sigma_{P_{\lambda+h}}v-\Delta^\Sigma_{P_\lambda}v)\right)
    &=\frac{1}{h}v-\frac{1}{h}(P_{\lambda+h}-P_\lambda)\Delta^\Sigma_{P_\lambda}v-\frac{1}{h}v \\
    &=-\frac{1}{h}(P_{\lambda+h}-P_\lambda)\Delta^\Sigma_{P_\lambda}v
  \end{split}
  \]
  and 
  \[
  \rho_0^\Sigma\left(\frac{1}{h}(\Delta^\Sigma_{P_{\lambda+h}}v-\Delta^\Sigma_{P_\lambda}v)\right)=
  \rho_1^\Sigma\left(\frac{1}{h}(\Delta^\Sigma_{P_{\lambda+h}}v-\Delta^\Sigma_{P_\lambda}v)\right)=0\ .
  \]
  The same goes for \eqref{s3e25} and \eqref{s3e26}, once we choose a Cauchy hypersurface
  $\Sigma$ contained in $I^-(\supp v^+)$\\$\sm\supp v^+$ (resp. $I^+(\supp v^-)\sm\supp v^-$), 
  which can always be done since $\supp v^+$ (resp. $\supp v^-$) is past (resp. future) 
  compact -- we omit the remaining details.
\end{proof}

In the same way one derives the $k$-th order resolvent formula \eqref{a1e12} from the
first-order case \eqref{a1e11}, the same can be done from \eqref{s3e23}--\eqref{s3e26}.

Let $g'$ be a Lorentzian metric on $\M$, \emph{a priori} unrelated to either the 
space-time metric $g$ or the metric $\hat{g}$ associated to the principal symbol of a 
normally hyperbolic linear partial differential operator $P$. Recall now the definition 
of the Hodge star operator $*_{\!g'}$ acting on $d$-forms on $\M$: Given $\omega\in
\Gamma^\infty(\wedge^d T^*\!\!\M\To\M)$, we define $*_{\!g'}\omega\in\C^\infty(\M)$ as the 
unique smooth function on $\M$ such that
\begin{equation}\label{s3e27}
  \omega=(*_{\!g'}\omega)\ud\mu_{g'}\ .
\end{equation}
Conversely, if $\vec{\varphi}\in\C^\infty(\M)$, we have that
\begin{equation}\label{s3e28}
  \vec{\varphi}=*_{\!g'}(\vec{\varphi}\ud\mu_{g'})\ .
\end{equation}
The following result follows immediately from Theorem \ref{s3t2}.

\begin{lemma}\label{s3l6}
  Let $g'$ be a Lorentzian metric on $\M$ and $P:\C^\infty(\M)\To\C^\infty(\M)$ be 
  a linear partial differential operator. Then $P$ is formally self-adjoint with
  respect to the $L^2$ scalar product associated to $\ud\mu_{g'}$ if and only if 
  the map $\C^\infty(\M)\ni\vec{\varphi}\mapsto(P\vec{\varphi})\ud\mu_{g'}\in\Gamma^\infty
  (\wedge^d T^*\!\!\M\To\M)$ has a symmetric distribution kernel. If either fact
  holds (hence both), the distribution kernel of $\Delta^\Av_P\circ *_{\!g'}$ is 
  the adjoint of the distribution kernel of $\Delta^\Rt_P\circ *_{\!g'}$.\qed
\end{lemma}

The situation we have in mind is, of course, when $P\vec{\varphi}=*_{\!g'}E'(\Li)
[\varphi_0]\vec{\varphi}$, where $\Li$ is a \emph{real-valued}, microlocal generalized 
Lagrangian of first order on $\UU\subset\C^\infty(\M)$ open in the compact-open topology. 
Generally, given a microlocal generalized Lagrangian $\Li$ of order $r$ on $\UU$, 
$E(\Li)$ is a \emph{quasi-linear} partial differential operator, that is, 
$E(\Li)[\varphi]$ is linear in the highest order derivatives of $\varphi$. 
Therefore, we say that the partial differential operator of second order 
$E(\Li)$ is \emph{normally hyperbolic} on $\UU$ if, for all $\varphi_0\in\UU$, 
$P=*_{\!g'}E'(\Li)[\varphi_0]$ is normally hyperbolic for some (hence any) 
Lorentzian metric $g'$ on $\M$. In this case, we denote the metric associated 
to the principal symbol of $P$ defined as above by $\hat{g}_\Li=\hat{g}_\Li[\varphi_0]$, 
and write
\begin{align}
  K^{\Sigma,j}_\Li[\varphi_0] &\doteq K^{\Sigma,j}_P\quad(j=0,1)\ ,\label{s3e29}\\
  \Delta^\Sigma_\Li[\varphi_0] &\doteq \Delta^\Sigma_P\circ *_{\!g'}\ ,\label{s3e30}\\
  \Delta^\Rt_\Li[\varphi_0] &\doteq \Delta^\Rt_P\circ *_{\!g'}\ ,\label{s3e31}\\
  \Delta^\Av_\Li[\varphi_0] &\doteq \Delta^\Av_P\circ *_{\!g'}\ ,\label{s3e32}\\
  \Delta_\Li[\varphi_0] &\doteq \Delta^\Rt_\Li[\varphi_0]-\Delta^\Av_\Li[\varphi_0]\ .\label{s3e33}
\end{align}
We remark that different choices of $g'$ affect $\hat{g}_\Li$ only by a
$\varphi_0$-independent conformal factor -- in particular, the causal structure
of $\hat{g}_\Li$ is independent of $g'$. For future convenience, we summarize the 
estimates on the wave front sets of the distribution kernels of the linear 
operators \eqref{s3e29}--\eqref{s3e32} derived from Theorem \ref{s3t2} and 
Corollary \ref{s3c2}. To wit, if $\gamma:[0,1]\To\M$ is a null geodesic segment 
with respect to $\hat{g}_\Li[\varphi]$ and 
\[
E^{\hat{g}_\Li[\varphi]}_\gamma=\{(\gamma(0),\hat{g}_\Li[\varphi]^\flat(\dot\gamma(0))),\,(\gamma(1),
\hat{g}_\Li[\varphi]^\flat(\dot\gamma(1)))\}
\]
is the set of endpoints of the corresponding bicharacteristic strip, then
\begin{align}
  \WF(K^{\Sigma,j}_\Li[\varphi]) &\subset\{(x_0,y;\xi_0,\eta)\in T^*(\Sigma\times\M)\ |\ \exists
                                   \gamma:[0,1]\To\M\text{ null geodesic}\label{s3e34}\\ 
                                 &\phantom{\subset\{}\text{such that }E^{\hat{g}_\Li[\varphi]}_\gamma
                                   =\{(x_0,\xi_0),(y,-\eta)\}\},\nonumber\\
  \WF(\Delta^\Sigma_\Li[\varphi]) &\subset\{(x,y;\xi,\eta)\in T^*(\M\times\M)\ |\ x=y,\,\xi=\eta
                                    \text{ or }\exists\gamma:[0,1]\To\M\text{ null geodesic}\label{s3e35}\\ 
                                 &\phantom{\subset\{}\text{such that either }x\leq_gy\leq_g\Sigma
                                   \text{ or }\Sigma\leq_gy\leq_g x\text{ and }E^{\hat{g}_\Li[\varphi]}_\gamma
                                   =\{(x,\xi),(y,-\eta)\}\}\ ,\nonumber\\
  \WF(\Delta^\Rt_\Li[\varphi]) &\subset\{(x,y;\xi,\eta)\in T^*(\M\times\M)\ |\ x=y,\,\xi=\eta
                                 \text{ or }\exists\gamma:[0,1]\To\M\label{s3e36}\\ 
                                 &\phantom{\subset\{}\text{ null geodesic such that }x\geq_g y
                                   \text{ and }E^{\hat{g}_\Li[\varphi]}_\gamma=\{(x,\xi),(y,-\eta)\}\}\ ,\nonumber\\
  \WF(\Delta^\Av_\Li[\varphi]) &\subset\{(x,y;\xi,\eta)\in T^*(\M\times\M)\ |\ x=y,\,\xi=\eta
                                 \text{ or }\exists\gamma:[0,1]\To\M\label{s3e37}\\ 
                                 &\phantom{\subset\{}\text{ null geodesic such that }x\leq_g y
                                   \text{ and }E^{\hat{g}_\Li[\varphi]}_\gamma=\{(x,\xi),(y,-\eta)\}\}\ ,\nonumber\\
  \WF(\Delta_\Li[\varphi]) &\subset\{(x,y;\xi,\eta)\in T^*(\M\times\M)\ |\ \exists\gamma:[0,1]\To\M
                             \text{ null geodesic}\label{s3e38}\\ 
                                 & \phantom{=\{}\text{such that }
                                   E^{\hat{g}_\Li[\varphi]}_\gamma=\{(x,\xi),(y,-\eta)\}\}\ ,\nonumber
\end{align}
where we identify each of the propagators above with the corresponding distribution kernels.
We stress once more that, due to the identity $E'(\Li)[\varphi]\Delta_\Li[\varphi]=0$, 
$\WF(\Delta_\Li[\varphi])$ has only pairs of \emph{null} covectors, even over the diagonal 
$\Delta_2(\M)$ of $\M^2$. This is no longer the case for $\WF(\Delta^\Sigma_\Li[\varphi])$, 
$\WF(\Delta^\Rt_\Li[\varphi])$ or $\WF(\Delta^\Av_\Li[\varphi])$, which may have conormal 
covectors over $\Delta_2(\M)$ which consist of pairs of covectors of arbitrary causal character.

\begin{remark}\label{s3r2} 
  Let us display a sufficiently nontrivial example of a microlocal generalized Lagrangian 
  with normally hyperbolic Euler-Lagrange operator. For instance,
  \begin{equation}\label{s3e39}
    \Li(f)(\varphi)=-\frac{1}{2}\int_\M f\left[g^{-1}(\ud\varphi,\ud\varphi)
      +\frac{\epsilon}{2}(1+\varphi^2)g^{-1}(\ud\varphi,\ud\varphi)^2\right]\ud\mu_g\ ,
    \quad\epsilon\geq 0\ .
  \end{equation}
  The Euler-Lagrange operator of $\Li$ is given by
  \begin{equation}\label{s3e40}
    \begin{split}
      E(\Li)[\varphi] &=\left[(1+\epsilon(1+\varphi^2)g^{-1}(\ud\varphi,\ud\varphi))
        \square_g\varphi+\epsilon(2\nabla^2\varphi(g^\sharp(\ud\varphi),g^\sharp
        (\ud\varphi))-\frac{1}{2}g^{-1}(\ud\varphi,\ud\varphi)\varphi)\right]\ud\mu_g\ ,\\ 
      \square_g\varphi &=g^{-1}(\nabla^2\varphi)\ ,
    \end{split}
  \end{equation}
  whose linearization around $\varphi_0$ is given by
  \begin{equation}\label{s3e41}
    \begin{split}
      E'(\Li)[\varphi_0]\vec{\varphi} &=\left[(1+\epsilon(1+\varphi_0^2)g^{-1}
        (\ud\varphi_0,\ud\varphi_0))\square_g\vec{\varphi}+2\epsilon\nabla^2
        \vec{\varphi}(g^\sharp(\ud\varphi_0),g^\sharp(\ud\varphi_0))\right.\\
      &+2\epsilon\left[\left((1+\varphi_0^2)\square_g\varphi_0-\frac{1}{2}
          \varphi_0\right)g^{-1}(\ud\varphi_0,\ud\vec{\varphi})+2\nabla^2
        \varphi_0(g^\sharp(\ud\varphi_0),g^\sharp(\ud\vec{\varphi}))\right.\\
      &\left.+\epsilon g^{-1}(\ud\varphi_0,\ud\varphi_0)\left(2\square_g\varphi_0
          -\frac{1}{2}\right)\vec{\varphi}\right]\ud\mu_g\\
      &=\left[\hat{g}^{-1}_\Li[\varphi_0](\nabla^2\vec{\varphi})+\nabla_A
        \vec{\varphi}+B\vec{\varphi}\right]\ud\mu_g\ ,
    \end{split}
  \end{equation}
  where $A(p)=A(g(p),\varphi_0(p),\nabla\varphi_0(p),\nabla^2\varphi_0(p))$ and
  $B=B(g(p),\varphi_0(p),\nabla\varphi_0(p),\nabla^2\varphi_0(p))$ for all $p\in\M$.
  The principal symbol of $P=*_g E'(\Li)[\varphi_0]$ reads
  \begin{equation}\label{s3e42}
    \begin{split}
      \hat{g}^{-1}_\Li[\varphi_0](g^\flat(X_1),g^\flat(X_2)) &=
      (1+\epsilon(1+\varphi_0^2)g^{-1}(\ud\varphi_0,\ud\varphi_0))g(X_1,X_2)\\
      &+2\epsilon(\nabla_{X_1}\varphi_0)(\nabla_{X_2}\varphi_0)\ ,
    \end{split}
  \end{equation}
  whence we conclude that 
  \begin{equation}\label{s3e43}
    \begin{split}
      \hat{g}^{-1}_\Li[\varphi_0](g^\flat(X),g^\flat(X))>0 &\Iff 
      (1+\epsilon(1+\varphi_0^2)g^{-1}(\ud\varphi_0,\ud\varphi_0))
      g(X,X)>-2\epsilon(\nabla_X\varphi_0)^2\ ,\\
      \hat{g}^{-1}_\Li[\varphi_0](g^\flat(X),g^\flat(X))=0 &\Iff 
      (1+\epsilon(1+\varphi_0^2)g^{-1}(\ud\varphi_0,\ud\varphi_0))
      g(X,X)=-2\epsilon(\nabla_X\varphi_0)^2\ ,\\
      \hat{g}^{-1}_\Li[\varphi_0](g^\flat(X),g^\flat(X))<0 &\Iff 
      (1+\epsilon(1+\varphi_0^2)g^{-1}(\ud\varphi_0,\ud\varphi_0))
      g(X,X)<-2\epsilon(\nabla_X\varphi_0)^2\ .
    \end{split}
  \end{equation}
  We consider the following three possibilities:
  \begin{align}
    g^{-1}(\ud\varphi_0,\ud\varphi_0) &>-\frac{1}{2\epsilon(1+\varphi_0^2)}\ ,\label{s3e44}\\
    g^{-1}(\ud\varphi_0,\ud\varphi_0) &=-\frac{1}{2\epsilon(1+\varphi_0^2)}\ ,\label{s3e45}\\
    g^{-1}(\ud\varphi_0,\ud\varphi_0) &<-\frac{1}{2\epsilon(1+\varphi_0^2)}\ .\label{s3e46}
  \end{align}
  Inequalities \eqref{s3e44} and \eqref{s3e46} define open subsets of $\C^\infty(\M)$ 
  in the Whitney topology. In case \eqref{s3e44} holds, we have that $\hat{g}^{-1}_\Li
  [\varphi_0](g^\flat(X),g^\flat(X))<0$ implies $g(X,X)<0$ and $\hat{g}^{-1}_\Li
  [\varphi_0](g^\flat(X),g^\flat(X))=0$ implies $g(X,X)\leq 0$, whereas $\hat{g}^{-1}_\Li
  [\varphi_0](g^\flat(X),g^\flat(X))>0$ does not constrain the causal character of $X$ 
  with respect to $g$. In case \eqref{s3e45} holds, we have that $\hat{g}^{-1}_\Li
  [\varphi_0](g^\flat(X),g^\flat(Y))=0$ for all tangent vectors $Y$ if $X$ satisfies 
  $\nabla_X\varphi_0=0$ (hence $\hat{g}^{-1}_\Li[\varphi_0]$ becomes degenerate); 
  moreover, $\hat{g}^{-1}_\Li[\varphi_0]$ cannot have any timelike covectors. In case 
  \eqref{s3e46} holds, we have that $\hat{g}^{-1}_\Li[\varphi_0](g^\flat(X),g^\flat(X))<0$ 
  implies $g(X,X)>0$ and $\hat{g}^{-1}_\Li[\varphi_0](g^\flat(X),g^\flat(X))=0$ implies 
  $g(X,X)\geq 0$, whereas $\hat{g}^{-1}_\Li[\varphi_0](g^\flat(X),g^\flat(X))$ $>0$ 
  does not constrain the causal character of $X$ with respect to $g$. To summarize,
  \begin{align}
    g^{-1}(\ud\varphi_0,\ud\varphi_0)>-\frac{1}{2\epsilon(1+\varphi_0^2)} 
    & \Then\hat{g}_\Li[\varphi_0]\lesssim g\ ,\label{s3e47}\\
    g^{-1}(\ud\varphi_0,\ud\varphi_0)=-\frac{1}{2\epsilon(1+\varphi_0^2)} 
    & \Then\hat{g}^{-1}_\Li[\varphi_0]\text{ degenerate}\ ,\label{s3e48}\\
    g^{-1}(\ud\varphi_0,\ud\varphi_0)<-\frac{1}{2\epsilon(1+\varphi_0^2)} 
    & \Then-\hat{g}_\Li[\varphi_0]\lesssim g\ .\label{s3e49}
  \end{align}
  In other words, crossing the boundary $g^{-1}(\ud\varphi_0,\ud\varphi_0)=
  -\frac{1}{2\epsilon(1+\varphi_0^2)}$ causes $\hat{g}_\Li[\varphi_0]$'s signature
  to change sign, partitioning $\C^\infty(\M)$ into two Whitney-open, disjoint 
  ``domains of hyperbolicity'' separated by the boundary $g^{-1}(\ud\varphi_0,\ud\varphi_0)
  =-\frac{1}{2\epsilon(1+\varphi_0^2)}$. The presence of this boundary is linked to 
  the lifespan of solutions of $E(\Li)[\varphi]=0$; indeed, the ``sharp continuation 
  principle'' of Majda (Theorem 2.2, pp. 31--32 in \cite{majda}) implies that, at least 
  when $(\M,g)$ is the Minkowski space-time, if a solution $\varphi$ to $E(\Li)[\varphi]=0$ 
  with given Cauchy data at $\Sigma=\tau^{-1}(0)$ blows up in $\C^\infty(\M)$ as $\tau(p)\To t^*>0$ 
  but the second-order jet prolongation of $\varphi$ is bounded in $K\cap\tau^{-1}([0,t^*))$ 
  for any compact subset $K\subset\M$, then we must have that $g^{-1}(\ud\varphi(p),\ud\varphi(p))
  +\frac{1}{2\epsilon(1+\varphi^2(p))}\gotoas{0}{\tau(p)}{t^*}$, where $\tau$ is a Cauchy 
  time function on $(\M,g)$. We stress that it is not hard to provide examples of $\varphi_0$ 
  which fall into either \eqref{s3e47} or \eqref{s3e49} -- for \eqref{s3e47} to hold, it 
  suffices to choose $\varphi_0$ with everywhere spacelike gradient; as for \eqref{s3e49}, 
  any Cauchy time function $\varphi_0=\tau$ on $(\M,g)$ satisfying $g^{-1}(\ud\tau,\ud\tau)
  <-(2\epsilon)^{-1}$ does the trick, and any globally hyperbolic space-time admits such Cauchy 
  time functions \cite{mulsan}. On the other hand, this is a typical ``large data'' phenomenon, 
  specially if $\epsilon$ is small. Since the nonlinear terms of $E(\Li)[\varphi]$ vanish to 
  third order at $\varphi=0$, one can show, at least when $(\M,g)$ is the Minkowski space-time, 
  that $E(\Li)[\varphi]=0$ has unique, global smooth solutions for sufficiently small Cauchy 
  data \cite{horm3,sogge}.
\end{remark}

Motivated by formula \eqref{s2e17} in Remark \ref{s2r4}, we write for each $\vec{\psi}_j\in
\C^\infty(\Sigma)$, $j=0,1$, $\omega\in\Gamma^\infty(\wedge^d T^*\!\!\M\To\M)$, $\omega^\pm\in
\Gamma^\infty_{\pm}(\wedge^d T^*\!\!\M\To\M,g)$
\begin{align}
  D^kK^{\Sigma,j}_\Li[\varphi_0](\vec{\varphi}_1,\ldots,\vec{\varphi}_k)\vec{\psi}_j 
  &\doteq\frac{\dd^k}{\dd\lambda_1\cdots\dd\lambda_k}\Restr{\lambda_1=\cdots=\lambda_k=0}
    K^{\Sigma,j}_\Li\left[\varphi_0+\sum^k_{l=1}\lambda_l\vec{\varphi}_l\right]\vec{\psi}_j\ ,\label{s3e50}\\
  D^k\Delta^\Sigma_\Li[\varphi_0](\vec{\varphi}_1,\ldots,\vec{\varphi}_k)\omega 
  &\doteq\frac{\dd^k}{\dd\lambda_1\cdots\dd\lambda_k}\Restr{\lambda_1=\cdots=\lambda_k=0}
    \Delta^\Sigma_\Li\left[\varphi_0+\sum^k_{l=1}\lambda_l\vec{\varphi}_l\right]\omega\ ,\label{s3e51}\\
  D^k\Delta^\Rt_\Li[\varphi_0](\vec{\varphi}_1,\ldots,\vec{\varphi}_k)\omega^+ 
  &\doteq\frac{\dd^k}{\dd\lambda_1\cdots\dd\lambda_k}\Restr{\lambda_1=\cdots=\lambda_k=0}
    \Delta^\Rt_\Li\left[\varphi_0+\sum^k_{l=1}\lambda_l\vec{\varphi}_l\right]\omega^+\ ,\label{s3e52}\\
  D^k\Delta^\Av_\Li[\varphi_0](\vec{\varphi}_1,\ldots,\vec{\varphi}_k)\omega^-
  &\doteq\frac{\dd^k}{\dd\lambda_1\cdots\dd\lambda_k}\Restr{\lambda_1=\cdots=\lambda_k=0}
    \Delta^\Av_\Li\left[\varphi_0+\sum^k_{l=1}\lambda_l\vec{\varphi}_l\right]\omega^-\ .\label{s3e53}
\end{align}
Combining Proposition \ref{s3t1} with the chain rule \eqref{a1e3} yields for each 
$\vec{\psi}_j\in\C^\infty(\Sigma)$, $j=0,1$, $\omega\in\Gamma^\infty(\wedge^d 
T^*\!\!\M\To\M)$, $\omega^\pm\in\Gamma^\infty_{\pm}(\wedge^d T^*\!\!\M\To\M,g)$ that
\begin{align}
  DK^{\Sigma,j}_\Li[\varphi_0](\vec{\varphi})\vec{\psi}_j 
  &=-\Delta^\Sigma_\Li[\varphi_0]D^2E(\Li)[\varphi_0](\vec{\varphi})K^{\Sigma,j}_\Li\vec{\psi}_j\ ,\label{s3e54}\\
  D\Delta^\Sigma_\Li[\varphi_0](\vec{\varphi})\omega 
  &=-\Delta^\Sigma_\Li[\varphi_0]D^2E(\Li)[\varphi_0](\vec{\varphi})\Delta^\Sigma_\Li\omega\ ,\label{s3e55}\\
  D\Delta^\Rt_\Li[\varphi_0](\vec{\varphi})\omega^+ 
  &=-\Delta^\Rt_\Li[\varphi_0]D^2E(\Li)[\varphi_0](\vec{\varphi})\Delta^\Rt_\Li\omega^+\ ,\label{s3e56}\\
  D\Delta^\Av_\Li[\varphi_0](\vec{\varphi})\omega^- 
  &=-\Delta^\Av_\Li[\varphi_0]D^2E(\Li)[\varphi_0](\vec{\varphi})\Delta^\Av_\Li\omega^-\ ,\label{s3e57}
\end{align}
whence it follows from the same reasoning leading from the first-order resolvent
formula \eqref{a1e11} to the $k$-th order resolvent formula \eqref{a1e12} that
\begin{align}
  D^k&K^{\Sigma,j}_\Li[\varphi_0](\vec{\varphi}_1,\ldots,\vec{\varphi}_k)\vec{\psi}_j\label{s3e58}\\
     &=\sum^k_{l=1}(-1)^l\sum_{\{I_1,\ldots,I_l\}\in P_k}\sum_{\sigma\in S_l}\left(\prod^l_{j=1}
       \Delta^\Sigma_\Li[\varphi_0]D^{|I_{\sigma(j)}|+1}E(\Li)[\varphi_0](\otimes_{i\in I_{\sigma(j)}}
       \vec{\varphi}_i)\right)K^{\Sigma,j}_\Li[\varphi_0]\vec{\psi}_j\ ,\nonumber\\
  D^k&\Delta^\Sigma_\Li[\varphi_0](\vec{\varphi}_1,\ldots,\vec{\varphi}_k)\omega\label{s3e59}\\
     &=\sum^k_{l=1}(-1)^l\sum_{\{I_1,\ldots,I_l\}\in P_k}\sum_{\sigma\in S_l}\left(\prod^l_{j=1}
       \Delta^\Sigma_\Li[\varphi_0]D^{|I_{\sigma(j)}|+1}E(\Li)[\varphi_0](\otimes_{i\in I_{\sigma(j)}}
       \vec{\varphi}_i)\right)\Delta^\Sigma_\Li[\varphi_0]\omega\ ,\nonumber\\
  D^k&\Delta^\Rt_\Li[\varphi_0](\vec{\varphi}_1,\ldots,\vec{\varphi}_k)\omega^+\label{s3e60}\\
     &=\sum^k_{l=1}(-1)^l\sum_{\{I_1,\ldots,I_l\}\in P_k}\sum_{\sigma\in S_l}\left(\prod^l_{j=1}
       \Delta^\Rt_\Li[\varphi_0]D^{|I_{\sigma(j)}|+1}E(\Li)[\varphi_0](\otimes_{i\in I_{\sigma(j)}}
       \vec{\varphi}_i)\right)\Delta^\Rt_\Li[\varphi_0]\omega^+\ ,\nonumber\\
  D^k&\Delta^\Av_\Li[\varphi_0](\vec{\varphi}_1,\ldots,\vec{\varphi}_k)\omega^-\label{s3e61}\\
     &=\sum^k_{l=1}(-1)^l\sum_{\{I_1,\ldots,I_l\}\in P_k}\sum_{\sigma\in S_l}\left(\prod^l_{j=1}
       \Delta^\Av_\Li[\varphi_0]D^{|I_{\sigma(j)}|+1}E(\Li)[\varphi_0](\otimes_{i\in I_{\sigma(j)}}
       \vec{\varphi}_i)\right)\Delta^\Av_\Li[\varphi_0]\omega^-\ .\nonumber
\end{align}
and therefore
\begin{equation}\label{s3e62}
  \begin{split}
    D^k K^{\Sigma,j}_\Li &:\UU\times(\C^\infty(\M))^k\times\D(\Sigma)\To\D'_\Sigma(\M)\quad  j=0,1\ ,\\ 
    D^k\Delta^\Sigma_\Li &:\UU\times(\C^\infty(\M))^k\times\D'_\Sigma(\wedge^d T^*\!\!\M\To\M)\To\D'_\Sigma(\M)\ ,\\ 
    D^k\Delta^\Rt_\Li &:\UU\times(\C^\infty(\M))^k\times\D'_{+}(\wedge^d T^*\!\!\M\To\M,g)\To\D'_\Sigma(\M)\text{ and}\\
    D^k\Delta^\Av_\Li &:\UU\times(\C^\infty(\M))^k\times\D'_{-}(\wedge^d T^*\!\!\M\To\M,g)\To\D'_\Sigma(\M) 
  \end{split}
\end{equation}
exist and are jointly continuous for all $k\geq 1$, where
\[
\begin{split}
  \D'_\Sigma(\wedge^d T^*\!\!\M\To\M)&\doteq\{u\in\D'(\wedge^d T^*\!\!\M\To\M)\ |\ 
  \WF(u)\cap N^*\Sigma=\varnothing\}\ ,\\
  \D'_{\pm}(\wedge^d T^*\!\!\M\To\M,g)&\doteq\{u\in\D'(\wedge^d T^*\!\!\M\To\M)\ |\ 
  \exists K\subset\M\text{ compact}\\ &\phantom{\doteq\{}\text{such that }J^\mp(K)
  \cap\supp u\text{ is compact}\}\ .
\end{split}
\]

\begin{definition}\label{s3d5}
  Let $\UU\subset\C^\infty(\M)$ be open in the compact-open topology, and $F,G\in\Fun_{\mu\loc}
  (\M,\UU)$. The \emph{retarded} and \emph{advanced products} $\Ret_\Li(F,G)$, $\Adv_\Li(F,G)$ 
  with respect to $\Li$ are functionals respectively given by
  \begin{equation}\label{s3e63}
    \Ret_\Li(F,G)(\varphi) \doteq\Spr{F^{(1)}[\varphi],\Delta^\Rt_\Li[\varphi] G^{(1)}[\varphi]}
  \end{equation} 
  and 
  \begin{equation}\label{s3e64}
    \begin{split}
      \Adv_\Li(F,G)(\varphi) &\doteq\Spr{F^{(1)}[\varphi],\Delta^\Av_\Li[\varphi] G^{(1)}[\varphi]}\\
      &=\Ret_\Li(G,F)(\varphi) \ .
    \end{split}
  \end{equation}
  Their difference 
  \begin{equation}\label{s3e65}
    \{F,G\}_\Li\doteq\Ret_\Li(F,G)-\Adv_\Li(F,G)=\Ret_\Li(F,G)-\Ret_\Li(G,F)
  \end{equation} 
  is called the \emph{Peierls bracket} of $F$ with $G$ with respect to $\Li$.
\end{definition}

By Lemma \ref{s3l6}, the Peierls bracket is antisymmetric in its entries, becoming
an obvious candidate for a Poisson bracket. Let us prove some basic properties of 
$\Ret_\Li(\cdot,\cdot)$, $\Adv_\Li(\cdot,\cdot)$ and $\{\cdot,\cdot\}_\Li$. For later 
convenience, given $\varnothing\neq K,L\subset\M$ we define
\begin{equation}\label{s3supps}
  \OO^\Rt_{K,L}=J^+(K,g)\cap J^-(L,g)\ ,\,\OO^\Av_{K,L}=\OO^\Rt_{L,K}\ ,\,\OO_{K,L}=\OO^\Rt_{K,L}
  \cup\OO^\Av_{K,L}\ .
\end{equation}
By global hyperbolicity of $(\M,g)$, we have that $\OO^\Rt_{K,L}$, $\OO^\Av_{K,L}$ and $\OO_{K,L}$
are compact if $K,L$ also are.

\begin{proposition}\label{s3p2}
  Let $\UU,F,G$ as in Definition \ref{s3d5}. Then $\Ret_\Li(F,G)$, $\Adv_\Li(F,G)$ and 
  $\{F,G\}_\Li$ are smooth and satisfy the support properties
  \begin{align}
    \supp\Ret_\Li(F,G) &\subset\OO^\Rt_{\supp F,\supp G}\ ,\label{s3e66}\\
    \supp\Adv_\Li(F,G) &\subset\OO^\Av_{\supp F,\supp G}\ ,\label{s3e67}\\
    \supp\{F,G\}_\Li &\subset\OO_{\supp F,\supp G}\ .\label{s3e68}
  \end{align}
\end{proposition}
\begin{proof}
  Notice that $\Delta^\Rt_\Li[\varphi]$ and $\Delta^\Av_\Li[\varphi]$ depend on the 
  background field configuration $\varphi$ only so far as the coefficients of 
  $E'(\Li)[\varphi]$ depend on $\varphi$. Therefore, by part (b) of Theorem \ref{s3t2},
  for all $\omega\in\Gamma^\infty_c(\wedge^d T^*\!\!\M\To\M)$ any modification of $\varphi$ 
  outside $J^+(\supp f,g)$ (resp. $J^-(\supp f,g)$) leaves $\Delta^\Rt_\Li[\varphi]f$ 
  (resp. $\Delta^\Av_\Li[\varphi]f$) unaltered. Since $\Delta^\Rt_\Li[\varphi]$ is
  the formal adjoint of $\Delta^\Av_\Li[\varphi]$, the above reasoning together
  with part (b) of Theorem \ref{s3t2} imply that $\Delta^\Rt_\Li[\varphi]$ and 
  $\Delta^\Av_\Li[\varphi]$ have the desired support properties. Now we are only 
  left with proving that $\Ret_\Li(F,G)$ and $\Adv_\Li(F,G)$ are smooth functionals, 
  since this implies the corresponding result for $\{F,G\}_\Li$. This, however, follows 
  from formulae \eqref{s3e60} and \eqref{s3e61} together with the trilinear Leibniz 
  rule \eqref{a1e8} (i.e. with $l=3$ therein), which give us that
  \begin{equation}\label{s3e69}
    \begin{split}
      D^k\Ret_\Li(F,G)[\varphi](\vec{\varphi}_1,\ldots,\vec{\varphi}_k) &=\sum_{\{J_1,J_2,J_3\}\subset P_k}
      F^{(|J_1|+1)}[\varphi]((\otimes_{j_1\in J_1}\vec{\varphi}_{j_1})\\ &\otimes D^{|J_2|}\Delta^\Rt_\Li
      [\varphi]((\otimes_{j_2\in J_2}\vec{\varphi}_{j_2})\otimes G^{(|J_3|+1)}[\varphi]
      (\otimes_{j_3\in J_3}\vec{\varphi}_{j_3})))\ ,\\
      D^k\Adv_\Li(F,G)[\varphi](\vec{\varphi}_1,\ldots,\vec{\varphi}_k) &=\sum_{\{J_1,J_2,J_3\}\subset P_k}
      F^{(|J_1|+1)}[\varphi]((\otimes_{j_1\in J_1}\vec{\varphi}_{j_1})\\ &\otimes D^{|J_2|}\Delta^\Av_\Li
      [\varphi]((\otimes_{j_2\in J_2}\vec{\varphi}_{j_2})\otimes G^{(|J_3|+1)}[\varphi]
      (\otimes_{j_3\in J_3}\vec{\varphi}_{j_3})))\ .
    \end{split}
  \end{equation}
  where $P_k$ is set of all partitions of the set $\{1,\ldots,k\}$. We notice that 
  due to \eqref{s3e60} and \eqref{s3e61}, each term in the right-hand side of \eqref{s3e69} before
  smearing with $\vec{\varphi}_1,\ldots,\vec{\varphi}_k$ can be seen as a string of 
  compositions of:
  \begin{enumerate}
  \item[(i)] $l+1$ propagators of the form $\Delta^\Rt_\Li[\varphi]$ (for the retarded
    product) or $\Delta^\Av_\Li[\varphi]$ (for the advanced product); and 
  \item[(ii)] $l+2$ $k_i$-linear \emph{differential} operators, $i=0,\ldots,l+1$ whose 
    distribution kernels are either of the form $F^{(k_0+1)}[\varphi]$, $D^{k_i+2}\Li(1)
    [\varphi]$ for $1\leq i\leq l$ or $G^{(k_{l+1}+1)}[\varphi]$,
  \end{enumerate}
  for each $l=1,\ldots,k$ with $k_0+\cdots+k_{l+1}=k$, followed by an integration over 
  $\M$ = smearing with the test function $f(x)\equiv 1$. The pairing of variables in 
  such a composition for each term in the right-hand side of \eqref{s3e60} and \eqref{s3e61} is 
  of the following form:
  \begin{itemize}
  \item The first variable of the kernel of the first propagator pairs with the first 
    variable of $F^{(k_1+1)}[\varphi]$;
  \item The first variable of $D^{k_i+2}\Li(1)[\varphi]$ pairs with the second variable
    of the kernel of the $i$-th propagator;
  \item The second variable of $D^{k_i+2}\Li(1)[\varphi]$ pairs with the first variable
    of the kernel of the $(i+1)$-th propagator;
  \item The second variable of the kernel of the last propagator pairs with the first 
    variable of $G^{(k_{l+1}+1)}[\varphi]$.
  \end{itemize}
  It is clear that such a string of compositions is well defined. Finally, the 
  smearing with $\vec{\varphi}_1,\ldots,\vec{\varphi}_k$ is allowed since the 
  distribution obtained is compactly supported. The proof is complete.
\end{proof}

By \eqref{s3e38}, $\{F,G\}_\Li$ is actually defined for any pair of smooth functionals 
$F,G$ with compact space-time support such that $\WF(F^{(1)}[\varphi])$ and $\WF(G^{(1)}
[\varphi])$ do not contain any causal covectors with respect to $g$ for all $\varphi
\in\UU$, provided that $\hat{g}_\Li[\varphi]\lesssim g$ for all such $\varphi$. This 
motivates the following

\begin{definition}\label{s3d6}
  Let $(\M,g)$ be a globally hyperbolic space-time. Define for all $k\geq 1$ the open 
  subsets $\Upsilon_{k,g}\subset T^*\!\!\M^k\sm 0$ as follows:
  \begin{equation}\label{s3e70}
    \begin{split}
      \Upsilon_{k,g}\ =&\ \{(x_1,\ldots,x_k;\xi_1,\ldots,\xi_k)\in T^*\!\!\M^k\sm 0\ |\ \\
      & (\xi_1,\ldots,\xi_k)\not\in\ol{V}^k_{+,g}(x_1,\ldots,x_k)\cup\ol{V}^k_{-,g}(x_1,
      \ldots,x_k)\}\ ,\\ \ol{V}^k_{\pm,g}(x_1,\ldots,x_k)\ =&\ \prod^k_{j=1}\ol{V}_{\pm,g}(x_j)\ ,\\
      V_{\pm,g}(x)\ =&\ I^\pm(0,g^{-1}(x))\subset T^*_x\M\ ,\,x\in\M\ .
    \end{split}
  \end{equation}
  Let now $\UU\subset\C^\infty(\M)$ be open in the compact-open topology. We say that a
  smooth functional $F$ with compact space-time support is \emph{microcausal} with respect 
  to $g$ if $\WF(F^{(k)}[\varphi])\subset\Upsilon_{k,g}$ for all $\varphi\in\UU$, $k\geq 1$. 
  The space of all microcausal functionals in $\UU$ with respect to $g$ is denoted by 
  $\Fun((\M,g),\UU)$.
\end{definition}

We obviously have that $\Fun_0(\M,\UU)\subset\Fun((\M,g),\UU)$. A much more interesting 
inclusion is given by the following

\begin{proposition}\label{s3p3}
  Let $\UU\subset\C^\infty(\M)$ be open in the compact-open topology, $F\in\Fun_{\mu\loc}(\M,\UU)$. 
  Then $\WF(F^{(k)}[\varphi])\perp T\Delta_k(\M)$ for all $\varphi\in\UU$, $k\geq 2$. 
  In particular, $\Fun_{\mu\loc}(\M,\UU)\subset\Fun((\M,g),\UU)$.
\end{proposition}
\begin{proof}
  Let $\varphi\in\UU$. For all $k\geq 2$, $F^{(k)}[\varphi]$ is the kernel of a $(k-1)$-linear,
  $(k-1)$-differential operator taking values in the vector bundle of $\CC$-valued $d$-forms, as 
  shown by Propositions \ref{s2p1} and \ref{s2p2}. That is, if $\vec{\varphi}_1,\ldots,
  \vec{\varphi}_{k-1}\in\C^\infty(\M)$, then $F^{(k)}[\varphi](\cdot,\vec{\varphi}_1,\ldots,
  \vec{\varphi}_{k-1})$ can be thought of locally in $\M$ as a sum of products of $d$-form-valued 
  linear partial differential operators acting on $\vec{\varphi}_1$ multiplied by a product of 
  derivatives of $\vec{\varphi}_j$ for all $2\leq j\leq k-1$. This means that $F^{(k)}[\varphi]$
  can be written locally as a finite sum of derivatives of the Dirac kernel $\delta_k$ in 
  $\M^k$, defined by
  \[
  \delta_k(\omega\otimes\vec{\varphi}_1\otimes\cdots\otimes\vec{\varphi}_{k-1})\doteq\int_\M
  \prod^{k-1}_{j=1}\vec{\varphi}_j\omega\ ,\quad\omega\in\Gamma^\infty_c(\wedge^dT^*\!\!\M\To\M)\ ,
  \]
  with each term of the sum evaluated at a possibly different, $\varphi$-dependent $\omega$. 
  Since $\delta_k$ is simply the pullback of the constant function $u\equiv 1$ (seen as a 
  distribution in $\M^k$) by the inclusion $\Delta_k(\M)\hookrightarrow\M^k$, the assertion 
  follows from Theorem 8.2.4, pp. 263--265 of \cite{horm1}.
\end{proof}

Proposition \ref{s3p3} justifies the term ``microlocal'' for designating the elements of 
$\Fun_{\mu\loc}(\M,\UU)$, establishing the link with the notion of local functional employed 
in \cite{brudf}. One may wonder whether locality in the sense of Definition \ref{s2d4} (i.e.
through formula \eqref{s2e20}) and microcausality together entail microlocality, as claimed
e.g. in Section 2 of \cite{dutfre2} (more precisely, see formula (2.8), pp. 1296). This happens
to be \emph{false}, as example \eqref{s2e25} shows -- there $F^{(k)}\equiv 0$ for $k>1$ but
$\WF(F^{(1)}[\varphi])$ is conormal to $\N$, hence it consists of spacelike covectors only.
This shows that such an $F$ is microcausal. However, we have seen that $F$ is local but not
microlocal, thus establishing our claim.

It is of paramount importance that the Peierls bracket can actually be extended from microlocal 
to arbitrary microcausal functionals.

\begin{theorem}\label{s3t3} 
  Let $\UU, \Li$ be as in Proposition \ref{s3p2}. The Peierls bracket associated with any such 
  $\Li$ extends to the whole of $\Fun((\M,g),\UU)$, possesses the support property \eqref{s3e68} 
  and depends \emph{only locally on} $\Li$ -- that is, for all $F,G\in\Fun((\M,g),\UU)$ we have 
  that $\{F,G\}_\Li$ is \emph{unaffected} by perturbations of $\Li$ outside $\OO_{\supp F,\supp G}$. 
  Likewise, for all $F,G\in\Fun_{\mu\loc}(\M,\UU)$ we have that $\Ret_\Li(F,G)$ (resp. $\Adv_\Li(F,
  G)$) is unaffected by perturbations of $\Li$ outside $\OO^\Rt_{\supp F,\supp F}$ (resp. 
  $\OO^\Av_{\supp F,\supp F}$).
\end{theorem}
\begin{proof}
  We first check whether the Peierls bracket is well defined when extended to $\Fun((\M,g),\UU)$. 
  As argued right after the proof of Proposition \ref{s3p2}, since the wave front set of the 
  first derivative of a microcausal functional contains only spacelike covectors and, by 
  \eqref{s3e38}, the wave front set of $\Delta_{\Li}[\varphi]$ contains only pairs of null 
  covectors, which after parallel transport along a null geodesic add to zero, the term 
  $\Delta_{\Li}[\varphi]G^{(1)}[\varphi]$ is smooth and can therefore be integrated with the 
  compactly supported distributional density $F^{(1)}[\varphi]$. The proof of \eqref{s3e68} 
  then carries through \emph{ipsis literis} as in the case that $F$ and $G$ are microlocal 
  (Proposition \ref{s3p2}).

  Concerning the dependence of the Peierls bracket on local data of $\Li$, let us first pick two
  arbitrary microcausal functionals $F,G$. Now, let two Lagrangians $\Li_1,\Li_2$ satisfy the 
  hypotheses of Proposition \ref{s3p2} and such that we have for any $\varphi\in\UU$ $E'(\Li_1)
  [\varphi]$ and $E'(\Li_2)[\varphi]$ differ only outside $\OO_{\supp F,\supp G}$. More precisely, 
  we suppose that 
  \begin{equation}\label{s3e71}
  \supp(E'(\Li_1)[\varphi]-E'(\Li_2)[\varphi])\cap\OO_{\supp F,\supp G}=\varnothing 
  \end{equation}
  for all $\varphi\in\UU$. We see that
  \begin{align}
    \left\langle\Delta_{\Li_1}^\Av[\varphi]F^{(1)}[\varphi],(E^\prime(\Li_2)[\varphi]\right.
    &\left.-E^\prime(\Li_1)[\varphi])\Delta^\Rt_{\Li_2}[\varphi]G^{(1)}[\varphi]\right\rangle
      \label{s3e72}\\ &=\Spr{F^{(1)}[\varphi],(\Delta^\Rt_{\Li_1}[\varphi]-\Delta^\Rt_{\Li_2}[\varphi])
      G^{(1)}[\varphi]}=0\ ,\nonumber\\
    \left\langle\Delta_{\Li_1}^\Rt[\varphi]F^{(1)}[\varphi],(E^\prime(\Li_2)[\varphi]\right.
    &\left.-E^\prime(\Li_1)[\varphi])\Delta^\Av_{\Li_2}[\varphi]G^{(1)}[\varphi]\right\rangle
      \label{s3e73}\\ &=\Spr{F^{(1)}[\varphi],(\Delta^\Av_{\Li_1}[\varphi]-\Delta^\Av_{\Li_2}[\varphi])
      G^{(1)}[\varphi]}=0\nonumber
  \end{align}
  for all $\varphi\in\UU$ thanks to \eqref{s3e71}, which also guarantees that the left-hand 
  sides of \eqref{s3e72} and \eqref{s3e73} are well defined since there are no common base 
  points in the wave front sets of either side for any of the dual pairings involved therein.
  This already entails the desired properties for $\Ret_\Li(F,G)$ and $\Adv_\Li(F,G)$ if $F,G$
  are microlocal, since \eqref{s3e72} (resp. \eqref{s3e73}) remain valid if we allow $E'(\Li_1)
  [\varphi]$ and $E'(\Li_2)[\varphi]$ to differ only outside $\OO^\Rt_{\supp F,\supp G}$ (resp.
  $\OO^\Av_{\supp F,\supp G}$) for all $\varphi\in\UU$. As for $\{F,G\}_\Li$, we have from 
  \eqref{s3e72} and \eqref{s3e73} that
  \begin{equation}\label{s3e74}
    \begin{split}
      \{F,G\}_{\Li_1}(\varphi)-\{F,G\}_{\Li_2}(\varphi) &=\Spr{F^{(1)}[\varphi],
      (\Delta_{\Li_1}[\varphi]-\Delta_{\Li_2}[\varphi])G^{(1)}[\varphi]}\\ &=\Spr{F^{(1)}[\varphi],
      (\Delta_{\Li_1}^\Rt[\varphi]-\Delta_{\Li_2}^\Rt[\varphi])G^{(1)}[\varphi]}\\ &\phantom{=}
    -\Spr{F^{(1)}[\varphi],(\Delta_{\Li_1}^\Av[\varphi]-\Delta_{\Li_2}^\Av[\varphi])G^{(1)}[\varphi]}\\ 
    &=0
    \end{split}
  \end{equation}
  for all $\varphi\in\UU$, as asserted.
\end{proof}

We have now the following strengthening of Proposition \ref{s3p2} and Theorem \ref{s3t3}, which 
is crucial to this whole Subsection and justifies the christening ``microcausal'' given to the 
elements of $\Fun((\M,g),\UU)$. We shall take advantage of the fact that, thanks to Theorem 
\ref{s3t3}, we may replace $E(\Li)[\varphi]$ by its \emph{cutoff version}
\begin{equation}\label{s3e75}
  E'(\Li)[\varphi_0]\varphi+f(E(\Li)[\varphi]-E'(\Li)[\varphi_0]\varphi)
\end{equation}
with any $f\in\C^\infty_c(\M)$ such that $f\equiv 1$ in a neighborhood of $(J^+(\supp F,g)\cup 
J^-(\supp F,g))\cap(J^+(\supp G,g)\cup J^-(\supp G,g))$ while keeping $\{F,G\}_\Li(\varphi)$
\emph{unaltered} for all $\varphi,\varphi_0\in\UU$. The term in the right-hand side of 
\eqref{s3e75} proportional to the cutoff function $f$ corresponds to the nonlinear 
(interaction) part of $E(\Li)$ around the background field configuration $\varphi_0$. 

\begin{proposition}\label{s3p4}
  Let $\UU,\Li$ be as in Proposition \ref{s3p2}, and $F,G\in\Fun((\M,g),\UU)$. Then 
  $\{F,G\}_\Li$ also belongs to $\Fun((\M,g),\UU)$.
\end{proposition}
\begin{proof}    
  We look at the derivatives of $\{F,G\}_\Li$. To that end, we shall replace $E(\Li)$ by
  its cutoff version \eqref{s3e75} in order to make the distribution kernel of $D^2E(\Li)
  [\varphi]$ \emph{compactly supported}. This will allow us to obtain a technically more 
  convenient formula for the derivatives of the causal propagator. In what follows we shall 
  use the same notation for the cutoff Euler-Lagrange operator for simplicity. Since 
  $\Delta_\Li[\varphi]=\Delta^\Rt_\Li[\varphi]-\Delta^\Av_\Li[\varphi]$, formulae \eqref{s3e56} 
  and \eqref{s3e57} together imply
  \begin{equation}\label{s3e76}
    D\Delta_\Li[\varphi](\vec{\varphi})=-\Delta_\Li[\varphi]D^2E(\Li)[\varphi](\vec{\varphi})
    \Delta^\Rt_\Li[\varphi]-\Delta^\Av_\Li[\varphi]D^2E(\Li)[\varphi](\vec{\varphi})\Delta_\Li
    [\varphi]\ .
  \end{equation}
  In particular, such a formula implies that $D\Delta_\Li[\varphi](\vec{\varphi})$ has 
  the same wave front set as $\Delta_\Li[\varphi]$. As for higher orders, one obtains that
  \begin{equation}\label{s3e77}
    \begin{split}
      D^k\Delta_\Li&[\varphi](\vec{\varphi}_1,\ldots,\vec{\varphi}_k)  \\
      &=\sum^k_{l=0}(-1)^l\sum_{\{I_1,\ldots,I_l\}\in P_k}\sum_{\sigma\in S_l}\sum^l_{m=0}\left(\prod^m_{j=1}
        \Delta^\Av_\Li[\varphi_0]D^{|I_{\sigma(j)}|+1}E(\Li)[\varphi_0](\otimes_{i\in I_{\sigma(j)}}
        \vec{\varphi}_i)\right)\\ &\phantom{=\sum^k_{l=0}(-1)^l}
      \cdot\Delta_\Li[\varphi]\left(\prod^l_{j=m+1}D^{|I_{\sigma(j)}|+1}E(\Li)
        [\varphi_0](\otimes_{i\in I_{\sigma(j)}}\vec{\varphi}_i)\Delta^\Rt_\Li[\varphi_0]\right)\ .
    \end{split}
  \end{equation}
  and hence the $k$-th order functional derivative of the Peierls bracket at $\varphi$  
  is formally given by
  \begin{equation}\label{s3e78}
    \begin{split}
      D^k\{F,G\}_\Li[\varphi](\vec{\varphi}_1,\ldots,\vec{\varphi}_k) &=\sum_{\{J_1,J_2,J_3\}\subset P_k}
      F^{(|J_1|+1)}[\varphi]((\otimes_{j_1\in J_1}\vec{\varphi}_{j_1})\\ &\otimes D^{|J_2|}\Delta_\Li
      [\varphi]((\otimes_{j_2\in J_2}\vec{\varphi}_{j_2})\otimes G^{(|J_3|+1)}[\varphi]
      (\otimes_{j_3\in J_3}\vec{\varphi}_{j_3})))
    \end{split}
  \end{equation}
  with $D^k\Delta_\Li[\varphi](\vec{\varphi}_1,\ldots,\vec{\varphi}_k)$ given as above.
  Moreover, since $F$ and $G$ are microcausal, the wave front sets of their derivatives 
  contain no elements where either
  \begin{itemize}
  \item All covectors are in the closed forward light cone $\ol{V}_+$, or 
  \item All are in the closed backward light cone $\ol{V}_-$. 
  \end{itemize}
  Recall as well that since $D^kE(\Li)[\varphi]$ is a $k$-linear partial differential 
  operator with distribution kernel $D^{k+1}\Li(1)[\varphi]$, we conclude that (see 
  the proof of Proposition \ref{s3p3} for more details)
  \[
  \WF(D^{k+1}\Li(1)[\varphi])\subset N^*\Delta_{k+1}(\M)\sm 0\ ,
  \]
  where
  \[
  N^*\Delta_k(\M)=\{(x_1,\ldots,x_k;\xi_1,\ldots,\xi_k)\in T^*\!\!\M^k\ |\ x_1=\cdots=x_k,\,
  \xi_1+\cdots+\xi_k=0\}\ ,\quad k\geq 2
  \]
  is the \emph{conormal bundle} to the (small) diagonal $\Delta_k(\M)$ of $\M^k$.
  
  By a reasoning similar to that employed in the proof of Proposition \ref{s3p2}, 
  we notice that due to \eqref{s3e77} each term in the right-hand side of \eqref{s3e78} before
  smearing with $\vec{\varphi}_1,\ldots,\vec{\varphi}_k$ can be seen as a string of 
  compositions of:
  \begin{enumerate}
  \item[(i)] $l+1$ propagators either of the form $\Delta^\Av_\Li[\varphi]$, $\Delta_\Li
    [\varphi]$ or $\Delta^\Rt_\Li[\varphi]$; and 
  \item[(ii)] $l+2$ $k_i$-linear operators, $i=0,\ldots,l+1$ whose distribution kernels 
    are either of the form $F^{(k_0+1)}[\varphi]$, $D^{k_i+2}\Li(1)[\varphi]$ for $1\leq i
    \leq l$ or $G^{(k_{l+1}+1)}[\varphi]$,
  \end{enumerate}
  for each $l=1,\ldots,k$ with $k_0+\cdots+k_{l+1}=k$. Moreover, the obtained distributions 
  are once again compactly supported. The pairing of variables in such a composition for 
  each term in the right-hand side of \eqref{s3e77} is of the following form:
  \begin{itemize}
  \item The first variable of the kernel of the first propagator pairs with the first 
    variable of $F^{(k_1+1)}[\varphi]$;
  \item The first variable of $D^{k_i+2}\Li(1)[\varphi]$ pairs with the second variable
    of the kernel of the $i$-th propagator;
  \item The second variable of $D^{k_i+2}\Li(1)[\varphi]$ pairs with the first variable
    of the kernel of the $(i+1)$-th propagator;
  \item The second variable of the kernel of the last propagator pairs with the first 
    variable of $G^{(k_{l+1}+1)}[\varphi]$.
  \end{itemize}
  In particular, the kernel of the causal propagator $\Delta_\Li[\varphi]$ has its
  first variable paired with the second variable of $D^{k_m+2}\Li(1)[\varphi]$ and
  its second variable paired with the first variable of $D^{k_{m+1}+2}\Li(1)[\varphi]$.
  
  Suppose now that $(\xi_1,\ldots,\xi_k)\in\ol{V}^k_{+,g}(x_1,\ldots,x_k)$ is in
  $\WF(D^k\{F,G\}_\Li[\varphi])$. If $(y_m,z_m;$ $\eta_m,\zeta_m)\in \WF(\Delta_\Li[\varphi])$,
  then either $\eta_m$ or $\zeta_m$ is a past directed null covector w.r.t. $\hat{g}_\Li
  [\varphi]$. Suppose it is $\eta_m$, so that $\eta_m\in\ol{V}_{-,g}(y_m)$ -- then by 
  Theorem 8.2.14, pp. 269--270 of \cite{horm1} we must have that
  \[
  \begin{split}
    (z_{m-1},y_m,x_{k_0+\cdots+k_{m-1}+1},\ldots,x_{k_0+\cdots+k_m}&;\zeta_{m-1},-\eta_m,
    \xi_{k_0+\cdots+k_{m-1}+1},\ldots,\xi_{k_0+\cdots+k_m})\\ &\in \WF(D^{k_m+2}\Li(1)[\varphi])\ ,
  \end{split}
  \]
  implying that $z_{m-1}=y_m$ and $\zeta_{m-1}\in\ol{V}_{-,g}(z_{m-1})$. Now, if $(y_{m-1},z_{m-1};
  \eta_{m-1},\zeta_{m-1})\in \WF(\Delta^\Av_\Li[\varphi])$ for some $(y_{m-1},\eta_{m-1})$, then
  either $(y_{m-1},\eta_{m-1})=(z_{m-1},\zeta_{m-1})$ or $\zeta_{m-1}$ is null past directed
  and therefore $\eta_{m-1}$ is null future directed. In either case, we have that
  $\eta_{m-1}\in\ol{V}_{+,g}(y_{m-1})$. Repeating the above procedure backwards as many times
  as needed as dictated by Theorem 8.2.14, pp. 269--270 of \cite{horm1}, we conclude that 
  \[
  (x_1,\ldots,x_{k_0},y_1;\xi_1,\ldots,\xi_{k_0},\eta_1)\in \WF(F^{(k_1+1)}[\varphi])
  \] 
  with $\eta_1\in\ol{V}_{+,g}(y_1)$, which is absurd since $F$ is assumed to be microcausal. 
  Likewise, if instead $\zeta_m$ is null past directed, proceeding as above but forwards we 
  conclude that
  \[
  (z_l,x_{k_0+\cdots+k_l+1},\ldots,x_{k_0+\cdots+k_{l+1}};\zeta_l,\xi_{k_0+\cdots+k_l+1},\ldots,
  \xi_{k_0+\cdots+k_{l+1}})\in \WF(G^{(k_{l+1}+1)}[\varphi])
  \]
  with $\zeta_l\in\ol{V}_{+,g}(z_l)$, which is absurd since $G$ is assumed to be microcausal.
  In the same fashion, we conclude that no $(\xi_1,\ldots,\xi_k)\in\ol{V}^k_{-,g}(x_1,\ldots,
  x_k)$ can belong to $\WF(D^k\{F,G\}_\Li[\varphi])$. In particular, we see that no $k$-tuple 
  of zero covectors can arise from the above procedure, hence again by Theorem 8.2.14, pp. 
  269--270 of \cite{horm1} the (compactly supported) distribution $D^k\{F,G\}_\Li[\varphi]$ 
  is well defined. The proof is complete.
\end{proof}

We stress that the presence of $\Delta_\Li[\varphi]$ is crucial for the propagation argument 
underlying the proof of Proposition \ref{s3p4} to work, since it prevents the appearance of 
spacelike covectors which may disrupt the propagation procedure. Such an argument would not 
work if we had only retarded or only advanced propagators in each term of \eqref{s3e78}, 
because their wave front set may have elements over the diagonal whose covectors are 
spacelike. On the other hand, as the proof of Proposition \ref{s3p2} shows, in the case of
microlocal $F,G$ the clash of (spacelike) covectors though the propagation procedure is 
prevented by the fact that the wave front sets of $F^{(k)}[\varphi]$ and $G^{(l)}[\varphi]$
are conormal to $\Delta_k(\M)$ and $\Delta_l(\M)$ respectively for all $k,l>0$. More generally,
if $\M$ is parallelizable (e.g. if $(\M,g)$ is Minkowski space-time or if $d=4$ \cite{stiefel}) 
then one may define for each $k\geq 2$ the sets
\[
N^k(\M)=\{(x_1,\ldots,x_k;\xi_1,\ldots,\xi_k)\in T^*\!\!\M^k\sm 0\ |\ \xi_1+\cdots+\xi_k=0\}\ ,
\,N^1(\M)=\varnothing\ .
\]
If to deem $F,G$ as microcausal we required in addition to Definition \ref{s3d6} that $\WF(F^{(k)}
[\varphi])$, $\WF(G^{(k)}[\varphi])\subset N^k(\M)$ for all $k\geq 1$, one would be able to conclude 
by the same reasoning as in the proof of Proposition \ref{s3p2} that $\Ret_\Li(F,G)$ and 
$\Adv_\Li(F,G)$ are smooth, in fact even microcausal in this strengthened sense. This was the
path followed e.g. by \cite{dutfre1} in Minkowski space-time, see discussion between formulae 
(8) and (9), pp. 280 therein. However, it is clear that the definition of $N^k(\M)$ is tied to
a choice of global trivialization for $T^*\!\!\M$ (tacitly assumed therein), which is natural in 
the case of Minkowski space-time but generally no longer so, thus such requirement seems unnatural
in curved space-times.

\begin{corollary}\label{s3c3} 
  The Peierls bracket $F,G\mapsto\{F,G\}_\Li$ defines a Lie bracket on $\Fun((\M,g),\UU)$ 
  for any globally hyperbolic metric $g$ on $\M$ such that $g\gtrsim\hat{g}_\Li[\varphi]$ 
  for all $\varphi\in\UU$.
\end{corollary}
\begin{proof}
  $\{\cdot,\cdot\}_\Li$ is clearly bilinear. Antisymmetry of $\{\cdot,\cdot\}_\Li$ 
  follows from the argument right after Definition \ref{s3d5}. All that is left 
  to us is to prove that the Jacobi identity holds, that is, 
  \begin{equation}\label{s3e79}
    \{F,\{G,H\}_\Li\}_\Li+\{G,\{H,F\}_\Li\}_\Li+\{H,\{F,G\}_\Li\}_\Li=0
  \end{equation}
  for any $F,G,H\in\Fun((\M,g),\UU)$. To that end, we argue as in the proof of Proposition 
  \ref{s3p4} and replace once more $E(\Li)$ by the cutoff version \eqref{s3e75} while 
  keeping the same notation, this time with the cutoff function $f$ such that $f\equiv 1$ 
  in a neighborhood of the (compact) region $\OO_{\supp F,\supp G,\supp H}$, where for $\varnothing
  \neq K,L,M\subset\M$ we set
  \[
  \OO_{K,L,M}=\OO_{K,L}\cup\OO_{K,M}\cup\OO_{L,M}\cup\OO_{\OO_{K,L},M}\cup\OO_{\OO_{L,M},K}
  \cup\OO_{\OO_{M,K},L}\ .
  \]
  with $\OO_{K,L}$ defined as in \eqref{s3supps}. It is clear from Theorem \ref{s3t3} that all 
  Peierls brackets involved in the left-hand side of \eqref{s3e79} remain unaltered by the 
  cutoff. Now we have that
  \begin{equation}\label{s3e80}
    \begin{split}
      \{F,\{G,H\}_\Li\}_\Li(\varphi) &=H^{(2)}[\varphi](\Delta_\Li[\varphi]G^{(1)}[\varphi],
      \Delta_\Li[\varphi]F^{(1)}[\varphi])\\ &\phantom{=}-G^{(2)}[\varphi](\Delta_\Li[\varphi]
      H^{(1)}[\varphi],\Delta_\Li[\varphi]F^{(1)}[\varphi])\\ &\phantom{=}-G^{(1)}[\varphi]
      (D\Delta_\Li[\varphi](\Delta_\Li[\varphi]F^{(1)}[\varphi])H^{(1)}[\varphi])\ .
    \end{split}
  \end{equation}
  Consider the first two terms in the right-hand side of \eqref{s3e80}. Summing them along
  all three cyclic permutations of $F,G,H$ yields zero thanks to the symmetry of
  second-order derivatives of functionals in their linear entries, so one is only 
  left to show that
  \begin{equation}\label{s3e81}
    \begin{split}
      G^{(1)}[\varphi](D\Delta_\Li[\varphi](\Delta_\Li[\varphi]F^{(1)}[\varphi])H^{(1)}[\varphi])
      &+F^{(1)}[\varphi](D\Delta_\Li[\varphi](\Delta_\Li[\varphi]H^{(1)}[\varphi])
      G^{(1)}[\varphi])\\ &+H^{(1)}[\varphi](D\Delta_\Li[\varphi](\Delta_\Li[\varphi]
      G^{(1)}[\varphi])F^{(1)}[\varphi])=0\ .
    \end{split}
  \end{equation}
  Inserting into \eqref{s3e81} the formula \eqref{s3e76} for the derivative of 
  $\Delta_\Li[\varphi]$ obtained in the proof of Proposition \ref{s3p4} we find that
  \begin{equation}\label{s3e82}
    \begin{split}
      G^{(1)}[\varphi]&(D\Delta_\Li[\varphi](\Delta_\Li[\varphi]F^{(1)}[\varphi])H^{(1)}[\varphi])\\
      &=-G^{(1)}[\varphi](\Delta_\Li[\varphi]D^2E(\Li)[\varphi](\Delta_\Li[\varphi]F^{(1)}[\varphi])
      \Delta^\Rt_\Li[\varphi]H^{(1)}[\varphi])\\ &\phantom{=}-G^{(1)}[\varphi](\Delta^\Av_\Li[\varphi]
      D^2E(\Li)[\varphi](\Delta_\Li[\varphi]F^{(1)}[\varphi])\Delta_\Li[\varphi]H^{(1)}[\varphi])\\
      &=D^3\Li(1)[\varphi](\Delta_\Li[\varphi]G^{(1)}[\varphi],\Delta^\Rt_\Li[\varphi]H^{(1)}[\varphi],
      \Delta_\Li[\varphi]F^{(1)}[\varphi])\\ &\phantom{=}-D^3\Li(1)[\varphi](\Delta_\Li[\varphi]
      H^{(1)}[\varphi],\Delta^\Rt_\Li[\varphi]G^{(1)}[\varphi],\Delta_\Li[\varphi]F^{(1)}[\varphi])\ ,
    \end{split}
  \end{equation}
  where in the last identity we have exploited the symmetry of $D^3\Li(1)[\varphi]$ in its
  first two entries. Summing the last formula of \eqref{s3e82} along all three cyclic 
  permutations of $F,G,H$ yields zero once more thanks to the symmetry of $D^3\Li(1)[\varphi]$ 
  in its first and last entries. The proof is complete.
\end{proof}

Notice that the arguments employed in the proofs of Proposition \ref{s3p4} and Corollary 
\ref{s3c3} rely on the compactness of the support of the distribution kernel of $D^2E(\Li)
[\varphi]$ through the formula \eqref{s3e76} for the derivative of $\Delta_\Li[\varphi]$, 
for therein one adds and subtracts a term of the form 
\[
\Delta^\Av_\Li[\varphi]D^2E(\Li)[\varphi](\vec{\varphi})\Delta^\Rt_\Li[\varphi]
\] 
which is otherwise ill defined. However, as argued right after the proof of Theorem \ref{s3t3}, 
this entails no loss of generality since we can perform a suitable cutoff of the nonlinear
part of $E(\Li)$ through \eqref{s3e75}.

We shall prove in Section \ref{s4-gen} that $\Fun((\M,g),\UU)$ is closed under products 
and that the Peierls bracket satisfies Leibniz's rule (Theorem \ref{s4t1}). In other words, 
$\Fun((\M,g),\UU)$ becomes a \emph{Poisson} algebra when endowed with the Peierls bracket 
associated to $\Li$. 

\section{\label{s4-gen}First structural results}

With the body of results of Sections \ref{s2-kin} and \ref{s3-dyn} at hand, we can 
start a detailed and motivated discussion of the mathematical structures underlying 
our approach.

\subsection{\label{s4-gen-top}Topology of the space of microcausal functionals}

We can endow $\Fun((\M,g),\UU)$ with a topology which, despite being quite weak,
accommodates rather well our algebraic operations. The weakest possible choice
is the topology of pointwise convergence of functionals and their derivatives
of all orders, which is the locally convex topology on $\Fun((\M,g),\UU)$ induced 
by the separating system of seminorms
\[
F\mapsto|F^{(k)}[\varphi](\vec{\varphi}_1,\ldots,\vec{\varphi}_k)|\ ,\quad
\varphi\in\UU,\,\vec{\varphi}_1,\ldots,\vec{\varphi}_k\in\C^\infty(\M),\,k\in\NN\ .
\]
Equivalently, this topology is the initial locally convex topology on $\Fun((\M,g),\UU)$ 
induced by the linear maps
\begin{equation}\label{s4e1}
  F\mapsto 
  \begin{cases}
    F(\varphi)\in\CC & (k=0)\\
    F^{(k)}[\varphi]\in\E'(\wedge^{kd}T^*\!\!\M^k\To\M) & (k\geq 1)
  \end{cases}\ ,\quad\varphi\in\UU\ ,
\end{equation}
where the space $\E'(\wedge^{kd}T^*\!\!\M^k\To\M)$ of $d$-form-valued distributions of
compact support on $\M^k$ is the topological dual of $\C^\infty(\M^k)$. This choice, 
however, ignores the extra information on the wave front sets of $F^{(k)}[\varphi]$ 
which enters Definition \ref{s3d6} for microcausal functionals. A more natural choice 
is to replace the spaces of general, compactly supported distribution densities in 
\eqref{s4e1} for each $k\geq 1$ by the following subspaces:
\begin{equation}\label{s4e2}
  \E'_{\Upsilon_{k,g}}(\wedge^{kd}T^*\!\!\M^k\To\M^k)=\{u\in\E'(\wedge^{kd}T^*\!\!\M^k\To\M^k)
  \ |\ \WF(u)\subset\Upsilon_{k,g}\}\ .
\end{equation}
These, however, are not standard spaces of compactly supported distributions with
wave front sets within a prescribed (closed) cone, for $\Upsilon_{k,g}$ as defined in \eqref{s3e66}
is an \emph{open} conic subset of $T^*\!\!\M^k\sm 0$. Therefore, one cannot immediately endow 
$\E'_{\Upsilon_{k,g}}(\wedge^{kd}T^*\!\!\M^k\To\M^k)$ with the Hörmander topology (see, for 
instance, Section 8.2 of \cite{horm1}). It is possible, on the other hand, to define 
$\E'_{\Upsilon_{k,g}}(\wedge^{kd}T^*\!\!\M^k\To\M^k)$ as an inductive limit of an increasing 
sequence of spaces of compactly supported distributions with wave front sets contained in 
an increasing sequence of \emph{closed} conic subsets of $T^*\!\!\M^k\sm 0$, each of 
these spaces being endowed with the Hörmander topology. The key result which allows us to 
do this is the following

\begin{lemma}\label{s4l1}
  For each $k=1,2,\ldots$ there is a countable family $\{\Gamma_{k,m}\}_{m\in\NN}$ of closed
  conic subsets of $T^*\!\!\M^k$ such that $\Gamma_{k,m}\subset\In{\Gamma}_{k,m+1}$ and $\cup^\infty_{m=0}
  \Gamma_{k,m}=\Upsilon_{k,g}$ is given by \eqref{s3e70}.
\end{lemma}
\begin{proof}
  Let $\omega$ be a future directed timelike covector field in $(\M,g)$ and $\epsilon>0$ such
  that $g_\epsilon\doteq g-\epsilon u\otimes u$ is a Lorentzian metric. We have that 
  $g<g_{\epsilon'}<g_\epsilon$ and hence $V_{\pm,g_\epsilon}(x)\supset\ol{V}_{\pm,g_{\epsilon'}}(x)$ for all 
  $0<\epsilon'<\epsilon$, $x\in M$. Let now $(\epsilon_m)_{m\in\NN}$ be a sequence of positive real 
  numbers such that $\epsilon_0=\epsilon$, $\epsilon_{m+1}<\epsilon_m$ and 
  $\epsilon_m\gotoas{0}{m}{\infty}$. We conclude that, for all $x\in\M$,
  \[
  \begin{split}
    \complement(\ol{V}_{+,g}(x)\cup\ol{V}_{-,g}(x)) &=\left(\bigcup^\infty_{m=0}
      \complement(V_{+,g_{\epsilon_m}}(x)\cup V_{-,g_{\epsilon_m}}(x))\right)\sm\{0\}\\
    &=\left(\bigcup^\infty_{m=0}\complement(V_{+,g_{\epsilon_m}}(x)
      \cup V_{-,g_{\epsilon_m}}(x))\sm\{0\}\right)\ .
  \end{split}
  \]
  The above argument settles the case $k=1$. For $k>1$, we can write $\complement(\ol{V}^k_{+,g}
  (x_1,\ldots,x_k)\cup\ol{V}^k_{-,g}(x_1,\ldots,x_k))$ as a union of subsets of the form $\Omega=
  \prod^k_{j=1}W_j$, such that the possibilities for $\Omega$ fall in exactly one of the following 
  three categories:
  \begin{enumerate}
  \item[(a)] $W_j=\complement(\ol{V}_{+,g}(x_j)\cup\ol{V}_{-,g}(x_j))$ for at least one $j$, and
    all $W_{j'}$'s which are not of this form are of the form $W_{j'}=\ol{V}_{+,g}(x_{j'})\cup
    \ol{V}_{-,g}(x_{j'})$. There are $\sum_{k'=1}^k\binom{k}{k'}=2^k-1$ such $\Omega$'s.
  \item[(b)] For all $j=1,\ldots,k$, we have either $W_j=\ol{V}_{+,g}(x_j)\sm\{0\}$ or
    $W_j=\ol{V}_{-,g}(x_j)\sm\{0\}$, and there is at least one pair ${j,j'}\subset\{1,\ldots,k\}$,
    such that $W_j=\ol{V}_{+,g}(x_j)\sm\{0\}$ and $W_{j'}=\ol{V}_{-,g}(x_{j'})\sm\{0\}$. There 
    are $\sum_{k'=1}^{k-1}\binom{k}{k'}=2^k-2$ such $\Omega$'s (we remark that this number
    is zero for $k=1$).
  \item[(c)] $W_j=\{0\}$ for at least one $j$, all $W_{j'}$'s which are not of this form are 
    either of the form $W_{j'}=\ol{V}_{+,g}(x_{j'})\sm\{0\}$ or $W_{j'}=\ol{V}_{-,g}(x_{j'})\sm\{0\}$, 
    and there is at least one pair ${j',j''}\subset\{j=1,\ldots,k\ |\ W_j\neq\{0\}\}$ such that 
    $W_{j'}=\ol{V}_{+,g}(x_{j'})\sm\{0\}$ and $W_{j''}=\ol{V}_{-,g}(x_{j''})\sm\{0\}$. There are
    \[
    \begin{split}
      \sum^{k-2}_{k'=1}(2^{k-k'}-2)\binom{k}{k'} &=2^k\left(\frac{3^k}{2^k}-1-\frac{1}{2^k}
        -\frac{2k}{2^k}\right)-2(2^k-2-k)\\ &=3^k-3\cdot 2^k+3
    \end{split}
    \]
    such $\Omega$'s (we remark that this number is zero for $k=1,2$).
  \end{enumerate}
  Let us enumerate the $3^k-2^k$ subsets $\Omega$ listed above, so that the first $2^k-1$ 
  ones are of type (a), and the remaining ones are of types (b) and (c):
  \[
  \begin{split}
    \complement(\ol{V}^k_{+,g}(x_1,\ldots,x_k)\cup\ol{V}^k_{-,g}(x_1,\ldots,x_k)) 
    &=\bigcup^{3^k-2^k}_{l=1}\Omega_l\\ &=\left(\bigcup^{2^k-1}_{l=1}\Omega_l\right)\cup
    \left(\bigcup^{3^k-2^k}_{l=2^k}\Omega_l\right)\ ,\\ \Omega_l &=\prod^k_{j=1}W_{j,l}\ .
  \end{split}
  \]
  Let now $l<2^k$. We can write $\Omega_l$ as the countable union of an increasing sequence
  of closed conic subsets of $T^*_{(x_1,\ldots,x_k)}\M^k\sm 0$ 
  \[
  \Omega_l=\bigcup^\infty_{m=0}\Omega_{l,m}\ ,\,\Omega_{l,m}=\prod^k_{j=1}W_{j,l,m}\ ,
  \]
  where
  \[
  W_{j,l,m}=\begin{cases}
    \complement(V_{+,g_{\epsilon_m}}(x_j)\cup V_{-,g_{\epsilon_m}}(x_j))\sm\{0\} & 
    \text{ if }W_{j,l}=\complement(\ol{V}_{+,g}(x_j)\cup\ol{V}_{-,g}(x_j))\ ,\\
    W_{j,l} & \text{ if }W_{j,l}=\ol{V}_{+,g}(x_j)\cup\ol{V}_{-,g}(x_j)\ .
  \end{cases}
  \]
  If $l\geq 2^k$, $\Omega_l$ is already a closed conic subset of $T^*_{(x_1,\ldots,x_k)}\M^k\sm 0$. 
  Finally, define
  \[
  \Gamma_{k,m}(x_1,\ldots,x_k)=\left(\bigcup^{2^k-1}_{l=1}\Omega_{l,m}\right)\cup
  \left(\bigcup^{3^k-2^k}_{l=2^k}\Omega_l\right)\ .
  \]
  By construction, $\Gamma_{k,m}(x_1,\ldots,x_k)$ is a closed conic subset of $T^*_{(x_1,\ldots,x_k)}
  \M^k\sm 0$, $\Gamma_{k,m}(x_1,\ldots,x_k)\subset\In{\Gamma}_{k,m+1}(x_1,\ldots,x_k)$ and 
  \[
  \complement(\ol{V}^k_{+,g}(x_1,\ldots,x_k)\cup\ol{V}^k_{-,g}(x_1,\ldots,x_k))=\bigcup^\infty_{m=0}
  \Gamma_{k,m}(x_1,\ldots,x_k)
  \]
  for all $(x_1,\ldots,x_k)\in\M^k$. Taking $\Gamma_{k,m}$ as the disjoint union of the
  $\Gamma_{k,m}(x_1,\ldots,x_k)$'s for all $(x_1,\ldots,x_k)\in\M^k$ gives the thesis.
\end{proof}

\begin{corollary}\label{s4c1}
  One can write $\E'_{\Upsilon_{k,g}}(\wedge^{kd}T^*\!\!\M^k\To\M^k)$ for all $k=1,2,\ldots$ 
  as the countable inductive limit
  \begin{equation}\label{s4e3}
    \E'_{\Upsilon_{k,g}}(\wedge^{kd}T^*\!\!\M^k\To\M^k)=\dlim{m\in\NN}\E'_{\Gamma_{k,m}}(\wedge^{kd}
    T^*\!\!\M^k\To\M^k)
  \end{equation}
  of the spaces $\E'_{\Gamma_{k,m}}(\wedge^{kd}T^*\!\!\M^k\To\M^k)$. Let $\E'_{\Upsilon_{k,g}}
  (\wedge^{kd}T^*\!\!\M^k\To\M^k)$ be endowed with the locally convex inductive limit topology
  induced by the Hörmander topology on each $\E'_{\Gamma_{k,m}}(\wedge^{kd}T^*\!\M^k\To\M^k)$ for 
  all $k=1,2,\ldots$; then $\E'_{\Upsilon_{k,g}}(\wedge^{kd}T^*\!\!\M^k\To\M^k)$ is nuclear for
  all such $k$.
\end{corollary}
\begin{proof}
  By Lemma \ref{s4l1}, one has the inclusions
  \[
  \E'_{\Gamma_{k,m}}(\wedge^{kd}T^*\!\!\M^k\To\M^k)\subset\E'_{\Gamma_{k,m'}}
  (\wedge^{kd}T^*\!\!\M^k\To\M^k)
  \]
  for all $m<m'$. Since, given any closed conic subset $\Gamma\subset T^*\!\!\M^k\sm 0$, one 
  can construct $u\in\E'_\Gamma=\E'_\Gamma(\wedge^{kd}T^*\!\!\M^k\To\M^k)$ with $\WF(u)=\Gamma$ 
  (Theorem 8.1.4, pp. 255--256 of \cite{horm1}), the above set inclusion is proper for all 
  $m<m'$. For the last statement, we recall that, for any given non-void, closed conic subset 
  $\Gamma$ of the cotangent bundle minus the range of its zero section, the Hörmander topology 
  on $\E'_\Gamma$ is the initial topology induced by the linear maps $u\mapsto u(f)\in\CC$ and 
  $u\mapsto Pu\in\Gamma^\infty_c(\wedge^{kd}T^*\!\!\M^k\To\M^k)$, where $f$ runs through all smooth 
  functions and $P$ runs through all properly supported pseudodifferential operators of order 
  zero on the vector bundle $\wedge^{kd}T^*\!\!\M^k$ over $\M^k$ such that $\WF(P)\cap\Gamma=
  \varnothing$, where 
  \[
  \begin{split}
  \WF(P)&=\{(x_1,\ldots,x_k;\xi_1,\ldots,\xi_k)\in T^*\!\!\M^k\sm 0\ |\\
  &\phantom{=\{}(x_1,\ldots,x_k,x_1,\ldots,x_k;\xi_1,\ldots,\xi_k,-\xi_1,\ldots,-\xi_k)\in\WF(K_P)\}
  \end{split}
  \]
  denotes the \emph{microsupport} of $P$ (see e.g. Proposition 18.1.26 and formulas (18.1.34), 
  (18.1.35), pp. 88 as well as the remark following Theorem 18.1.28, pp. 89--90 of \cite{horm2}). 
  Here $K_P$ is the Schwartz kernel of $P$. For the convenience of the reader, we recall that 
  (a) $\WF(K_P)\subset N^*\!\!\Delta_2(\M^k)$ (see e.g. Theorem 18.1.16, pp. 80 of \cite{horm2}), 
  (b) $P$ being properly supported means that the restrictions to $\supp K_P$ of the canonical 
  projections onto the first $k$ and the last $k$ arguments are proper maps, which entails that 
  $Pu$ is compactly supported if $u$ is, and (c) $\WF(P)\cap\Gamma=\varnothing$ implies that 
  $Pu$ is smooth (see e.g. Theorem 8.2.13, pp. 268--269 of \cite{horm1}). The above 
  inductive limit topology on $\E'_{\Gamma_{k,m}}(\wedge^{kd}T^*\!\!\M^k\To\M^k)$ is strictly finer 
  than the topology induced from $\E'(\wedge^{kd}T^*\!\!\M^k\To\M^k)$. Since the latter is 
  Hausdorff, we conclude that the former is also Hausdorff. Moreover, since $\CC$ is 
  finite-dimensional and $\Gamma^\infty_c(\wedge^{kd}T^*\!\!\M^k\To\M^k)$ is nuclear, it follows 
  from the permanence of nuclearity for initial topologies (Proposition 5.2.3, pp. 92 of 
  \cite{pietsch}) that $\E'_\Gamma$ is nuclear as well. By Proposition 4.2.1, pp. 76 of 
  \cite{jarchow} together with Theorems 5.1.1, pp. 85 and 5.2.2, pp. 91--92 of \cite{pietsch}, 
  any Hausdorff countable inductive limit of nuclear locally convex spaces is nuclear. Therefore, 
  $\E'_{\Upsilon_{k,g}}(\wedge^{kd}T^*\!\!\M^k\To\M^k)$ must be nuclear for all $k$, as claimed.
\end{proof}

\begin{remark}[Stefan Waldmann, personal communication]\label{s4r1}
  We remark that the inclusion
  \[
  \E'_{\Gamma_{k,m}}(\wedge^{kd}T^*\!\!\M^k\To\M^k)\subset\E'_{\Gamma_{k,m+1}}(\wedge^{kd}T^*\!\!\M^k\To\M^k)
  \]
  although being a proper injection, is \emph{not} a topological embedding. The reason
  is the following: the space $\Gamma^\infty_c(\wedge^{kd}T^*\!\!\M^k\To\M^k)$ of test
  densities is dense in the Hörmander topology of $\E'_{\Gamma_{k,m}}(\wedge^{kd}T^*\!\!\M^k\To\M^k)$
  for all $m\in\NN$. Since by Theorem 8.1.4 of \cite{horm1} one can find $u_m\in\E'_{\Gamma_{k,m+1}}
  (\wedge^{kd}T^*\!\!\M^k\To\M^k)$ which does not belong to $\E'_{\Gamma_{k,m}}(\wedge^{kd}T^*
  \!\!\M^k\To\M^k)$, there is a sequence $(v_{n,m}\in\Gamma^\infty_c(\wedge^{kd}T^*\!\!\M^k\To
  \M^k))_{n\in\NN}$ converging to $u_m$ in the Hörmander topology of $\E'_{\Gamma_{k,m+1}}(\wedge^{kd}
  T^*\!\!\M^k\To\M^k)$. If $\E'_{\Gamma_{k,m}}(\wedge^{kd}T^*\!\!\M^k\To\M^k)$ were a closed subspace 
  in this topology, then $u_m$ would have to be an element of this subspace, which is false by 
  assumption. In fact, even more is true: since $\Gamma^\infty_c(\wedge^{kd}T^*\!\!\M^k\To\M^k)$
  is dense in the Hörmander topology of $\E'_{\Gamma_{k,m+1}}(\wedge^{kd}T^*\!\!\M^k\To\M^k)$, so
  is $\E'_{\Gamma_{k,m}}(\wedge^{kd}T^*\!\!\M^k\To\M^k)$. Hence, the inductive limit \eqref{s4e3}
  cannot be a \emph{strict} one. 
\end{remark}

From now on we tacitly assume that $\E'_{\Upsilon_{k,g}}(\wedge^{kd}T^*\!\!\M^k\To\M^k)$ is
endowed with the topology defined in Corollary \ref{s4c1} for all $k$. Once this is done,
we may proceed to proving the main result of this Subsection.

\begin{theorem}\label{s4t1}
  Let $\UU\subset\C^\infty(\M)$ be open in the compact-open topology, and let $\Fun((\M,g),\UU)$
  be endowed with the initial topology induced by the linear maps
  \begin{align}
    F &\mapsto F(\varphi)\in\CC\ ,\label{s4e4}\\
    F &\mapsto F^{(k)}[\varphi]\in\E'_{\Upsilon_{k,g}}(\wedge^{kd}T^*\!\!\M^k\To\M^k)\ ,\,k=1,2,\ldots\ ,\label{s4e5}
  \end{align}
  with $\varphi$ running through all elements of $\UU$. Then $\Fun((\M,g),\UU)$ is a nuclear
  locally convex vector space over the complex numbers which is also a Poisson *-algebra when
  endowed with the Peierls bracket associated to a microlocal generalized Lagrangian of
  first order on $\UU$ with normally hyperbolic Euler-Lagrange operator. As a consequence, 
  the Poisson *-subalgebra $\Fun_0(\M,\UU)$ and the self-adjoint linear subspace 
  $\Fun_{\mu\loc}(\M,\UU)$ are also nuclear locally convex subspaces when endowed with 
  the relative topology. Moreover, the involution of $\Fun((\M,g),\UU)$ is continuous.
\end{theorem}
\begin{proof}
  That $\Fun((\M,g),\UU)$ is a nuclear locally convex space follows from the permanence
  of nuclearity for initial topologies (Proposition 5.2.3, pp. 92 of \cite{pietsch}) together 
  with Corollary \ref{s4c1}. Involution is obviously well-defined, continuous and commutes 
  with the Peierls bracket; Proposition \ref{s3p4} and Corollary \ref{s3c3} show that the 
  Peierls bracket $\{\cdot,\cdot\}_\Li$ associated to $\Li$ is a Lie bracket on $\Fun((\M,g),\UU)$.
  It remains to check that $\Fun((\M,g),\UU)$ is closed under products -- Leibniz's rule for 
  $\{\cdot,\cdot\}_\Li$ 
  \begin{equation}\label{s4e6}
    \{F,GH\}_\Li=\{F,G\}_\Li H+G\{F,H\}_\Li
  \end{equation}
  will then follow from Leibniz's rule \eqref{a1e4} for functional derivatives of pointwise 
  products of functionals (see also the discussion right after Theorem \ref{s4t2}).
  
  Take $F,G\in\Fun((\M,g),\UU)$, and consider their pointwise product $(F\cdot G)(\varphi)=F(\varphi)
  G(\varphi)$. Now, by Leibniz's rule for derivatives of order $k$,
  \begin{align*}
    (F\cdot G)^{(k)}&[\varphi](\vec{\varphi}_1,\dots,\vec{\varphi}_k)\\
                    &=\sum_{\pi\in\Po_k}\sum_{l=0}^{k} F^{(k-l)}[\varphi](\vec{\varphi}_{\pi(1)},
                      \ldots,\vec{\varphi}_{\pi(k-l)})G^{(l)}[\varphi](\vec{\varphi}_{\pi(k-l+1)},
                      \ldots, \vec{\varphi}_{\pi(k)})\ ,
  \end{align*}
  where $\Po_k$ is the group of permutations of $k$ elements. By Theorem 8.2.9, pp. 267 of 
  \cite{horm1}, the wave front set of each distribution appearing in the above sum is given by 
  \begin{align*}
    \WF(F^{(k-l)}[\varphi]\otimes G^{(l)}[\varphi])
    &\subset\WF(F^{(k-l)}[\varphi])\times\WF(G^{(l)}[\varphi]) \\
    &\cup\left(\WF(F^{(k-l)}[\varphi])\times(\supp(G^{(l)}[\varphi])\times\{0\})\right)\\
    &\cup\left((\supp F^{(k-l)}[\varphi])\times\{0\})\times \WF(G^{(l)}[\varphi])\right)\ .
  \end{align*}
  By direct inspection the right-hand side is included in the open set $\Upsilon_{k,g}$, as it 
  should. The wave front set of $(F\cdot G)^{(k)}[\varphi]$ is clearly contained in the union 
  of all the wave front sets of the components, which satisfies again the requested bound by 
  the closedness of the wave front sets. 
  The remaining claims follow immediately from the permanence of nuclearity under taking linear 
  subspaces (Proposition 5.1.1, pp. 85 of \cite{pietsch}).
\end{proof}

\subsection{\label{s4-gen-cr}$\C^\infty$-ring structure and its consequences}

Actually, one can strengthen Theorem \ref{s4t1} considerably:

\begin{theorem}[Smooth functional calculus]\label{s4t2}
  Given $F_1,\ldots,F_n\in\Fun((\M,g),\UU)$, let $V\subset\RR^{2n}\cong\CC^n$ be an open 
  set containing the range of $(F_1,\ldots,F_n)$, and let $\psi:V\To\CC$ be a smooth map. Then
  $\psi\circ(F_1,\ldots,F_n)\in\Fun((\M,g),\UU)$ with $\supp(\psi\circ(F_1,\ldots,F_n))
  \subset\cup^n_{j=1}\supp F_j$.
\end{theorem}
\begin{proof}
  Smoothness of $\psi\circ(F_1,\ldots,F_n)$ follows from Faà di Bruno's formula \eqref{a1e7}. 
  The validity of the aforementioned support property follows from an argument similar to that
  used in the proof of Lemma \ref{s2l3} for products (i.e. $\psi(z_1,z_2)=z_1z_2$), so we are 
  only left with proving that the wave front set of the functional derivative of 
  $\psi\circ(F_1,\ldots,F_n)$ of order $k$ is contained in $\Upsilon_{k,g}$ for all $k\geq 1$. 
  This fact then follows from Faà di Bruno's formula \eqref{a1e7} together with an argument 
  similar to that used for products in Theorem \ref{s4t1}. 
\end{proof}

In particular, $\Fun((\M,g),\UU)$ is a $\C^\infty$\emph{-ring} \cite{moerr}. To our knowledge, 
this is the first non-trivial example in which such a structure appears in applications outside 
pure mathematics. Moreover, the Peierls bracket acts as a $\C^\infty$-derivation on $\Fun((\M,g),
\UU)$, that is, if $\psi,F_1,\ldots,F_n$ are as in Theorem \ref{s4t2} and $G\in\Fun((\M,g),\UU)$, 
then
\begin{equation}\label{s4e7}
  \{\psi(F_1,\ldots,F_n),G\}_\Li=\sum^n_{j=1}\left[\frac{\dd\psi}{\dd\Re z_j}
    (F_1,\ldots,F_n)\{\Re F_j,G\}_\Li+\frac{\dd\psi}{\dd\Im z_j}(F_1,\ldots,F_n)
    \{\Im F_j,G\}_\Li\right]\ .
\end{equation}
The above formula follows immediately from the chain rule \eqref{a1e3} and yields Leibniz's 
rule for the Peierls bracket as a special case.

\begin{remark}\label{s4r2}
  A consequence of Theorem \ref{s4t2} is that the topology of $\Fun((\M,g),\UU)$ given 
  in Corollary \ref{s4c1} is \emph{not} sequentially complete. To see this, let $F:\UU
  \ni\varphi\mapsto F(\varphi)\doteq\int_\M\varphi\omega$, where $\omega$ is a smooth 
  real-valued $d$-form of compact support in $\M$. Let $(f_n)$ be a sequence of even 
  smooth functions $f_n:\RR\To[0,1]$ supported in $[-2,2]$ which converges pointwise 
  to the characteristic function $\chi_{[-1,1]}$ of $[-1,1]$ and whose derivatives of 
  all orders converge pointwise to zero (e.g. take $f_n(|x|)=1$ for $|x|\leq 1+(4n)^{-1}$ 
  and $f_n(|x|)=0$ for $|x|\geq 1+(2n)^{-1}$). Defining $F_n\doteq f_n\circ F$ gives a
  sequence $(F_n)$ of elements of $\Fun((\M,g),\UU)$, whose functional derivatives of 
  order $k\geq 1$ are given by Faà di Bruno's formula \eqref{a1e7} as
  \begin{equation}\label{s4e8}
    F^{(k)}_n[\varphi](\vec{\varphi}_1,\ldots,\vec{\varphi}_k)=f^{(k)}_n(F(\varphi))
    \left(\int_\M\vec{\varphi}_1\omega\right)\cdots\left(\int_\M\vec{\varphi}_k\omega\right)\ .
  \end{equation}
  Hence, the functional derivatives of all orders of the elements of the sequence $(F_n)$
  converge to zero in the respective topologies for all $\varphi\in\UU$. The sequence
  $(F_n(\varphi))$, however, converges pointwise to $\chi_{[-1,1]}\circ F(\varphi)$, which
  defines a functional on $\UU$ which is in general \emph{not even continuous}, let alone 
  microcausal. By the Stone-Weierstrass theorem in the interval $[-2,2]$ \cite{jarchow}, 
  there is even a sequence of functionals $F_n\in\Fun((\M,g),\UU\cap F^{-1}((-2,2)))$ 
  which lies in the *-subalgebra of $\Fun((\M,g),\UU\cap F^{-1}((-2,2)))$ generated by 
  $\Fun_{\mu\loc}(\M,\UU\cap F^{-1}((-2,2)))$ and converges to $\chi_{[-1,1]}\circ F$ in the 
  topology of $\Fun((\M,g),\UU\cap F^{-1}((-2,2)))$. 
\end{remark}

\begin{remark}\label{s4r3}
  In view of the counterexample discussed in Remark \ref{s4r2}, it would be desirable to find 
  a stronger topology on $\Fun((\M,g),\UU)$ which is compatible with its Poisson *-algebraic 
  and $\C^\infty$-ring structures, and (at least sequentially) complete. It is clear from this 
  counterexample that even if we follow the proposal of \cite{dabrowskib} and replace the (weak) 
  seminorms $|F^{(k)}[\varphi](\vec{\varphi}_1,\ldots,\vec{\varphi}_k)|$ by the \emph{strong} 
  seminorms 
  \[
  F\mapsto\sup\left\{\big|F^{(k)}[\varphi](\vec{\varphi})\big|\ \Big|
    \ \vec{\varphi}\in\Bo\right\}\ ,
  \]  
  where $\varphi$ runs over $\UU$ and $\Bo$ runs over all closed and bounded subsets of 
  $\C^\infty(\M^k)$, we still get sequential incompleteness since the functional derivatives 
  of each element in the sequence we have constructed have \emph{empty} wave front sets and 
  therefore weak convergence of the derivatives entails their strong convergence by the 
  Banach-Steinhaus theorem. We point that the phenomenon described in Remark \ref{s4r2} 
  is \emph{independent} of the actual failures of sequential completeness for 
  $\E'_{\Upsilon_{k,g}}(\wedge^{kd}T^*\!\!\M^k\To\M^k)$ in either the weak or strong topologies, 
  which were shown in \cite{dabrowskib}. Therefore, this phenomenon \emph{cannot} be 
  circumvented by allowing the functional derivatives of microcausal functionals to take 
  values in the completions of $\E'_{\Upsilon_{k,g}}(\wedge^{kd}T^*\!\!\M^k\To\M^k)$ in the 
  strong topology for each $k$, as advocated in \cite{broudlr,dabrowski1,dabrowski2}. 

  A seemingly better way out is to take full advantage of the Michal-Bastiani notion of 
  differentiability and replace the seminorms $|F^{(k)}[\varphi](\vec{\varphi}_1,\ldots,
  \vec{\varphi}_k)|$ by the even stronger seminorms 
  \[
  F\mapsto\sup\left\{\big|F^{(k)}[\varphi](\vec{\varphi})\big|\ \Big|\ \vec{\varphi}\in\Bo\right\}\ ,
  \] 
  where $\Bo$ is as above and $\Ko$ runs over the compact subsets of $\UU$. In other words,
  we require now uniform convergence of functional derivatives of all orders in compact subsets
  (this topology is called \emph{Bastiani topology} in \cite{broudlr}). Since $\C^\infty(\M)$ is 
  semi-Montel, we have that $\Ko\times\Bo$ is compact. Therefore, since $\UU\times\C^\infty(\M^k)$ 
  is metrizable and hence compactly generated\footnote{\label{s4f1} Recall that a completely 
    regular topological space $X$ is said to be \emph{compactly generated} or a 
    $k$\emph{-space} if the topology of $X$ coincides with the final topology induced 
    by the inclusions of compact subsets of $X$. This is equivalent to the space of continuous 
    real-valued functions on $X$ being complete with respect to the topology of uniform convergence 
    on compact subsets of $X$ (see e.g. Theorem 3.6.4, pp. 70 of \cite{jarchow}).}, we conclude 
  that the completion of $\Fun((\M,g),\UU)$ in this stronger topology does correspond to allowing 
  the functional derivatives of each order $k\in\NN$ to take values in the completion of 
  $\E'_{\Upsilon_{k,g}}(\wedge^{kd}T^*\!\!\M^k\To\M^k)$ in the strong topology (see e.g. Proposition 
  16.6.2, pp. 361 of \cite{jarchow}). Thanks to the results of \cite{broudh2}, one then has 
  separate continuity of the pointwise product and the Peierls bracket on $\Fun((\M,g),\UU)$ 
  with respect to this topology and thus these bilinear operations extend (separately) continuously 
  to the completion (see also \cite{broudlr,dabrowski2} for related results). By Faà di Bruno's 
  formula \eqref{a1e7}, separate continuity also holds for the $\C^\infty$-ring operations. It is 
  not clear, however, whether \emph{nuclearity} survives in this stronger topology. For instance, 
  as Meise has shown \cite{meise}, the space of MB-smooth functionals on $\UU$ endowed 
  with the topology of uniform convergence of functional derivatives on compact subsets of $\UU$ 
  \emph{cannot} be nuclear despite being complete.

  Fortunately, there is a middle course able to get the best of both worlds, thanks to the 
  coincidence of MB smoothness and convenient smoothness in Fréchet spaces (see Remark \ref{a1r2} 
  below). We start from the simple but important observation (see e.g. Lemma 3.11, pp. 30 of 
  \cite{km}) that the space $\C^\infty(\UU,\CC)$ of (conveniently) smooth maps from $\UU\subset
  \C^\infty(\M)$ open to $\CC$ is the projective limit 
  \[
  \begin{split}
    \C^\infty(\UU,\CC) &=\ilim{\gamma\in\C^\infty(\RR,\UU)}\C^\infty(\RR,\CC)\\ &=\Bigg\{
      (F_\gamma)_{\gamma}\in\prod_{\gamma\in\C^\infty(\UU,\CC)}\C^\infty(\RR,\CC)\ \Bigg|\ F_\gamma
      \circ\kappa=F_{\gamma\circ\kappa}\\ &\phantom{=\Bigg\{}\text{ for all }\kappa\in
    \C^\infty(\RR,\RR)\Bigg\}
  \end{split}  
  \]
  along the preordered set $(\C^\infty(\RR,\UU),\preccurlyeq)$ with preorder $\preccurlyeq$ given 
  by smooth reparametrization:
  \[
  \gamma\preccurlyeq\tilde{\gamma}\;\Iff\;\gamma=\tilde{\gamma}\circ\kappa\text{ for some }
  \kappa\in\C^\infty(\RR,\RR)\ .
  \]
  To see the second identity, notice that any $(F_\gamma)_{\gamma}\in\prod_{\gamma\in\C^\infty(\UU,\CC)}
  \C^\infty(\RR,\CC)$ such that $F_\gamma\circ\kappa=F_{\gamma\circ\kappa}$ for all $\kappa\in\C^\infty
  (\RR,\RR)$ defines a map $\UU\ni\varphi\mapsto F(\varphi)=F_{\varphi}$, where we identify 
  $\varphi$ with the constant curve $\RR\ni t\mapsto\varphi(t)\equiv\varphi\in\UU$. One 
  immediately sees that $F_\gamma(t_0)=F(\gamma(t_0))$ for all $\gamma\in\C^\infty(\RR,\UU)$, 
  $t_0\in\RR$ by means of the constant reparametrization $\kappa(t)\equiv t_0$. Conversely, any 
  $F\in\C^\infty(\UU,\RR)$ gives rise to such an $(F_\gamma)_\gamma$ by setting $F_\gamma=\gamma^*F=
  F\circ\gamma$ for all $\gamma\in\C^\infty(\RR,\UU)$ -- one then obviously has $F_\gamma\circ
  \kappa=F\circ\gamma\circ\kappa=F_{\gamma\circ\kappa}$ for all $\kappa\in\C^\infty(\RR,\RR)$. As 
  such, it is natural to impose on $\C^\infty(\UU,\CC)$ the initial topology induced from the 
  compact-open topology of $\C^\infty(\RR,\CC)$ through the pullbacks $\gamma^*$ by all $\gamma
  \in\C^\infty(\RR,\UU)$ as in Definition 3.11, pp. 30 of \cite{km}, which is just the induced 
  subspace topology from the direct product $\prod_{\gamma\in\C^\infty(\UU,\CC)}\C^\infty(\RR,\CC)$. 
  Since $\C^\infty(\UU,\CC)$ is a \emph{closed} subspace of the latter and the compact-open 
  topology of $\C^\infty(\RR,\CC)$ is nuclear and complete, it follows from the permanence of 
  nuclearity for initial topologies (Proposition 5.2.3, pp. 92 of \cite{pietsch}) and the 
  permanence of completeness for closed subspaces and products (respectively Propositions 
  3.2.5 and 3.2.6, pp. 59 of \cite{jarchow}) that this topology on $\C^\infty(\UU,\CC)$ is also 
  nuclear and complete, as desired. Likewise, since convenient smoothness and MB smoothness 
  coincide on $\UU$, the topology induced on the (closed) subspace $\Fun_{00}(\M,\UU)\cap
  \C^\infty(\UU,\CC)$ is nuclear (due to the permanence of nuclearity for linear subspaces, 
  see Proposition 5.1.1, pp. 85 of \cite{pietsch}), complete and finer than the topology of 
  pointwise convergence of all derivatives. To see the latter, notice that this topology is 
  induced by the so-called \emph{(strong) convenient seminorms}
  \[
  F\mapsto\sup\left\{\big|F^{(k)}[\gamma(t)](\vec{\varphi})\big|\ \Big|\ t\in[a,b],\,a<b
    \in\RR,\,\gamma\in\C^\infty(\RR,\UU),\,\vec{\varphi}\in\Bo\right\}\ ,
  \]
  with $\Bo$ as before. In other words, we consider only the ``at most one-dimensional'' 
  compact subsets $\Ko=\gamma([a,b])\subset\UU$, $a<b\in\RR$, $\gamma\in\C^\infty(\RR,\UU)$, 
  which of course include all singleton subsets of $\UU$ through all constant curves into 
  $\UU$. Substituting the convenient seminorms for $|F^{(k)}[\varphi](\vec{\varphi}_1,\ldots,
  \vec{\varphi}_k)|$ in $\Fun((\M,g),\UU)$ then yields a nuclear locally convex topology in
  the latter, which we suggestively call the \emph{(strong) convenient topology} and whose 
  completion amounts once more to allowing $F^{(k)}[\varphi]$ to take values in the 
  completion of $\E'_{\Upsilon_{k,g}}(\wedge^{kd}T^*\!\!\M^k\To\M^k)$ in the strong topology for 
  all $\varphi\in\UU$, $k\in\NN$. Unlike before, thanks to the permanence of nuclearity for 
  completions (Proposition 5.3.1, pp. 93 of \cite{pietsch}) we can be sure that the completion 
  of $\Fun((\M,g),\UU)$ in the convenient topology is also nuclear. A similar proposal has 
  been put forward in \cite{broudlr,dabrowski2} by including all smooth maps with finite 
  (but otherwise arbitrary) dimensional domains and $\UU$ as codomain in addition to just 
  smooth curves into $\UU$. The aforementioned continuity of the Poisson *-algebraic 
  operations of $\Fun((\M,g),\UU)$ still survives, of course.
\end{remark}

Finally, we recall the important fact that $\Fun_{\mu\loc}(\M,\UU)\subset\Fun((\M,g),\UU)$
contains the squared Sobolev seminorms 
\[
\varphi\mapsto F_{k,f}(\varphi)=\|\varphi\|^2_{2,k,f}\ ,
\]
defined in \eqref{s2e5}, for all $f\in\C^\infty_c(\M)$. This together with Theorem \ref{s4t2} 
yields:

\begin{proposition}\label{s4p1}
  The following facts hold true:
  \begin{enumerate}
  \item[(i)] Given any open set $\UU\subset\C^\infty(\M)$ in the compact-open topology 
    and $\varphi_0\in\UU$, there is $F\in\Fun((\M,g),\C^\infty(\M))$ such that $F(\varphi_0)=1$, 
    $0\leq F\leq 1$, and $F\equiv 0$ in $\C^\infty(\M)\sm\UU$. In particular, one can
    completely recover the compact-open topology of $\C^\infty(\M)$ from the complements
    of zero sets of elements of $\Fun((\M,g),\C^\infty(\M))$.
  \item[(ii)] Any $\UU\subset\C^\infty(\M)$ open in the compact-open topology admits 
    locally finite partitions of unity whose elements belong to $\Fun((\M,g),\UU)$.
  \item[(iii)] Given any open set $\UU\subset\C^\infty(\M)$ in the compact-open topology, 
    the algebra $\Fun((\M,g),\UU)$ separates the points of $\UU$, that is, for any $\varphi_1,
    \varphi_2\in\UU$ there is an $F\in\Fun((\M,g),\UU)$ such that $F(\varphi_1)\neq F(\varphi_2)$.
  \item[(iv)] Given any open set $\UU\subset\C^\infty(\M)$ in the compact-open topology, 
    any unital *-morphism $\omega:\Fun((\M,g),\UU)\To\CC$ (i.e. a \emph{*-character} on 
    $\Fun((\M,g),\UU)$) is given by the evaluation functional at some $\varphi\in\UU$ 
    (by (iii), $\varphi$ must be unique).
  \item[(v)] Given any open sets $\UU,\VV\subset\C^\infty(\M)$ in the compact-open topology, 
    any continuous unital *-morphism $\alpha:\Fun((\M,g),\UU)\To\Fun((\M,g),\VV)$ is 
    the pullback of a unique smooth map $\alpha^*:\VV\To\UU$.
  \end{enumerate}
\end{proposition}
\begin{proof}
  \begin{enumerate}
  \item[(i)] Let $\chi:\RR\To[0,1]$ be an even smooth function such that $\chi(t)=0$
    for $|t|\geq 1$ and $\chi(t)=1$ $|t|\leq\frac{1}{2}$. There are $k\in\NN$, 
    $f\in\C^\infty_c(\M)$ and $R>0$ such that $\varphi_0\in\{\varphi\in\C^\infty(\M)\ |\ F_{k,f}
    (\varphi-\varphi_0)<R^2\}\subset\UU$. Set $F(\varphi)=\chi(R^{-2}F_{k,f}(\varphi-\varphi_0))$, 
    and we are done by Theorem \ref{s4t2}.
  \item[(ii)] Recall that, since $\C^\infty(\M)$ is a nuclear Fréchet space, it follows that
    $\C^\infty(\M)$ is separable, hence second countable and Lindelöf. $\UU$ is then a second 
    countable metric space, hence also separable and Lindelöf. Since (i) holds, the result
    then follows from Theorem 16.10, pp. 171--172 of \cite{km}.
  \item[(iii)] $\UU$ is Hausdorff, hence the result follows immediately from (i).
  \item[(iv)] By the proof of (ii), we know that $\UU$ is Lindelöf. Since (i) implies that
    $\UU$ is completely regular, it follows that it must be \emph{realcompact}, that is, any 
    $\RR$-algebra homomorphism from the $\RR$-valued continuous functions on $\UU$ into $\RR$ 
    is given by evaluation at some $\varphi\in\UU$ (see \cite{engelking}, Theorem 3.11.12, pp. 
    216). The result then follows for the $\RR$-subalgebra of real-valued elements of 
    $\Fun((\M,g),\UU)$ by combining (ii) with Theorem 17.6, pp. 187--188, Remark 18.1, 
    pp. 188--189 and Proposition 18.3, pp. 191 of \cite{km}. The general case is immediate.
  \item[(v)] Notice that the pullback of any *-character by $\alpha$ is also a *-character, 
    hence by (iii)-(iv) $\alpha^*$ as above is really the pullback by $\alpha$ (thus also 
    justifying our notation). Moreover, the action of $\alpha$ on functionals of the form 
    $\UU\ni\varphi\mapsto\int_\M f\varphi\ud\mu_g$, $f\in\C^\infty_c(\M)$, shows that $\alpha^*$ 
    must be smooth on $\VV$.
  \end{enumerate}
\end{proof}

Some comments about the meaning of Proposition \ref{s4p1} are in order. Proposition \ref{s4p1} 
(ii) shows that we can ``glue together'' microcausal functionals defined on an open covering 
of $\C^\infty(\M)$, that is, the assignment 
\begin{equation}\label{s4e9}
  \UU\subset\C^\infty(\M)\text{ open }\To\Fun((\M,g),\UU)\ ,
\end{equation}
together with the restriction morphisms induced by inclusions between pairs of open subsets 
in $\C^\infty(\M)$ in the compact-open topology, constitute a \emph{sheaf} of *-algebras over 
the topological space $\C^\infty(\M)$. However, multiplying $F\in\Fun((\M,g),\UU)$ by a 
``bump'' functional as given by Proposition \ref{s4p1} (i) improves the localization of $F$ 
in field configuration space at the cost of \emph{losing} information about the 
\emph{space-time} support of $F$. This must be kept in mind when multiplying $F$ by the 
elements of a partition of unity on $\UU$ belonging to $\Fun((\M,g),\UU)$. A more conceptual 
discussion of the interplay between these two notions of localization will take place in 
future work.

\subsection{\label{s4-gen-em}On-shell ideals}

\begin{definition}\label{s4d1}
  Let $\UU,\Li$ be as in Proposition \ref{s3p2}. We define the \emph{on-shell ideal} of
  $\Fun((\M,g),\UU)$ associated to $\Li$ as the subspace $\JJ_\Li((\M,g),\UU)\subset\Fun
  ((\M,g),\UU)$ of all microcausal functionals $F$ of the form
  \begin{equation}\label{s4e10}
    F(\varphi)=X[\varphi]E(\Li)[\varphi]\ ,\quad\varphi\in\UU\ ,
  \end{equation}
  where $X:\UU\times\Gamma^\infty(\wedge^{kd}T^*\!\!\M\To\M)\ni(\varphi,\omega)\mapsto
  X[\varphi]\omega\in\CC$ is jointly smooth and linear with respect to $\omega$. 
\end{definition}

It is clear that $F(\varphi)=0$ for all $F\in\JJ_\Li((\M,g),\UU)$ and all $\varphi\in\UU$
such that $E(\Li)[\varphi]=0$. A key consequence of \eqref{s4e10} is the following

\begin{proposition}\label{s4p2}
  $\JJ_\Li((\M,g),\UU)$ is a Poisson *-ideal of $\Fun((\M,g),\UU)$.
\end{proposition}
\begin{proof}
  It is clear that $\JJ_\Li((\M,g),\UU)$ is a *-ideal of $\Fun((\M,g),\UU)$, so what
  is left is to show that $\JJ_\Li((\M,g),\UU)$ is also a Lie ideal of $\Fun((\M,g),\UU)$.
  Let $G\in\JJ_\Li((\M,g),\UU)$, so that $G(\varphi)=X[\varphi]E(\Li)[\varphi]$ with 
  $X$ as in Definition \ref{s4d1}. By the chain rule \eqref{a1e3} applied to the pair 
  of maps $X,(\id,E(\Li))$, we get that
  \begin{equation}\label{s4e11}
    G^{(1)}[\varphi](\vec{\varphi})=DX[\varphi](\vec{\varphi})E(\Li)[\varphi]
    +X[\varphi]E'(\Li)[\varphi]\vec{\varphi}\ ,\quad\vec{\varphi}\in\C^\infty(\M)\ ,
  \end{equation}
  where $DX$ is defined as in \eqref{a1e10}. Let now $F\in\Fun((\M,g),\UU)$. Then
  \[
  \begin{split}
    \{F,G\}_\Li(\varphi) &=\Spr{F^{(1)}[\varphi],\Delta_\Li[\varphi]G^{(1)}[\varphi]}\\
    &=\Spr{F^{(1)}[\varphi],DX[\varphi](\vec{\varphi})E(\Li)[\varphi]}-X[\varphi]E'(\Li)[\varphi]
    \Delta_\Li[\varphi]F^{(1)}[\varphi]\\ &=\Spr{F^{(1)}[\varphi],DX[\varphi](\vec{\varphi})
      E(\Li)[\varphi]}
  \end{split}
  \]
  and therefore $\{F,G\}_\Li\in\JJ_\Li((\M,g),\UU)$, as desired.
\end{proof}

One is then led to the 

\begin{definition}\label{s4d2}
Let $\UU,\Li$ as in Proposition \ref{s3p2}. The quotient Poisson*-algebra 
\begin{equation}\label{s4e12}
  \Fun_\Li((\M,g),\UU)\doteq\Fun((\M,g),\UU)/\JJ_\Li((\M,g),\UU)
\end{equation}
is called the \emph{on-shell algebra} over $\UU$ associated to $\Li$. 
\end{definition}

As stated in the introduction, the on-shell algebra correspond to our algebra of observables
once we have imposed the equations of motion $E(\Li)[\varphi]=0$ on field configurations in
$\UU$. A natural question at this point is whether any $F\in\Fun((\M,g),\UU)$ vanishing on 
solutions $\varphi\in\UU$ of $E(\Li)[\varphi]=0$ is of the form \eqref{s4e10}. This question 
shall be addressed in future work.

\section{\label{s5-ciao}Final considerations}

We have presented the very first steps into a novel, algebraic approach 
to classical field theory in which the main role is played by algebras 
of functionals over sets of field configurations on any globally hyperbolic 
space-time. 

As a whole, our formalism can be extended to field theories living on
any fiber bundle over space-time. In fact, extensions of parts of our framework
have already appeared in the literature, including fermion fields \cite{rejzner},
Yang-Mills models and gravity \cite{frerej,weise}. These works also show that our
formalism is capable of dealing with Lagrangians possessing local symmetries which
constrain the dynamics -- more precisely, a rigorous version of the classical
Batalin-Vilkoviski\u{\i} approach to gauge theories can be provided within our setup 
\cite{frerej}. Such subtleties are absent in the case of real scalar fields, which 
do not possess any ``internal'' structure. A full account of our framework 
encompassing all the above examples will be pursued in the future.

On a more technical side, treating the above examples will occasionally require 
(particularly in the case of fermion fields) extending the results concerning 
normally hyperbolic linear partial differential operators presented in this
series of papers to more general hyperbolic systems. Theorem \ref{s3t1} can be 
extended to symmetrizable, first-order hyperbolic systems with very few changes 
in the arguments. Arguably, Theorem \ref{s3t2} could be reworked along the lines 
of the paper of Dencker \cite{dencker} to encompass symmetrizable, first-order 
hyperbolic systems of real principal type, of which the Dirac operator is an 
example \cite{rejzner}. One could try to go even further and encompass the case 
of second-order regularly hyperbolic systems of Christodoulou \cite{christo}, 
but the microlocal analysis of such systems is severely underdeveloped, due to 
the possibility of occurrence of bicharacteristics with varying multiplicity 
(e.g. birefringence in crystal optics; see \cite{liess} for the state of the 
art on these matters). 

%

In this paper we have restricted ourselves to studying \emph{linearized}
dynamics. This, of course, is far from being the full story -- the analysis 
of full \emph{nonlinear} dynamics within our approach, to be undertaken in a 
followup publication \cite{bfr2}, will be based on a \emph{semi-global solvability} 
result for second-order, quasi-linear hyperbolic partial differential operators 
$P:\C^\infty(\M)\To\C^\infty(\M)$. More precisely, for a suitably large family 
of compact regions $K$ of the space-time manifold $\M$, that the equation
\[
P(\varphi_0+\varphi)=P(\varphi_0)+f
\]
has a smooth solution $\varphi$ in $K$ for any $\varphi_0,f\in\C^\infty(K)$, 
$f$ sufficiently small. Moreover, if we prescribe the Cauchy data for $\varphi$
on a suitable Cauchy hypersurface crossing $K$, this solution must be unique.
Such a result can be proved by combining a simple refinement (due to Klainerman 
\cite{klai1,klai2}, see also Hintz and Vasy \cite{hinv}) of classical energy 
estimates for second-order linear hyperbolic partial differential operators with 
a variant of the Nash-Moser-Hörmander inverse function theorem \cite{hamilton}, 
pretty much in the spirit of the results by Bryant, Griffiths and Yang \cite{bgy} 
and Tso \cite{tso}. Taking $f=P_0(\varphi_0)-P(\varphi_0)$, where $P$ is a ``small'' 
perturbation of $P_0$, yields that setting $m_{P,P_0}(\varphi_0)\doteq\varphi_0
+\varphi$ with $\varphi$ as above leads to the formula
\[
P\circ m_{P,P_0}=P_0\ .
\]
A map $m_{P,P_0}$ intertwining $P$ and $P_0$ in the above sense is called a 
\emph{M\o{}ller map}, in analogy with the M\o{}ller wave operators in quantum
mechanical scattering theory. M\o{}ller maps in classical field theory were
discussed formally in \cite{brudf,brufre2,dutfre1,dutfre2} and will constitute 
the backbone of our take on nonlinear dynamics -- in particular, since they act 
as Poisson maps with respect to the Peierls brackets associated to two 
Euler-Lagrange operators differing by a perturbation, they can be used to 
locally linearize a Peierls bracket around a given field configuration, pretty 
much like the Darboux-Weinstein theorem for regular, finite-dimensional Poisson 
manifolds \cite{vaisman}. We also hope that finer details of on-shell ideals
might be elucidated with such methods.

%

The final release of the present paper was delayed because of incomplete proofs 
of Proposition \ref{s3p4} and Corollary \ref{s3c3} in previous versions. We hope 
that we have now clarified the validity of those statements. In the meantime, 
several papers appeared dealing with other side aspects of the present paper, 
namely \cite{broudlr,dabrowski1,dabrowski2,dabrowskib}. Some of these aspects
were addressed in Remark \ref{s4r3}.

\section*{Acknowledgements}

We would like to thank Prof. Frank Michael Forger for a critical reading of 
an early version of the Introduction, as well as for invaluable advice on the 
mathematical literature on classical field theory and general presentation 
details. We are specially grateful to him for discussions on physically relevant 
functionals, which led to most of the examples presented in Subsection \ref{s2-kin-obs}.
We would also like to thank Prof. Stefan Waldmann for pointing out a mistake in
the proof of Corollary \ref{s4c1} in a previous version of the present paper, and
also for his clarifying comments. Finally, we are much grateful to Prof. Christian 
Brouder for his several inquires about our work, particularly for pointing out a
substantial gap in the previous proofs of of Proposition \ref{s3p4}
and Corollary \ref{s3c3}, and for numerous discussions, as well as to Prof. Peter
Michor for the enlightening observations on MathOverflow which led us to the crucial 
Lemma \ref{s2l6}, and the anonymous referees for the valuable comments.

The junior author (P.L.R.) would like to thank the hospitality of the II. Institut für
theoretische Physik, Universität Hamburg, the Dipartimento di Matematica, Facoltà di 
Scienze della Università di Trento, and the Instituto de Matemática e Estatística, 
University of São Paulo, where most of the work presented in this paper was developed,
and also the support from the Research Training Group 1670 -- ``Mathematics Inspired by
String Theory and Quantum Field Theory'', Universität Hamburg as well as from the Centro 
Italiano di Ricerca Matematica (CIRM) and the Bruno Kessler Foundation in the final 
stages of the writing.

\begin{appendix}

  %
  
  \renewcommand\thetheorem{\thesection.\arabic{theorem}}
  \makeatletter
  \@addtoreset{theorem}{section}
  \makeatother
  \setcounter{theorem}{0}
  
  
  \renewcommand\theequation{\thesection.\arabic{equation}}
  \makeatletter
  \@addtoreset{equation}{section}
  \makeatother
  \setcounter{equation}{0}
  
  %
  
  \section{\label{a1-calc}A short review of differential calculus on locally convex topological vector spaces} 
  
  In this Appendix we list the basic definitions and results of differential 
  calculus we need. Our basic references are \cite{hamilton} and \cite{km}, to 
  whom we refer for more details and proofs. The first reference works only with 
  Fréchet spaces, but the proofs of the results quoted below work in the general 
  case with little or no change.
  
  The notion of differentiability of curves in locally convex topological vector 
  spaces is straightforward.
  
  \begin{definition}\label{a1d1}
    Let $\gamma:(a,b)\To\Fun$, $a<b\in\RR\cup\{\pm\infty\}$ be a continuous curve 
    into a locally convex topological vector space $\Fun$. We say that $\gamma$ is 
    a $\C^1$ \emph{curve} if for all $t\in(a,b)$ the limit 
    \[
    \gamma'(t)\doteq\lim_{s\To 0}\frac{1}{s}(\gamma(t+s)-\gamma(t))
    \]
    exists and defines a \emph{continuous} curve $\gamma':(a,b)\To\Fun$ (continuity
    of $\gamma$ actually follows from these conditions alone, hence it does not hurt 
    to assume it from the start). We also say that $\gamma$ is a $\C^m$ curve, $m\geq 1$, 
    if $\gamma^{(k)}\doteq(\gamma^{(k-1)})'$ exists and is continuous for all $1\leq k\leq m$, 
    where $\gamma^{(0)}\doteq\gamma$. If $\gamma$ is a $\C^m$ curve for all $m$, we 
    say that $\gamma$ is a \emph{smooth} curve.
  \end{definition}
  
  We stress that there would be no loss of generality if we required the domain 
  of smooth curves to be the whole real line: by the chain rule \eqref{a1e3},
  $\gamma:(a,b)\To\Fun$ is smooth if and only if $\gamma\circ f:\RR\To\Fun$ is 
  smooth for any diffeomorphism $f:\RR\To(a,b)$ (e.g. $f(\lambda)=\frac{b+a}{2}
  +\frac{b-a}{2}\tanh(\lambda)$). Once this is said, let us see how Definition 
  \ref{a1d1} is realized in the concrete cases that interest us.
  
  \begin{itemize}
  \item $\Fun=\C^\infty(\M)$ (endowed with the compact-open topology): 
    $\gamma:\RR\To\Fun$ is smooth if and only if $\gamma(\lambda)(p)=\Phi(\lambda,p)$ 
    for all $(\lambda,p)\in\RR\times\M$, where $\Phi\in\C^\infty(\RR\times\M)$;
  \item $\Fun=\C^\infty_c(\M)$ (endowed with the usual inductive limit topology): 
    $\gamma:\RR\To\Fun$ is smooth if and only if $\gamma(\lambda)(p)=\Phi(\lambda,p)$ 
    for all $(\lambda,p)\in\RR\times\M$, where $\Phi\in\C^\infty(\RR\times\M)$ is
    such that for any $a<b\in\RR$ there is a compact subset $K\subset\M$ such 
    that $\Phi(\lambda,p)=\Phi(a,p)$ for all $p\not\in K$, $\lambda\in[a,b]$.
  \end{itemize}
  
  The notion of smooth curves allows one to introduce another topology on $\Fun$, 
  given by the final topology induced by $\RR$ through all smooth curves $\gamma:
  \RR\To\Fun$. We call this topology the $c^\infty$-topology on $\Fun$. This topology 
  is necessarily finer than the original one, but it is not in general a vector 
  space topology -- the finest locally convex vector space topology on $\Fun$ that 
  is coarser then the $c^\infty$-topology is the bornologification of $\Fun$'s original 
  topology. The $c^\infty$- and the original locally convex vector space topologies
  coincide if $\Fun$ is e.g. metrizable (such as $\C^\infty(\M)$), but are distinct for 
  $\Fun=\C^\infty_c(\M)$ if $\M$ is non-compact since then the $c^\infty$-topology 
  is not a vector space topology (see e.g. Proposition 4.26 (ii), pp. 45 of 
  \cite{km}).\footnote{\label{a1f1} Nonetheless, in this case the $c^\infty$-topology 
    coincides with the so-called \emph{Kelleyfication} of $\Fun$, which is the final 
    topology induced by all compact subsets of $\Fun$ through their respective inclusions 
    (see e.g. Theorem 4.11 (3), pp. 39--40 of \cite{km}). It is clear that the Kelleyfication 
    of $\Fun$ coinciding with the original topology of $\Fun$ amounts to $\Fun$ being compactly 
    generated (see footnote \ref{s4f1} above). This happens if e.g. $\Fun$ is metrizable.}
  
  Given two locally convex vector spaces $\Fun_1$, $\Fun_2$, $\UU\subset\Fun_1$ 
  $c^\infty$-open, we say that a map $\Phi:\UU\To\Fun_2$ is \emph{conveniently smooth} 
  if $\Phi\circ\gamma$ is a smooth curve on $\Fun_2$ for every smooth curve $\gamma:
  \RR\To\UU$. We stress that conveniently smooth maps need not even be continuous (see 
  \cite{glockner} for a counterexample). A simple non-trivial example of a conveniently
  smooth map $\Phi:\Fun\To\Fun$ is, of course, the translation $\varphi\mapsto\Phi
  (\varphi)=\varphi+\varphi_0$ by a fixed element $\varphi_0\in\Fun$. In particular,
  the coordinate change maps $\kappa_{\varphi_2}\circ\kappa_{\varphi_1}^{-1}:\C^\infty_c(\M)
  \To\C^\infty_c(\M)$ in the affine flat manifold $\C^\infty(\M)$ (endowed with the 
  Whitney topology) are conveniently smooth for all $\varphi_1,\varphi_2\in\C^\infty(\M)$ 
  such that $\varphi_1-\varphi_2\in\C^\infty_c(\M)$. This shows that the atlas 
  $\Un$ defined in \eqref{s2e7} induces a smooth structure on $\C^\infty(\M)$; 
  the corresponding smooth manifold topology is, of course, the manifold topology 
  generated by the $c^\infty$-open subsets of the modelling vector space $\C^\infty_c(\M)$, 
  which is even finer than the Whitney topology. The connected components of this topology 
  are, however, also of the form $\C^\infty_c(\M)+\varphi_0$, $\varphi_0\in\C^\infty(\M)$; 
  therefore, the smooth curves in $\C^\infty(\M)$ with respect to the smooth structure 
  induced by the atlas $\Un$ must be of the form $\RR\ni\lambda\mapsto\gamma
  (\lambda)=\varphi_0+\gamma_0(\lambda)$, where $\gamma_0:\RR\To\C^\infty_c(\M)$ is 
  smooth. Hence, it is just fair to say that such $\gamma$ is \emph{a smooth curve 
    with respect to the Whitney topology}, and the smooth structure induced
  by the atlas $\Un$, \emph{the smooth structure on} $\C^\infty(\M)$ 
  \emph{induced by the Whitney topology}.
  
  \begin{remark}\label{a1r1}
    It can be shown \cite{km} that, for $\C^\infty(\M)$ endowed with the smooth 
    structure induced by the Whitney topology, the bundles
    \[
    T^{r,s}\C^\infty(\M)=\left(\otimes^sT^*\C^\infty(\M)\right)\otimes
    \left(\otimes^rT\C^\infty(\M)\right)
    \]
    of tensors of contravariant rank $r$ and covariant rank $s$ are given 
    at each $\varphi\in\C^\infty(\M)$ by the space of bounded linear mappings 
    from $\otimes^s_\beta\C^\infty_c(\M)$ to $\otimes^r_\beta\C^\infty_c(\M)$. 
    Here $\otimes_\beta$ denotes the \emph{bornological} tensor product, whose 
    topology is the finest locally convex topology on the algebraic tensor product  
    such that the canonical quotient map is \emph{bounded}; this topology is finer 
    than the projective tensor product topology. Nonetheless, $T\C^\infty(\M)$ and 
    $T^*\C^\infty(\M)$ do assume the form given in Subsection \ref{s2-kin-geom} 
    (see the proof of Theorem 42.17, pp. 447--448 of \cite{km}). It also turns 
    out that the particular structure of $\C^\infty_c(\M)$, together with Theorems 
    6.14, pp. 72--73 and 28.7, pp. 280--281 of \cite{km}, imply that every 
    kinematical tangent vector on $\C^\infty(\M)$ is also an \emph{operational} 
    one, i.e. it defines a point derivation on (conveniently) smooth maps 
    $F:\C^\infty(\M)\To\RR$. 
  \end{remark}
  
  In principle, we could develop essentially all tools of differential calculus 
  by using convenient smoothness. However, for the purposes of this paper, it is 
  often preferrable to use a stronger concept of smoothness. Such a notion is 
  provided, for instance, by Michal \cite{michal} and Bastiani \cite{bastiani}. 
  This is also the notion employed in the accounts of infinite dimensional 
  differential calculus done by Milnor \cite{milnor} and Hamilton \cite{hamilton}, 
  and all the basic results of Calculus we present in the remainder of this Appendix
  are formulated in this context (see, however, Remark \ref{a1r2} below). The
  basic definition is as follows (See also Definition \ref{s2d3} for the special 
  case of real-valued maps):
  
  \begin{definition}\label{a1d2}
    Let $\Fun_1,\Fun_2$ be locally convex topological vector spaces, $\UU\subset\Fun_1$ open,
    and $F:\UU\To\Fun_2$ a continuous map. We say that $F$ is \emph{(MB-)differentiable 
      of order $m$} (``MB'' stands for the names of Michal and Bastiani) if for all $k=1,
    \ldots,m$ the $k$-th order directional (Gâteaux) derivatives 
    \begin{equation}\label{a1e1}
      F^{(k)}[\varphi](\vec{\varphi}_1,\ldots,\vec{\varphi}_k)\doteq\frac{\dd^k}{\dd\lambda_1
        \cdots\dd\lambda_k}\Restr{\lambda_1=\cdots=\lambda_k=0}F\left(\varphi+\sum^k_{j=1}\lambda_j
        \vec{\varphi}_j\right)
    \end{equation}
    exist as jointly continuous maps from $\UU\times\Fun^k_1\ni(\varphi,\vec{\varphi}_1,\ldots,
    \vec{\varphi}_k)$ to $\Fun_2$. If $F$ is differentiable of order $m$ for all $m\in\NN$, we 
    say that $F$ is \emph{(MB-)smooth}.\footnote{\label{a1f2} MB differentiability and MB 
      smoothness are respectively listed in Keller's treatise \cite{keller} as ``$\C^k_c$- and 
      $\C^\infty_c$-differentiability''. Here we avoid his nomenclature, for it clashes with the 
      usual notation for differentiable and smooth functions with compact support.}
  \end{definition}
  
  The right-hand side of formula \eqref{a1e1} should be understood as the differentiation
  of a $k$-parameter curve taking values in $\Fun_2$, for fixed $\varphi,\vec{\varphi}_1,\ldots,
  \vec{\varphi}_k$. The argument of $F$ inside the limit is guaranteed to lie inside $\UU$ 
  for sufficiently small $\lambda_1,\ldots,\lambda_k$.
  
  It follows from Definition \ref{a1d2} that if $F:\UU\subset\Fun_1\To\Fun_2$ is MB-differentiable
  of order $m>0$ then the maps $\UU\ni\varphi\mapsto F^{(k)}[\varphi]\in\Li^k(\Fun_1,\Fun_2)$ are
  \emph{continuous} for all $1\leq k\leq m$, where $\Li^k(\Fun_1,\Fun_2)$ is the locally convex
  topological vector space of all $k$-linear maps from $\Fun_1^k$ to $\Fun_2$ endowed with the 
  compact-open topology. If $\Fun_1$ is semi-Montel (i.e. closed and bounded subsets of 
  $\Fun_1$ are compact), such topology amounts to uniform convergence in bounded subsets of 
  $\Fun_1^k$. If $\Fun_1^k$ is compactly generated (e.g. when $\Fun_1$ is metrizable, see 
  e.g. Proposition 3.3.20, pp. 152 of \cite{engelking} and footnote \ref{a1f1} above) and $\Fun_2$ 
  is complete, then by Proposition 16.6.2, pp. 361 of \cite{jarchow} $\Li^k(\Fun_1,\Fun_2)$ is also 
  complete.

  Given $\UU$ an \emph{arbitrary} (i.e. not necessarily open) subset of $\Fun_1$, we say that 
  a continuous map $F:\UU\To\Fun_2$ is differentiable of order $m$ (resp. smooth) if there is 
  $\VV\supset\UU$ open in the compact-open topology and a functional $\tilde{F}:\VV\To\Fun_2$ 
  extending $F$ (i.e. $\tilde{F}\restr{\UU}=F$) such that $\tilde{F}$ is differentiable of order 
  $m$ (resp. smooth). For completely arbitrary $\UU$, the derivatives of $F$ on $\UU$ depend on 
  the choice of extension $\tilde{F}$ (take for instance $\UU=\{\varphi\}$ for some $\varphi\in
  \Fun_1$). However, if $\UU$ happens to have a nonvoid interior, then it is easily shown that 
  the derivatives of $F$ on $\UU$ do not depend on the choice of extension. Under certain 
  conditions on $F$, one can weaken this condition (see, for instance, Remark \ref{s2r4}).
  
  \begin{remark}\label{a1r2}
    For Mackey-complete locally convex topological vector spaces (also called 
    $c^\infty$\emph{-complete} or \emph{convenient} topological vector spaces), 
    convenient smoothness enjoys essentially all the rules of Calculus presented in the remainder 
    of this Appendix assuming MB differentiability (see e.g. footnote \ref{a1f3} below). Moreover, 
    for Fréchet spaces (which are convenient and whose topology coincides with the corresponding 
    $c^\infty$-topology) convenient and MB smoothness coincide (see e.g. Theorem 1, pp. 77 of 
    \cite{frolicher} together with Theorem 2.14, pp. 20--21 of \cite{km}).
  \end{remark}
  
  Let $\gamma:[a,b]\To\Fun$, $a<b\in\RR$, be a continuous curve segment in the \emph{complete}
  locally convex topological vector space $\Fun$. We can define the \emph{(Riemann) integral} of 
  $\gamma$ along $[a,b]$
  \[
  \int^b_{a}\gamma(\lambda)\ud\lambda\in\Fun
  \]
  as the unique linear map from the space $\C([a,b],\Fun)$ of continuous curves from
  $[a,b]$ to $\Fun$ into the space $\Fun$ such that\footnote{\label{a1f3} However, as argued 
    e.g. in Proposition 2.7, pp. 17 of \cite{km}, if $\gamma$ is Lipschitz (i.e. the subset 
    $\{(t-s)^{-1}(\gamma(t)-\gamma(s))\ |\ t\neq s\ ,\,a\leq t,s\leq b\}$ is bounded) then it 
    suffices to assume that $\Fun$ is convenient to get the Riemann integral of $\gamma$ along 
    $[a,b]$ with all the properties discussed in this Appendix.}:
  \begin{enumerate}
  \item For any continuous linear functional $u:\Fun\To\RR$, we have that $u\left(\int^b_{a}
      \gamma(\lambda)\ud\lambda\right)=\int^b_{a}u(\gamma(\lambda))\ud\lambda$;
  \item For any continuous seminorm $\|\cdot\|$ on $\Fun$, we have that $\left\|\int^b_{a}
      \gamma(\lambda)\ud\lambda\right\|\leq\int^b_{a}\|\gamma(\lambda)\|\ud\lambda$;
  \item If $a<c<b\in\RR$, then $\int^b_{a}\gamma(\lambda)\ud\lambda=\int^c_{a}\gamma(\lambda)
    \ud\lambda+\int^b_{c}\gamma(\lambda)\ud\lambda$.
  \end{enumerate}
  
  The \emph{Fundamental Theorem of Calculus} holds for the Riemann integral of curves taking
  values in $\Fun$:
  
  \begin{theorem}[\cite{hamilton}, Theorems 2.2.3 and 2.2.2]\label{a1t1}
    Let $\gamma_0:[a,b]\To\Fun$ be a continuous curve, $a\leq t\leq b$, and define
    $\gamma_1(t)\doteq\int^t_{a}\gamma_0(\lambda)\ud\lambda$. Then $\gamma_1:[a,b]\To\Fun$ 
    is a $\C^1$ curve, and $\gamma'_1(t)=\gamma_0(t)$. Conversely, if $\gamma_1:[a,b]\To
    \Fun$ is a $\C^1$ curve, then $\gamma_1(b)-\gamma_1(a)=\int^b_{a}\gamma'_1(\lambda)\ud
    \lambda$.\hspace*{\fill}\qed
  \end{theorem}
  
  \begin{corollary}[\cite{hamilton}, Theorem 3.2.2]\label{a1c1}
    Let $F:\UU\subset\Fun_1\To\Fun_2$ be a continuous map with $\Fun_2$ complete, $\varphi_0
    \in\UU$, and $\vec{\varphi}\in\UU-\varphi_0\doteq\{\varphi-\varphi_0\in\Fun_1\ |\ \varphi
    \in\UU\}$. Assume that $\UU$ is convex for simplicity. If $F$ is differentiable of order 
    one in the sense of Definition \ref{a1d2}, then 
    \begin{equation}\label{a1e2}
      F(\varphi_0+\vec{\varphi})-F(\varphi_0)=\int^1_0F^{(1)}[\varphi_0+\lambda\vec{\varphi}]
      (\vec{\varphi})\ud\lambda\ .
    \end{equation}\hspace*{\fill}\qed
  \end{corollary}
  
  With the aid of the fundamental theorem of Calculus \ref{a1t1}, the following key results 
  can be proven. First, the usual linearity property for first-order derivatives holds:

  \begin{lemma}[\cite{hamilton}, Lemma 3.2.3 and Theorem 3.2.5]\label{a1l1}
    Let $F:\UU\subset\Fun_1\To\Fun_2$ be a continuous map with $\Fun_2$ complete, $\varphi
    \in\UU$. If $F$ is differentiable of order one in the sense of Definition \ref{a1d2}, 
    then for all scalars $\lambda,\mu$ and all $\vec{\varphi},\vec{\varphi}'\in\Fun_1$
    we have that
    \[
    F^{(1)}[\varphi](\lambda\vec{\varphi}+\mu\vec{\varphi}')=\lambda F^{(1)}[\varphi](\vec{\varphi})
    +\mu F^{(1)}[\varphi](\vec{\varphi}')\ .
    \]\hspace*{\fill}\qed
  \end{lemma}
  
  Next, the \emph{chain rule} holds:
  
  \begin{theorem}[\cite{hamilton}, Theorem 3.3.4]\label{a1t2}
    Let $F:\UU\subset\Fun_1\To\Fun_2$, $G:\VV\subset\Fun_2\To\Fun_3$ be respectively continuous 
    maps from open subsets $\UU,\VV$ of locally convex topological vector spaces $\Fun_1,\Fun_2$
    into $\Fun_2$ and the locally convex topological vector space $\Fun_3$, such that
    $F(\UU)\subset\VV$. Suppose that $\Fun_2$ and $\Fun_3$ are complete. If $F$ (resp. $G$) is 
    once differentiable on $\UU$ (resp. $\VV$) in the sense of Definition \ref{a1d2}, then for 
    all $\varphi\in\UU$, $\vec{\varphi}\in\Fun_1$ we have that
    \begin{equation}\label{a1e3}
      (G\circ F)^{(1)}(\varphi)(\vec{\varphi})=G^{(1)}[F(\varphi)](F^{(1)}[\varphi](\vec{\varphi}))\ .
    \end{equation}\hspace*{\fill}\qed
  \end{theorem}
  
  The chain rule \eqref{a1e3} yields, after taking direct sums, the \emph{Leibniz's rule} 
  for derivatives of composition of $n$-tuples of maps $F_1,\ldots,F_n$ with 
  a continuous $n$-linear map $\psi$
  \begin{equation}\label{a1e4}
    (\psi(F_1,\ldots,F_n))^{(1)}[\varphi](\vec{\varphi})=\sum^n_{j=1}\psi(F_1[\varphi],
    \ldots,F^{(1)}_j[\varphi](\vec{\varphi}),\ldots,F_n[\varphi])\ .
  \end{equation}
  This, together with the fundamental theorem of Calculus \eqref{a1e2}, 
  yields the \emph{integration by parts formula} and, even more importantly, 
  \emph{Taylor's formula with (integral) remainder}
  \begin{equation}\label{a1e5}
    F(\varphi_0+\vec{\varphi})=\sum^k_{j=0}\frac{1}{j!}F^{(j)}[\varphi_0]
    (\vec{\varphi},\ldots,\vec{\varphi})+\int^1_0\frac{(1-\lambda)^k}{k!}F^{(k+1)}
    [\varphi_0+\lambda\vec{\varphi}](\vec{\varphi},\ldots,\vec{\varphi})\ud\lambda\ .
  \end{equation}
  To see this, note that Leibniz's rule implies the following key formula:
  \begin{equation}\label{a1e6}
    \begin{split}
      \frac{(1-\lambda)^{k-1}}{(k-1)!}F^{(k)}[\varphi_0+\lambda\vec{\varphi}](\vec{\varphi},
      \ldots,\vec{\varphi}) &=\frac{(1-\lambda)^k}{k!}F^{(k+1)}[\varphi_0+\lambda\vec{\varphi}]
      (\vec{\varphi},\ldots,\vec{\varphi})\\ &-\frac{\ud}{\ud\lambda}
      \left[\frac{(1-\lambda)^k}{k!}F^{(k)}[\varphi_0+\lambda\vec{\varphi}](\vec{\varphi},\ldots,
        \vec{\varphi})\right]\ .
    \end{split}
  \end{equation}
  Integrating both sides of formula \eqref{a1e6} from $\lambda=0$ to $\lambda=1$ by means
  of the fundamental theorem of Calculus \eqref{a1e2} yields the fundamental induction step 
  from $k-1$ to $k$. Since the case $k=0$ of \eqref{a1e5} is settled by the fundamental
  theorem of Calculus itself, we are done.
  
  For the convenience of the reader, we prove the generalization of the chain rule \eqref{a1e3}
  for higher derivatives, since this proof is not easy to find in the literature at the 
  present level of generality. We follow the argument employed in \cite{kai}.
  
  \begin{corollary}[Faà di Bruno's formula]\label{a1c2}
    Let $F:\UU\subset\Fun_1\To\Fun_2$, $G:\VV\subset\Fun_2\To\Fun_3$ satisfy the hypotheses of
    Theorem \ref{a1t2}. If $F$ (resp. $G$) is $m$-times differentiable on $\UU$ 
    (resp. $\VV$), then $G\circ F$ is also $m$-times differentiable on $\UU$, and
    for all $1\leq k\leq m$, 
    \begin{equation}\label{a1e7}
      (G\circ F)^{(k)}[\varphi](\vec{\varphi}_1,\ldots,\vec{\varphi}_k)=\sum_{\pi\in P_k}
      G^{(|\pi|)}[F(\varphi)]\left(\bigotimes_{I\in\pi} F^{(|I|)}[\varphi](\otimes_{j\in I}
        \vec{\varphi}_j)\right)\ ,
    \end{equation}
    where $P_k$ is the set of all partitions $\pi=\{I_1,\ldots,I_l\}$ of $\{1,\ldots,k\}$,
    that is, $I_j\neq\varnothing$, $I_j\cap I_{j'}=\varnothing$ for $j\neq j'$ and 
    $\cup^l_{j=1}I_j=\{1,\ldots,k\}$.
  \end{corollary}
  \begin{proof}
    We proceed by induction on $k$. The case $k=1$ is just the usual chain rule \eqref{a1e3}. 
    Assume that the formula is valid up to order $k-1$ along $\vec{\varphi}_1,\ldots,
    \vec{\varphi}_{k-1}$. Then for each partition $\pi$ of $\{1,\ldots,k-1\}$ in the above sum 
    we have, by Leibniz's rule \eqref{a1e4},
    \[
    \begin{split}
      \Bigg[G^{(|\pi|)}&\left.\circ\:F\left(\bigotimes_{I\in\pi} F^{(|I|)}(\otimes_{j\in I}\vec{\varphi}_j)
        \right)\right]^{(1)}[\varphi](\vec{\varphi}_k) \\ 
      &=G^{(|\pi|+1)}[F(\varphi)]\left(F^{(1)}[\varphi](\vec{\varphi}_k)\otimes\bigotimes_{I\in\pi} 
        F^{(|I|)}[\varphi](\otimes_{j\in I}\vec{\varphi}_j)\right) \\ 
      &+\sum_{I'\in\pi}G^{(|\pi|)}[F(\varphi)]\left(F^{(|I'|+1)}[\varphi]\left(\vec{\varphi}_k\otimes
          \bigotimes_{j\in I'}\vec{\varphi}_j\right)\otimes\bigotimes_{I\in\pi\sm\{I'\}} F^{(|I|)}[\varphi]
        (\otimes_{l\in I}\vec{\varphi}_l)\right)\ .
    \end{split}
    \]
    However, any partition $\pi'$ of $\{1,\ldots,k\}$ is either of the form $\pi'=\{\{k\}\}\cup\pi$
    or $\pi'=(\pi\sm\{I'\})\cup\{I'\cup\{k\}\}$ for some $I'\in\pi$, $\pi\in P_{k-1}$. Hence, 
    summing the above identities over all such $\pi$ gives the desired result.
  \end{proof}
  
  A consequence of Faà di Bruno's formula \eqref{a1e7} is the generalization of Leibniz's rule 
  \eqref{a1e4} for higher order derivatives of composition of $l$-tuples of maps $F_1,\ldots,F_l$ 
  with a continuous $l$-linear map $\psi$
  \begin{equation}\label{a1e8}
    (\psi(F_1,\ldots,F_l))^{(k)}[\varphi]\left(\vec{\varphi}_1,\ldots,\vec{\varphi}_k\right)
    =\sum_{\{I_1,\ldots,I_l\}\in\tilde{P}_{k,l}}\psi\left(F_1^{(|I_1|)}[\varphi](\otimes_{j\in I_1}
      \vec{\varphi}_j),\ldots,F_l^{(|I_l|)}[\varphi](\otimes_{j\in I_l}\vec{\varphi}_j)\right)\ ,
  \end{equation}
  where $\tilde{P}_{k,l}$ is the set of all partitions $\pi=\{I_1,\ldots,I_l\}$ of $\{1,\ldots,k\}$
  in $l$ \emph{possibly (but not all) empty} subsets, i.e. $I_j\cap I_{j'}=\varnothing$ for 
  $j\neq j'$ and $\cup^l_{j=1}I_j=\{1,\ldots,k\}$. As another application, we obtain the so-called
  \emph{$k$-th order resolvent formula} \eqref{a1e12} below which shall often be useful. Consider 
  two MB-differentiable maps $F:\UU\times\Fun_1\To\Fun_2$, $G:\UU\times\Fun_2\To\Fun_1$ of order 
  one, where $\Fun_1,\Fun_2$ are locally convex topological vector spaces and $\UU\subset\Fun$ 
  is a nonvoid open subset of the locally convex topological vector space $\Fun$. For notational 
  convenience, we also occasionally write $F(\varphi,\vec{\varphi})\doteq F[\varphi]\vec{\varphi}$, 
  $G(\varphi,\vec{\psi})\doteq G[\varphi]\vec{\psi}$. Suppose that both $F$ and $G$ are 
  \emph{linear} in their second arguments and satisfy
  \begin{equation}\label{a1e9}
    \begin{split}
      F[\varphi]G[\varphi]\vec{\psi} &=\vec{\psi}\ ,\quad\forall\varphi\in\UU\ ,\,
      \vec{\psi}\in\Fun_2\ ,\\
      G[\varphi]F[\varphi]\vec{\varphi} &=\vec{\varphi}\ ,\quad\forall\varphi\in\UU\ ,\,
      \vec{\varphi}\in\Fun_1\ .
    \end{split}
  \end{equation}
  If we define 
  \begin{equation}\label{a1e10}
    D^k_1F[\varphi](\vec{\varphi}_1,\ldots,\vec{\varphi}_k)\vec{\varphi}=F^{(k)}[\varphi,
    \vec{\varphi}]((\vec{\varphi}_1,0),\ldots,(\vec{\varphi}_k,0))\ ,\quad D^1_1\doteq D_1\ ,\,
    D^0_1=\id\ ,
  \end{equation}
  then by the chain rule \eqref{a1e3} applied to the pair of maps $F,(\id,G)$ and \eqref{a1e9} 
  we have the \emph{(first-order) resolvent formula}
  \begin{equation}\label{a1e11}
    D_1G[\varphi](\vec{\varphi}_1)\vec{\psi}=-G[\varphi]D_1F[\varphi](\vec{\varphi}_1)G[\varphi]
    \vec{\psi}\ .
  \end{equation}
  It follows from the above formula that if in addition $F$ is MB-smooth, then so 
  is $G$. More precisely, in this case we obtain the following (not so pleasant) higher-order 
  generalization of \eqref{a1e11}, obtained by induction on $k\geq 1$ from \eqref{a1e11} and 
  an argument analogous to the one used in the proof of Corollary \ref{a1c2}:
  \begin{equation}\label{a1e12}
    D^k_1G[\varphi]\left(\vec{\varphi}_1,\ldots,\vec{\varphi}_k\right)\vec{\psi}
    =\sum^k_{l=1}(-1)^l\sum_{\{I_1,\ldots,I_l\}\in P_k}\sum_{\sigma\in S_l}\left(\prod^l_{j=1}G[\varphi]
      D^{|I_{\sigma(j)}|}F[\varphi](\otimes_{i\in I_{\sigma(j)}}\vec{\varphi}_i)\right)G[\varphi]
    \vec{\psi}\ .
  \end{equation}
  
  Here, $P_k$ is again the set of all partitions of $\{1,\ldots,k\}$ as in the statement of
  Corollary \ref{a1c2}, whereas $S_l$ is the set of all permutations of $\{1,\ldots,l\}$.

  Finally, one can show that the order of differentiation for higher order derivatives is 
  irrelevant:
  
  \begin{theorem}[\cite{hamilton}, Theorem 3.6.2]\label{a1t3}
    Let $F:\UU\subset\Fun_1\To\Fun_2$ be a continuous map with $\Fun_2$ complete. If $F$ 
    is differentiable of order $m>1$ in the sense of Definition \ref{a1d2}, then $F^{(k)}
    [\varphi]:\Fun^k_1\ni(\vec{\varphi}_1,\ldots,\vec{\varphi}_k)\mapsto F^{(k)}[\varphi]
    (\vec{\varphi}_1,\ldots,\vec{\varphi}_k)\in\Fun_2$ is a symmetric, $k$-linear map for 
    all fixed $\varphi\in\UU$, $2\leq k\leq m$.\hspace*{\fill}\qed
  \end{theorem}
  
\end{appendix}


\begin{thebibliography}{200}
\bibitem{abraham}
  R. Abraham, J. E. Marsden, \emph{Foundations of Mechanics, Second Edition}
  (Addison-Wesley, 1978).
  
\bibitem{adams} R. A. Adams, J. J. F. Fournier, \emph{Sobolev Spaces, Second Edition}
  (Elsevier, 2002).

\bibitem{anderson}
  I. M. Anderson, \emph{The Variational Bicomplex}. 
  Technical report, Utah State University, 1989.

\bibitem{bgp}
  C. Bär, N. Ginoux, F. Pfäffle, \emph{Wave Equations on Lorentzian Manifolds and Quantization}
  (European Mathematical Society, 2007).

\bibitem{bastiani}
  A. Bastiani, \emph{Applications Différentiables et Varietés Différentiables de Dimension Infinie}. 
  J. Anal. Math. \textbf{13} (1964) 1--114.

\bibitem{benavs}
  J. J. Benavides Navarro, E. Minguzzi, \emph{Global Hyperbolicity is Stable in the Interval Topology}.
  J. Math. Phys. \textbf{52} (2011) 112504.
  \texttt{arXiv:1108.5120 [gr-qc]}.

\bibitem{bernsan1}
  A. N. Bernal, M. Sánchez, \emph{On Smooth Cauchy Hypersurfaces and Geroch's Splitting Theorem}. 
  Commun. Math. Phys. \textbf{243} (2003) 461--470. 
  \texttt{arXiv:gr-qc/0306108}.

\bibitem{bernsan2}
  A. N. Bernal, M. Sánchez, \emph{Smoothness of Time Functions and the Metric Splitting of Globally Hyperbolic Spacetimes}. 
  Commun. Math. Phys. \textbf{257} (2005) 43--50.
  \texttt{arXiv:gr-qc/0401112}.

\bibitem{bernsan3}
  A. N. Bernal, M. Sánchez, \emph{Further Results on the Smoothability of Cauchy Hypersurfaces and Cauchy Time Functions}. 
  Lett. Math. Phys. \textbf{77} (2006) 183--197. 
  \texttt{arXiv:gr-qc/0512095}.

\bibitem{bernsan4}
  A. N. Bernal, M. Sánchez, \emph{Globally Hyperbolic Spacetimes can be Defined as "Causal" instead of "Strongly Causal"}. 
  Clas. Quantum Grav. \textbf{24} (2007) 745--749. 
  \texttt{arXiv:gr-qc/0611138}.

\bibitem{binz}
  E. Binz, J. \'{S}niatycki, H. Fischer, \emph{Geometry of Classical Fields} 
  (North-Holland, 1988; reprinted by Dover, 2006).


\bibitem{brdut}
  F. Brennecke, M. Dütsch, \emph{Removal of Violations of the Master Ward Identity in Perturbative QFT}. 
  Rev. Math. Phys. \textbf{20} (2008) 119--172. 
  \texttt{arXiv:0705.3160 [hep-th]}.
  

\bibitem{broudh2}
  C. Brouder, N. V. Dang, F. Hélein, \emph{Boundedness and Continuity of the Fundamental Operations on Distributions Having a Specified Wave Front Set (with a counterexample by Semyon Alesker)}. 
  Studia Math. \textbf{232} (2016) 201--226.
  \texttt{arXiv:1409.7662 [math-ph]}.

\bibitem{broudlr}
  C. Brouder, N. V. Dang, C. Laurent-Gengoux, K. Rejzner, \emph{Properties of field functionals and characterization of local functionals}.
  J. Math. Phys. \textbf{59} (2018) 023508.
  \texttt{arXiv:1705.01937 [math-ph]}.
  
\bibitem{brudf}
  R. Brunetti, M. Dütsch, K. Fredenhagen, \emph{Perturbative Algebraic Quantum Field Theory and Renormalization Groups}. 
  Adv. Theor. Math. Phys. \textbf{13} (2009) 1541--1599. 
  \texttt{arXiv:0901.2038 [math-ph]}.

\bibitem{brufre1}
  R. Brunetti, K. Fredenhagen, \emph{Microlocal Analysis and Interacting Quantum Field Theories: Renormalization on Physical Backgrounds}. 
  Commun. Math. Phys. \textbf{208} (2000) 623--661.
  \texttt{arXiv:math-ph/9903028}.
  
\bibitem{brufre2}
  R. Brunetti, K. Fredenhagen, \emph{Quantum Field Theory on Curved Backgrounds}. 
  In: C. Bär, K. Fredenhagen (eds.), \emph{Quantum Field Theory on Curved Spacetimes: Concepts and Mathematical Foundations}. 
  Lecture Notes in Physics \textbf{786} (Springer-Verlag, 2009), pp. 129--155. 
  \texttt{arXiv:0901.2063 [math-ph]}.
  
\bibitem{bfk} 
  R. Brunetti, K. Fredenhagen, M. Köhler, \emph{The Microlocal Spectrum Condition and Wick Polynomials of Free Fields on Curved Spacetimes}. 
  Commun. Math. Phys. \textbf{180} (1996) 633--652.
  \texttt{arXiv:gr-qc/9510056}.

\bibitem{bfrej}
  R. Brunetti, K. Fredenhagen, K. Rejzner, \emph{Quantum Gravity from the Point of View of Locally Covariant Quantum Field Theory}.
  Commun. Math. Phys. \textbf{345} (2016) 741-779.
  \texttt{arXiv:1306.1058 [math-ph]}.

\bibitem{bfr2}
  R. Brunetti, K. Fredenhagen, P. L. Ribeiro, in preparation.

  
\bibitem{bgy}
  R. L. Bryant, P. A. Griffiths, D. Yang, \emph{Characteristics and Existence of Isometric Embeddings}.
  Duke Math. J. \textbf{50} (1983) 893--994.

\bibitem{carci}
  J. F. Cariñena, M. Crampin, L. A. Ibort, \emph{On the Multisymplectic Formalism for First Order Field Theories}.
  Diff. Geom. Appl. \textbf{1} (1991) 345--374.

\bibitem{chazp}
  J. Chazarain, A. Piriou, \emph{Introduction to the Theory of Linear Partial Differential Equations}
  (North-Holland, 1982).

\bibitem{christo}
  D. Christodoulou, \emph{The Action Principle and Partial Differential Equations} 
  (Princeton University Press, 2000).


\bibitem{crnwit}
  C. Crnkovi\'c, E. Witten, \emph{Covariant Description of Canonical Formalism in Geometrical Theories}. 
  In: S. W. Hawking, W. Israel (eds.), \emph{Three Hundred Years of Gravitation} (Cambridge University Press, 1987), pp. 676--684.
  
\bibitem{dabrowski1}
  Y. Dabrowski, \emph{Functional properties of Generalized H\"ormander spaces of distributions I: Duality Theory, Completions and Bornologifications}.
  Preprint, \texttt{arXiv:1411.3012 [math-ph]}

\bibitem{dabrowski2}
  Y. Dabrowski, \emph{Functional properties of Generalized H\"ormander spaces of distributions II: Multilinear maps and applications to spaces of functionals with wave front set conditions}.
  Preprint, \texttt{arXiv:1412.1749 [math-ph]}
  
\bibitem{dabrowskib}
  Y. Dabrowski, C. Brouder, \emph{Functional properties of H\"ormander's space of distributions having a specified wavefront set}.
  Commun. Math. Phys. \textbf{332} (2014) 1345--1380.
  \texttt{arXiv:1308.1061 [math-ph]}.
  
\bibitem{dedonder}
  T. de Donder, \emph{Théorie Invariante du Calcul des Variations} 
  (Gauthier-Villars, 1935).

\bibitem{dewitt}
  B. S. DeWitt, \emph{The Spacetime Approach to Quantum Field Theory}. 
  In: B. S. DeWitt, R. Stora (eds.), Les Houches Session XL, \emph{Relativity, Groups and Topology II} (North-Holland, 1983), pp. 382--738.

\bibitem{dencker}
  N. Dencker, \emph{On The Propagation of Polarization Sets for Systems of Real Principal Type}.
  J. Funct. Anal. \textbf{46} (1982) 351--372.

\bibitem{dimock}
  J. Dimock, \emph{Algebras of Local Observables on a Manifold}.
  Commun. Math. Phys. \textbf{77} (1980) 219--228.

\bibitem{duister}
  J. J. Duistermaat, \emph{Fourier Integral Operators} 
  (Birkhäuser, 1996).

\bibitem{dutfre1}
  M. Dütsch, K. Fredenhagen, \emph{The Master Ward Identity and Generalized Schwinger-Dyson Equation in Classical Field Theory}. 
  Commun. Math. Phys. \textbf{243} (2003) 275--314.
  \texttt{arXiv:hep-th/0211242}.
  
\bibitem{dutfre2}
  M. Dütsch, K. Fredenhagen, \emph{Causal Perturbation Theory in Terms of Retarded Products, and a Proof of the Action Ward Identity}. 
  Rev. Math. Phys. \textbf{16} (2004) 1291--1348.
  \texttt{arXiv:hep-th/0501228}.
  
\bibitem{engelking}
  R. Engelking, \emph{General Topology. Revised and Completed Edition}
  (Heldermann Verlag, 1989).

\bibitem{forgerr}
  M. Forger, S. V. Romero, \emph{Covariant Poisson Brackets in Geometric Field Theory}.
  Commun. Math. Phys. \textbf{256} (2005) 375--410.
  \texttt{arXiv:math-ph/0408008}.
  
\bibitem{frerej}
  K. Fredenhagen, K. Rejzner, \emph{Batalin-Vilkovisky Formalism in the Functional Approach to Classical Field Theory}.
  Commun. Math. Phys. \textbf{314} (2012) 93--127.
  \texttt{arXiv:1101.5112 [math-ph]}.

\bibitem{frolicher}
  A. Frölicher, \emph{Smooth Structures}. 
  In: K. H. Kamps, D. Pumplün, W. Tholen (eds.), \emph{Category Theory -- Applications to Algebra, Logic and Topology}.
  Lecture Notes in Mathematics \textbf{962} (Springer-Verlag, 1982), pp. 69--81.


\bibitem{geroch}
  R. Geroch, \emph{Domain of Dependence}.
  J. Math. Phys. \textbf{11} (1970) 437--449.

\bibitem{glockner}
  H. Glöckner, \emph{Discontinuous Non-linear Mappings on Locally Convex Direct Limits}.  
  Publ. Math. Debrecen \textbf{68}, (2006) 1--13.
  \texttt{arXiv:math/0503387}.
  
\bibitem{gotay}
  M. J. Gotay, \emph{A Multisymplectic Framework for Classical Field Theory and the Calculus of Variations}.  
  In: M. Francaviglia (ed.), \emph{Mechanics, Analysis and Geometry: 200 Years after Lagrange} (North-Holland, 1991), pp. 203--235.
  Available at the author's homepage (\texttt{http://www.pims.math.ca/$\mathtt{\sim}$gotay/Multi\_I.pdf}).

\bibitem{grafakos}
  L. Grafakos, \emph{Classical Fourier Analysis. Second Edition}
  (Springer-Verlag, 2008).

\bibitem{haag}
  R. Haag, \emph{Local Quantum Physics -- Fields, Particles, Algebras. Second Edition}
  (Springer-Verlag, 1996).
  
\bibitem{hamilton}
  R. S. Hamilton, \emph{The Inverse Function Theorem of Nash and Moser}.
  Bull. Amer. Math. Soc. (N.S.) \textbf{7} (1982) 65--222.
  
\bibitem{hawkellis}
  S. W. Hawking, G. F. R. Ellis, \emph{The Large Scale Structure of Space-Time}
  (Cambridge University Press, 1973).

\bibitem{helein1}
  F. Héléin, \emph{Multisymplectic Formalism and the Covariant Phase Space}.
  In: R. Bielawski, K. Houston, M. Speight (eds.), \emph{Variational Problems in Differential Geometry} (Cambridge University Press, 2012), pp. 94--126.
  \texttt{arXiv:1106.2086 [math]}.
  
\bibitem{helein2}
  F. Héléin, \emph{First Integrals for Nonlinear Dispersive Equations}.
  Trans. Amer. Math. Soc. \textbf{368} (2016) 6939--6978.
  \texttt{arXiv:1311.0722 [math-ph]}.

\bibitem{hinv}
  P. Hintz, A. Vasy, \emph{Global Analysis of Quasilinear Wave Equations on Asymptotically Kerr-de Sitter Spaces}.
  Int. Math. Res. Not. \textbf{2016} 5355--5426.
  \texttt{arXiv:1404.1348 [math]}.

\bibitem{horm1}
  L. Hörmander, \emph{The Analysis of Linear Partial Differential Operators I -- Distribution Theory and Fourier Analysis. Second Edition} 
  (Springer-Verlag, 1990).
  
\bibitem{horm2}
  L. Hörmander, \emph{The Analysis of Linear Partial Differential Operators III -- Pseudodifferential Operators. Second Edition}
  (Springer-Verlag, 1994).
  
\bibitem{horm3}
  L. Hörmander, \emph{Lectures on Nonlinear Hyperbolic Differential Operators}.
  Mathématiques \& Applications \textbf{26} (Springer-Verlag, 1997).

\bibitem{jakobs}
  S. Jakobs, \emph{Eichbrücken in der klassichen Feldtheorie}.
  Diplomarbeit, Universität Hamburg (2009). 
  \texttt{http://www-library.desy.de/preparch/desy/thesis/desy-thesis-09-009.pdf}

\bibitem{jarchow}
  H. Jarchow, \emph{Locally Convex Spaces}
  (B. G. Teubner Stuttgart, 1981).

\bibitem{keller} 
  H. H. Keller, \emph{Differential Calculus in Locally Convex Spaces}. 
  Lecture Notes in Mathematics \textbf{417} (Springer-Verlag, 1974).

\bibitem{kai}
  K. J. Keller, \emph{Dimensional Regularization in Position Space and a Forest Formula for Regularized Epstein-Glaser Renormalization}.
  PhD Thesis, Universität Hamburg (2010).
  
\bibitem{kijowski}
  J. Kijowski, \emph{A Finite-Dimensional Canonical Formalism in the Classical Field Theory}. 
  Commun. Math. Phys. \textbf{30} (1973) 99--128.
  
\bibitem{klai1}
  S. Klainerman, \textit{Global Existence for Nonlinear Wave Equations}. 
  Comm. Pure Appl. Math. \textbf{33} (1980) 43--101.

\bibitem{klai2}
  S. Klainerman, \textit{Long-Time Behavior of Solutions to Nonlinear Evolution Equations}.
  Arch. Rat. Mech. Anal. \textbf{78} (1982) 73--98.

\bibitem{kolar}
  I. Kolá\v{r}, \emph{A Geometrical Version of the Higher Order Hamilton Formalism in Fibered Manifolds}.
  J. Geom. Phys. \textbf{1} (1984) 127--137.
  
\bibitem{kms}
  I. Kolá\v{r}, P. W. Michor, J. Slovák, \emph{Natural Operations in Differential Geometry}
  (Springer-Verlag, 1993).

\bibitem{kralv}
  I. S. Krasil'shchik, V. V. Lychagin, A. M. Vinogradov, \emph{Geometry of Jet Spaces and Nonlinear Partial Differential Equations}
  (Gordon and Breach, 1986).

\bibitem{km}
  A. Kriegl, P. W. Michor, \emph{The Convenient Setting of Global Analysis}
  (American Mathematical Society, 1997).

\bibitem{lerner}
  D. E. Lerner, \emph{The Space of Lorentz Metrics}.
  Commun. Math. Phys. \textbf{32} (1973) 19--38.
  
\bibitem{leyrob}
  P. Leyland, J. E. Roberts, \emph{The Cohomology of Nets over Minkowski Space}. 
  Commun. Math. Phys. \textbf{62} (1978) 173--189.
  
\bibitem{liess} 
  O. Liess, \emph{Conical Refractions and Higher Microlocalization}. 
  Lecture Notes in Mathematics \textbf{1555} (Springer-Verlag, 1993).
  
\bibitem{majda}
  A. Majda, \emph{Compressible Fluid Flow and Systems of Conservation Laws in Several Space Variables}
  (Springer-Verlag, 1984).

\bibitem{marolf}
  D. Marolf, \emph{The Generalized Peierls Bracket}.
  Ann. Phys. (N.Y.) \textbf{236} (1994) 392--412.

\bibitem{meise}
  R. Meise, \emph{Nicht-Nuklearität von Räumen beliebig oft differenzierbarer Funktionen}.
  Arch. Math. \textbf{34} (1980) 143--148.

\bibitem{michal}
  A. D. Michal, \emph{Differential Calculus in Linear Topological Spaces}. 
  Proc. Nat. Acad. Sci. U. S. A. \textbf{24} (1938) 340--342.
  
  
\bibitem{milnor} 
  J. Milnor, \emph{Remarks on Infinite-Dimensional Lie Groups}. 
  In: B. DeWitt, R. Stora (eds.), Les Houches Session XL, \emph{Relativity, Groups and Topology II} (North-Holland, 1984), pp. 1007--1057.
  
\bibitem{moerr}
  I. Moerdijk, G. E. Reyes, \emph{Models for Smooth Infinitesimal Analysis}
  (Springer-Verlag, 1991).

\bibitem{mulsan}
  O. Müller, M. Sánchez, \emph{Lorentzian Manifolds Isometrically Embeddable in $\mathbb{L}^N$}.
  Trans. Amer. Math. Soc. \textbf{363} (2011) 5367--5379. 
  \texttt{arXiv:0812.4439 [math]}.

\bibitem{peetre}
  J. Peetre, \emph{Une Charactérisation Abstraite des Opérateurs Différentiels}.
  Math. Scand. \textbf{7} (1959) 211--218. Erratum: \emph{ibid.} \textbf{8} (1960) 116--120.

\bibitem{peierls}
  R. E. Peierls, \emph{The Commutation Laws of Relativistic Field Theory}. 
  Proc. Roy. Soc. London \textbf{A214} (1952) 143--157.
  
\bibitem{pietsch} 
  A. Pietsch, \emph{Nuclear Locally Convex Spaces} 
  (Springer-Verlag, 1972).

\bibitem{rao}
  M. M. Rao, \emph{Local Functionals}.
  In: D. Kolzow (ed.), \emph{Measure Theory, Oberwolfach 1979}. Lecture Notes in Mathematics \textbf{794} (Springer-Verlag, 1980), pp. 484--496.

\bibitem{rejzner}
  K. Rejzner, \emph{Fermionic Fields in the Functional Approach to Classical Field Theory}.
  Rev. Math, Phys. \textbf{23} (2011) 1009--1033.
  \texttt{arXiv:1101.5126 [math-ph]}.

\bibitem{seiler}
  W. M. Seiler, \emph{Involution: The Formal Theory of Differential Equations and its Applications in Computer Algebra}
  (Springer-Verlag, 2010).

\bibitem{slovak}
  J. Slovák, \emph{Peetre Theorem for Nonlinear Operators}.
  Ann. Global Anal. Geom. \textbf{6} (1988) 273--283.
  
\bibitem{sogge}
  C. D. Sogge, \emph{Lectures on Non-Linear Wave Equations. Second Edition}  
  (International Press, 2008).

\bibitem{stiefel}
  E. Stiefel, \emph{Richtungsfelder and Fernparallelismus in Mannigfaltigkeiten}.
  Comm. Math. Helv. \textbf{8} (1936) 3--51.

\bibitem{tso}
  K. Tso, \emph{Nonlinear Symmetric Positive Systems}.
  Ann. Inst. H. Poincaré Anal. Non Linéaire \textbf{9} (1992) 339--366.

\bibitem{vaisman}
  I. Vaisman, \emph{Lectures on the Geometry of Poisson Manifolds}
  (Birkhäuser, 1994).

\bibitem{vino1}
  A. M. Vinogradov, \emph{On the Algebro-Geometric Foundations of Lagrangian Field Theory}.
  Dokl. Akad. Nauk SSSR \textbf{236} (1977) 284--287; English translation in Sov. Math. Dokl. \textbf{18} (1977) 1200--1204.

\bibitem{vino2}
  A. M. Vinogradov, \emph{A Spectral Sequence Associated with a Nonlinear Differential Equation, and Algebro-Geometric Foundations of Lagrangian Field Theory with Constraints}.
  Dokl. Akad. Nauk SSSR \textbf{238} (1978) 1028--1031; English translation in Sov. Math. Dokl. \textbf{19} (1978) 144--148.

\bibitem{wald1}
  R. M. Wald, \emph{General Relativity}
  (Chicago University Press, 1984).

\bibitem{wald2}
  R. M. Wald, \emph{On Identically Closed Forms Locally Constructed from a Field}.
  J. Math. Phys. \textbf{31} (1990) 2378--2384.

\bibitem{weise}
  J.-C. Weise, \emph{On the Algebraic Formulation of Classical General Relativity}.
  Diplomarbeit, Universität Hamburg (2011). 
  \texttt{http://www.desy.de/uni-th/theses/Dipl\_Weise.pdf}

\bibitem{weyl}
  H. Weyl, \emph{Geodesic Fields in the Calculus of Variations for Multiple Integrals}. 
  Ann. Math. \textbf{36} (1935) 607--629.

\bibitem{zajtz}
  A. Zajtz, \emph{Nonlinear Peetre-like Theorems}.
  Univ. Iagel. Acta Math. \textbf{37} (1999) 351--361.
\end{thebibliography}
\end{document}